\newif\ifabstract
\abstracttrue
\newif\iffull
\ifabstract \fullfalse \else \fulltrue \fi

\documentclass[11pt]{article}
\usepackage{amsfonts}
\usepackage{amssymb}
\usepackage{amstext}
\usepackage{amsmath}
\usepackage{xspace}
\usepackage{ntheorem}
\usepackage{graphicx}
\usepackage{url}
\usepackage{graphics}
\usepackage{graphbox}
\usepackage{colordvi}
\usepackage{hyperref}
\usepackage{subfigure}

\advance \oddsidemargin by -0.8in
\newcommand{\myparskip}{3pt}
\parskip \myparskip

\newenvironment{proofof}[1]{\noindent{\bf Proof of #1.}}%
        {\hspace*{\fill}$\Box$\par\vspace{4mm}}

\newcommand{\badevent}{\mathcal{\event}_{\mathsf{bad}}}

\newcommand{\yi}{{\sc Yes-Instance}\xspace}
\renewcommand{\ni}{{\sc No-Instance}\xspace}
\newcommand{\YI}{{\sc Yes-Instance}\xspace}
\newcommand{\NI}{{\sc No-Instance}\xspace}

\newcommand{\yis}{{\sc Yes-Instances}\xspace}
\newcommand{\nis}{{\sc No-Instances}\xspace}

\newcommand{\NDPwC}{{\sf NDPwC}\xspace}
\newcommand{\EDPwC}{{\sf EDPwC}\xspace}

\newcommand{\NDP}{{\sf NDP}\xspace}
\newcommand{\EDP}{{\sf EDP}\xspace}
\newcommand{\EDPwall}{{\sf EDP-Wall}\xspace}

\newcommand{\NDPgrid}{{\sf NDP-Grid}\xspace}
\newcommand{\NDPwall}{{\sf NDP-Wall}\xspace}
\newcommand{\NDPplanar}{{\sf NDP-Planar}\xspace}
\newcommand{\WGP}{{\sf (r,h)-GP}\xspace}
\newcommand{\WGPwB}{{\sf (r,h)-GPwB}\xspace}

\newcommand{\DkS}{{\sf DkS}\xspace}

\newcommand{\cyi}{c_{\mbox{\textup{\tiny{YI}}}}}

\renewcommand{\H}{{\mathbf{H}}}
\newcommand{\block}{K}
\newcommand{\Y}{\Upsilon}

\newcommand{\ceil}[1]{\ensuremath{\left\lceil#1\right\rceil}}
\newcommand{\floor}[1]{\ensuremath{\left\lfloor#1\right\rfloor}}

\newcommand{\event}{{\cal{E}}}

\newcommand{\row}{\operatorname{row}}
\newcommand{\col}{\operatorname{col}}

\newcommand{\NP}{\mbox{\sf NP}}
\newcommand{\APX}{\mbox{\sf APX}}

\newcommand{\ZPTIME}{\mbox{\sf ZPTIME}}
\newcommand{\RTIME}{\mbox{\sf RTIME}}

\newcommand{\DTIME}{\mbox{\sf DTIME}}

\newcommand{\opt}{\mathsf{OPT}}

\newcommand{\set}[1]{\left\{ #1 \right\}}
\newcommand{\defAsSet}[2]{\set{ #1 \> | \> #2}}

\newcommand{\tset}{{\mathcal T}}

\newcommand{\iset}{{\mathcal{I}}}
\newcommand{\hiset}{\hat{\mathcal{I}}}
\newcommand{\hG}{\hat G}
\newcommand{\pset}{{\mathcal{P}}}
\newcommand{\qset}{{\mathcal{Q}}}
\newcommand{\lset}{{\mathcal{L}}}
\newcommand{\dset}{{\mathcal{D}}}
\newcommand{\bset}{{\mathcal{B}}}
\newcommand{\aset}{{\mathcal{A}}}
\newcommand{\cset}{{\mathcal{C}}}
\newcommand{\fset}{{\mathcal{F}}}
\newcommand{\mset}{{\mathcal M}}
\newcommand{\tG}{\tilde G}

\newcommand{\gset}{{\mathcal G}}

\newcommand{\wset}{{\mathcal{W}}}
\newcommand{\nset}{{\mathcal N}}
\newcommand{\uset}{{\mathcal U}}
\newcommand{\yset}{{\mathcal{Y}}}
\newcommand{\rset}{{\mathcal{R}}}

\newcommand{\hset}{{\mathcal{H}}}
\newcommand{\sset}{{\mathcal{S}}}

\newcommand{\be}{\begin{enumerate}}
\newcommand{\ee}{\end{enumerate}}
\newcommand{\bd}{\begin{description}}
\newcommand{\ed}{\end{description}}
\newcommand{\bi}{\begin{itemize}}
\newcommand{\ei}{\end{itemize}}

\newtheorem{theorem}{Theorem}[section]
\newtheorem{lemma}[theorem]{Lemma}
\newtheorem{observation}[theorem]{Observation}
\newtheorem{corollary}[theorem]{Corollary}
\newtheorem{claim}[theorem]{Claim}

\newtheorem{definition}{Definition}
\newenvironment{proof}{\par \smallskip{\bf Proof:}}{\hfill\stopproof}
\def\stopproof{\square}
\def\square{\vbox{\hrule height.2pt\hbox{\vrule width.2pt height5pt \kern5pt
\vrule width.2pt} \hrule height.2pt}}

\newenvironment{prog}[1]{
\begin{minipage}{5.8 in}
\begin{center}
{\sc #1}
\end{center}
}
{
\end{minipage}}

\newcommand{\program}[2]{\fbox{\vspace{2mm}\begin{prog}{#1} #2 \end{prog}\vspace{2mm}}}

\renewcommand{\phi}{\varphi}
\newcommand{\eps}{\epsilon}

\newcommand{\half}{\ensuremath{\frac{1}{2}}}

\newcommand{\poly}{\operatorname{poly}}

\newcommand{\expect}[2][]{\text{\bf E}_{#1}\left [#2\right]}
\newcommand{\prob}[2][]{\text{\bf Pr}_{#1}\left [#2\right]}

\renewcommand{\xi}{\ell}

\newcommand{\cro}{\operatorname{cr}}

\newcommand{\constantForSizeOfBlocks}{1024}
\newcommand{\twiceConstantForSizeOfBlocks}{2048}

\begin{document}

\title{Almost Polynomial Hardness of Node-Disjoint Paths in Grids}
\author{Julia Chuzhoy\thanks{Toyota Technological Institute at Chicago. Email: {\tt cjulia@ttic.edu}. Supported in part by NSF grants CCF-1318242 and CCF-1616584.}\and David H. K. Kim\thanks{Computer Science Department, University of Chicago. Email: {\tt hongk@cs.uchicago.edu}. Supported in part by NSF grant CCF-1318242 and CCF-1616584.} \and Rachit Nimavat\thanks{Toyota Technological Institute at Chicago. Email: {\tt nimavat@ttic.edu}. Supported in part by NSF grant CCF-1318242.}}

\begin{titlepage}
\maketitle

\thispagestyle{empty}

\begin{abstract}
In the classical Node-Disjoint Paths (\NDP) problem,  we are given an $n$-vertex graph $G=(V,E)$, and a collection $\mset=\set{(s_1,t_1),\ldots,(s_k,t_k)}$ of pairs of its vertices, called source-destination, or demand pairs. The goal is to route as many of the demand pairs as possible, where to route a pair we need to select a path connecting it, so that all selected paths are disjoint in their vertices.
The best current algorithm for \NDP achieves an $O(\sqrt{n})$-approximation, while, until recently, the best negative result was a factor $\Omega(\log^{1/2-\eps}n)$-hardness of approximation, for any constant $\eps$, unless $\NP \subseteq \ZPTIME(n^{\poly \log n})$. In a recent work, the authors have shown an improved $2^{\Omega(\sqrt{\log n})}$-hardness of approximation for \NDP, unless $\NP\subseteq \DTIME(n^{O(\log n)})$, even if the underlying graph is a {\bf subgraph} of a grid graph, and all source vertices lie on the boundary of the grid. Unfortunately, this result does not extend to grid graphs.

The approximability of the \NDP problem on grid graphs has remained a tantalizing open question, with the best current upper bound of $\tilde{O}(n^{1/4})$, and the best current lower bound of \APX-hardness. In a recent work, the authors showed a $2^{\tilde{O}(\sqrt{\log n})}$-approximation algorithm for \NDP in grid graphs, if all source vertices lie on the boundary of the grid -- a result that can be seen as suggesting that a sub-polynomial approximation may be achievable for \NDP in grids. 
In this paper we show that this is unlikely to be the case, and come close to resolving the approximability of \NDP in general, and \NDP in grids in particular. Our main result is that \NDP is $2^{\Omega(\log^{1-\eps} n)}$-hard to approximate for any constant $\eps$, assuming that $\NP\nsubseteq\RTIME(n^{\poly\log n})$, and that it is $n^{\Omega (1/(\log \log n)^2)}$-hard to approximate, assuming that for some constant $\delta>0$, $\NP \not \subseteq \RTIME(2^{n^{\delta}})$. These results hold even for grid graphs and wall graphs, and extend to the closely related Edge-Disjoint Paths problem, even in wall graphs.

Our hardness proof performs a reduction from the 3COL(5) problem to \NDP, using a new graph partitioning problem as a proxy. Unlike the more standard approach of employing Karp reductions to prove hardness of approximation, our proof  is a Cook-type reduction, where, given an input instance of 3COL(5), we produce a large number of instances of \NDP, and apply an approximation algorithm for \NDP to each of them. The construction of each new instance of \NDP crucially depends on  the solutions to the previous instances that were found by the approximation algorithm.
\end{abstract}

\end{titlepage}

\label{--------------------------------------sec: intro---------------------------------------}
\section{Introduction}\label{sec: intro}

We study the Node-Disjoint Paths (\NDP) problem: given an undirected $n$-vertex graph $G$ and a collection $\mset=\set{(s_1,t_1),\ldots,(s_k,t_k)}$ of pairs of  its vertices, called \emph{source-destination}, or \emph{demand} pairs,  we are interested in routing the demand pairs, where in order to route a pair $(s_i,t_i)$, we need to select a path connecting $s_i$ to $t_i$. The goal is to route as many of the pairs as possible, subject to the constraint that the selected routing paths are mutually disjoint in their vertices and their edges. We let $S=\set{s_1,\ldots,s_k}$ be the set of the source vertices, $T=\set{t_1,\ldots,t_k}$ the set of the destination vertices, and we refer to the vertices of $S\cup T$ as \emph{terminals}. We denote by \NDPplanar the special case of the problem where the graph $G$ is planar; by \NDPgrid the special case where $G$ is a square grid; and by \NDPwall the special case where $G$ is a wall (see Figure~\ref{fig: wall} for an illustration of a wall and Section~\ref{sec: prelims} for its formal definition).

\NDP is a fundamental problem in the area of graph routing, that has been studied extensively. 
Robertson and Seymour~\cite{RobertsonS,flat-wall-RS} showed, as part of their famous Graph Minors Series, an efficient algorithm for solving the problem if the number $k$ of the demand pairs is bounded by a constant. However, when $k$ is a part of the input, the problem becomes \NP-hard~\cite{Karp-NDP-hardness,EDP-hardness}, and  it remains \NP-hard even for planar graphs~\cite{npc_planar}, and for grid graphs~\cite{npc_grid}. 
The best current upper bound on the approximability of \NDP is $O(\sqrt{n})$, obtained by a simple greedy algorithm~\cite{KolliopoulosS}. Until recently, the best known lower bound was an $\Omega(\log^{1/2-\eps}n)$-hardness of approximation for any constant $\eps$, unless $\NP\subseteq\ZPTIME(n^{\poly\log n})$~\cite{AZ-undir-EDP,ACGKTZ}, and \APX-hardness for the special cases of  \NDPplanar and \NDPgrid~\cite{NDP-grids}. In a recent paper~\cite{NDPhardness2017}, the authors have shown an improved $2^{\Omega(\sqrt{\log n})}$-hardness of approximation for \NDP, assuming that $\NP\nsubseteq \DTIME(n^{O(\log n)})$. This result holds even for planar graphs with maximum vertex degree $3$, where all source vertices lie on the boundary of a single face, and for {\bf sub-graphs} of grid graphs, with all source vertices lying on the boundary of the grid. We note that for general planar graphs, the $O(\sqrt{n})$-approximation algorithm of~\cite{KolliopoulosS} was recently slightly improved to an $\tilde O(n^{9/19})$-approximation~\cite{NDP-planar}.

The approximability status of \NDPgrid --- the special case of \NDP where the underlying graph is a square grid --- remained a tantalizing open question. The study of this problem dates back to the 70's, and was initially motivated by applications to VLSI design. As grid graphs are extremely well-structured, one would expect that good approximation algorithms can be designed for them, or that, at the very least, they should be easy to understand. However, establishing the approximability of \NDPgrid has been elusive so far.
The simple greedy  $O(\sqrt{n})$-approximation algorithm of~\cite{KolliopoulosS} was only recently improved to a $\tilde O(n^{1/4})$-approximation for \NDPgrid~\cite{NDP-grids}, while on the negative side only \APX-hardness is known. In a very recent paper~\cite{NDPSourcesOnTop}, the authors designed a $2^{O(\sqrt{\log n}\cdot\log\log n})$-approximation algorithm for a special case of \NDPgrid, where the source vertices appear on the grid boundary. This result can be seen as complementing the $2^{\Omega(\sqrt{\log n})}$-hardness of approximation of \NDP on sub-graphs of grids with all sources lying on the grid boundary~\cite{NDPhardness2017}\footnote{Note that the results are not strictly complementary: the algorithm only works on grid graphs, while the hardness result is only valid for sub-graphs of grids.}. Furthermore, this result can be seen as suggesting that sub-polynomial approximation algorithms may be achievable for \NDPgrid.

In this paper we show that this is unlikely to be the case, and come close to resolving the approximability status of \NDPgrid, and of \NDP in general, by showing that \NDPgrid  is $2^{\Omega(\log^{1-\eps} n)}$-hard to approximate for any constant $\eps$, unless $\NP\subseteq\RTIME(n^{\poly\log n})$. We further show that it is $n^{\Omega(1/(\log\log n)^2)}$-hard to approximate, assuming that for some constant $\delta>0$, $\NP \not \subseteq \RTIME(2^{n^{\delta}})$. The same hardness results also extend to \NDPwall. These hardness results are stronger than the best currently known hardness for the general \NDP problem, and should be contrasted with the $2^{O(\sqrt{\log n}\cdot\log\log n})$-approximation algorithm for \NDPgrid with all sources lying on the grid boundary~\cite{NDPSourcesOnTop}.

Another basic routing problem that is closely related to \NDP is Edge-Disjoint Paths (\EDP). The input to this problem is the same as before: an undirected graph $G=(V,E)$ and a set $\mset=\set{(s_1,t_1),\ldots,(s_k,t_k)}$ of demand pairs. The goal, as before, is to route the largest number of the demand pairs via paths. However, we now allow the paths to share vertices, and only require that they are mutually edge-disjoint. In general, it is easy to see that  \EDP is a special case of \NDP. Indeed, given an \EDP instance $(G,\mset)$, computing the line graph of the input graph $G$ transforms it into an equivalent instance of \NDP. However, this transformation may inflate the number of the graph vertices, and so approximation factors that depend on $|V(G)|$ may no longer be preserved. Moreover, this transformation does not preserve planarity, and no such relationship is known between \NDP and \EDP in planar graphs. The approximability status of \EDP is very similar to that of \NDP: the best current approximation algorithm achieves an $O(\sqrt{n})$-approximation factor~\cite{EDP-alg}, and the recent $2^{\Omega(\sqrt{\log n})}$-hardness of approximation of~\cite{NDPhardness2017}, under the assumption that $\NP\not\subseteq \DTIME(n^{O(\log n)})$, extends to \EDP. Interestingly, \EDP appears to be relatively easy on grid graphs, and has a constant-factor approximation for this special case~\cite{grids1, grids3,grids4}. 
The analogue of the grid graph in the setting of \EDP seems to be the wall graph (see Figure~\ref{fig: wall}): the approximability status of \EDP on wall graphs is similar to that of \NDP on grid graphs, with the best current upper bound of $\tilde O(n^{1/4})$, and the best lower bound of APX-hardness~\cite{NDP-grids}. 
The results of~\cite{NDPhardness2017} extend to a $2^{\Omega(\sqrt{\log n})}$-hardness of approximation for \EDP on sub-graphs of wall graphs, with all source vertices lying on the wall boundary, under the same complexity assumption. We denote by \EDPwall the special case of the \EDP problem where the underlying graph is a wall. We show that our new almost polynomial hardness of approximation results also hold for \EDPwall and for \NDPwall.

\begin{figure}[h]
\centering
\scalebox{0.4}{\includegraphics{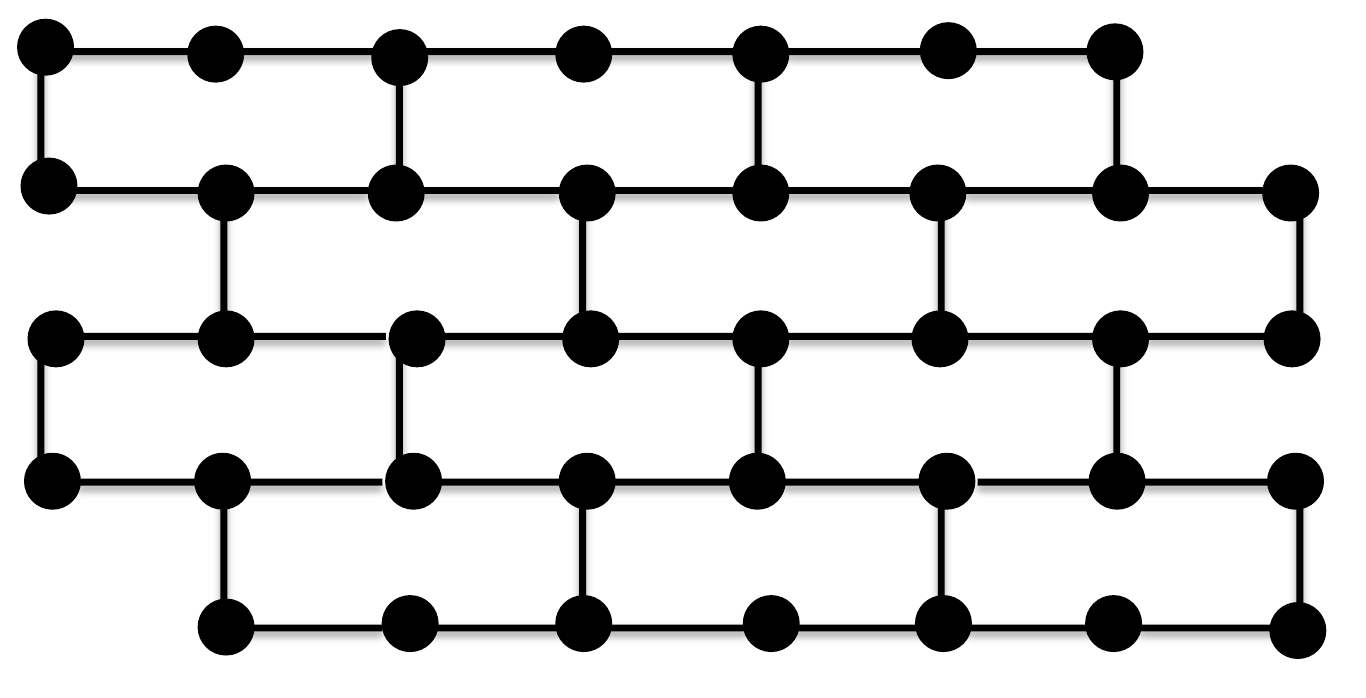}}
\caption{A wall graph.\label{fig: wall}}
\end{figure}

\paragraph*{Other Related Work.} 
Several other special cases of \EDP are known to have reasonably good approximation algorithms. For example, for the special case of Eulerian planar graphs, Kleinberg~\cite{Kleinberg-planar} showed an $O(\log^2 n)$-approximation algorithm, while Kawarabayashi and Kobayashi~\cite{KK-planar} provide an improved $O(\log n)$-approximation for both Eulerian and $4$-connected planar graphs. Polylogarithmic approximation algorithms are also known for bounded-degree expander graphs~\cite{LeightonRao99,BFU92,BroderFSU94,KleinbergRubinfeld96,Frieze2000}, and constant-factor approximation algorithms are known for trees~\cite{apx_tree,ChekuriMS07}, and grids and grid-like graphs~\cite{grids1,AwerbuchGLR94,grids3,grids4}. Rao and Zhou~\cite{RaoZhou} showed an efficient randomized $O(\poly\log n)$-approximation  algorithm for the special case of \EDP where the value of the global minimum cut in the input graph is $\Omega(\log^5 n)$.
Recently, Fleszar et al.~\cite{fleszar_et_al} designed an  $O(\sqrt{r}\cdot \log^{1.5} (kr))$-approximation algorithm for \EDP, where $r$ is the feedback vertex set number of the input graph $G=(V,E)$ --- the smallest number of vertices that need to be deleted from $G$ in order to turn it into a forest. 

A natural variation of \NDP and \EDP that relaxes the disjointness constraint by allowing a small vertex- or edge-congestion has been a subject of extensive study. In the \NDP with Congestion (\NDPwC) problem, the input consists of an undirected graph and a set of demand pairs as before, and additionally a non-negative integer $c$. The goal is to route a maximum number of the demand pairs with congestion $c$, that is, each vertex may participate in at most $c$ paths in the solution. The \EDP with Congestion problem (\EDPwC) is defined similarly, except that now the congestion is measured on the graph edges and not vertices. The famous result of Raghavan and Thompson~\cite{RaghavanT}, that introduced the randomized LP-rounding technique, obtained a constant-factor approximation for \NDPwC and \EDPwC, for a congestion value $c=\Theta(\log n/\log\log n)$. A long sequence of work~\cite{CKS,Raecke,Andrews,RaoZhou,Chuzhoy11,ChuzhoyL12,ChekuriE13,NDPwC2} has led to an $O(\poly\log k)$-approximation for \NDPwC and \EDPwC with congestion bound $c=2$. This result is essentially optimal, since it is known that for every constant $\eps$, and for  every congestion value $c=o(\log\log n/\log\log\log n)$, both problems are hard to approximate to within a factor $\Omega((\log n)^{\frac{1-\eps}{c+1}})$, unless $\NP \subseteq \ZPTIME(n^{\poly \log n})$~\cite{ACGKTZ}. When the input graph is planar, Seguin-Charbonneau and Shepherd~\cite{EDP-planar-c2}, improving on the result of Chekuri, Khanna and Shepherd~\cite{EDP-planar-constant-cong}, have shown a constant-factor approximation for \EDPwC with congestion 2.

\paragraph*{Our Results and Techniques.} Our main result is the proof of the following two theorems.

\begin{theorem} \label{thm: master NDP}
For every constant $\eps > 0$, there is no $2^{O(\log^{1-\eps} n)}$-approximation algorithm for \NDP, assuming that $\NP\nsubseteq\RTIME(n^{\poly\log n})$. Moreover, there is no $n^{O(1/(\log\log n)^2)}$-approximation algorithm for \NDP, assuming that  for some constant $\delta>0$, $\NP\nsubseteq \RTIME(2^{n^\delta})$. These results hold even when the input graph is a grid graph or a wall graph.
\end{theorem}

\begin{theorem} \label{thm: master EDP}
For every constant $\eps > 0$, there is no $2^{O(\log^{1-\eps} n)}$-approximation algorithm for \EDP, assuming that $\NP\nsubseteq\RTIME(n^{\poly\log n})$. Moreover, there is no $n^{O(1/(\log\log n)^2)}$-approximation algorithm for \EDP, assuming that  for some constant $\delta>0$, $\NP\nsubseteq \RTIME(2^{n^\delta})$. These results hold even when the input graph is a wall graph.
\end{theorem}

We now provide a high-level overview of our techniques.
The starting point of our hardness of approximation proof is 3COL(5) --- a special case of the $3$-coloring problem, where the underlying graph is $5$-regular. We define a new graph partitioning problem, that we refer to as $(r,h)$-Graph Partitioning, and denote by \WGP. In this problem, we are given a bipartite graph $\tilde G = (V_1,V_2,E)$ and two integral parameters $r,h > 0$. A solution to the problem is a partition $(W_1,W_2,\ldots,W_r)$ of $V_1\cup V_2$ into $r$ subsets, and for each $1\le i\le r$, a subset $E_i\subseteq E(W_i)$ of edges, so that $|E_i|\le h$ holds, and the goal is to maximize $\sum_{i=1}^r |E_i|$. A convenient intuitive way to think about this problem is that we would like to partition $\tilde G$ into a large number of subgraphs, in a roughly balanced way (with respect to the number of edges), so as to preserve as many of the edges as possible.
We show that \NDPgrid is at least as hard as the \WGP problem (to within polylogarithmic factors). Our reduction exploits the fact that routing in grids is closely related to graph drawing, and that graphs with small crossing number have small balanced separators.  The \WGP problem itself appears similar in flavor to the Densest $k$-Subgraph problem (\DkS). In the \DkS problem, we are given a graph $G=(V,E)$ and a parameter $k$, and the goal is to find a subset $U\subseteq V$ of $k$ vertices, that maximizes the number of edges in the induced graph $G[U]$. Intuitively, in the \WGP problem, the goal is to partition the graph into many dense subgraphs, and so in order to prove that \WGP is hard to approximate, it is natural to employ techniques used in hardness of approximation proofs for \DkS. 
The best current approximation algorithm for \DkS achieves a $n^{1/4+\eps}$-approximation for any constant $\eps$~\cite{dks10}. Even though the problem appears to be very hard to approximate, its hardness of approximation proof has been elusive until recently: only constant-factor hardness results were known for \DkS under various worst-case complexity assumptions, and $2^{\Omega(\log^{2/3}n)}$-hardness under average-case assumptions~\cite{Feige02,dks_average_hardness,Khot04,RaghavendraSteurer10}. 
In a recent breakthrough, Manurangsi~\cite{Manurangsi16} has shown that for some constant $c$, \DkS is hard to approximate to within a factor $n^{1/(\log\log n)^c}$,  under the Exponential Time Hypothesis. Despite our feeling that \WGP is somewhat similar to \DkS, we were unable to extend the techniques of~\cite{Manurangsi16} to this problem, or to prove its hardness of approximation via other techniques.

We overcome this difficulty as follows. First, we define a graph partitioning problem that is slightly more general than \WGP. The definition of this problem is somewhat technical and is deferred to Section~\ref{secL WGP}. This problem is specifically designed so that the reduction to \NDPgrid still goes through, but it is somewhat easier to control its solutions and prove hardness of approximation for it. Furthermore, instead of employing a standard Karp-type reduction (where an instance of 3COL(5) is reduced to a single instance of \NDPgrid, while using the graph partitioning problem as a proxy), we employ a sequential Cook-type reduction. We assume for contradiction that an $\alpha$-approximation algorithm $\aset$ for \NDPgrid exists, where $\alpha$ is the hardness of approximation factor we are trying to prove. Our reduction is iterative. In every iteration $j$, we reduce the 3COL(5) instance to a collection $\iset_j$ of instances of \NDPgrid, and apply the algorithm $\aset$ to each of them. If the 3COL(5) instance is a \YI, then we are guaranteed that each resulting instance of \NDPgrid has a large solution, and so all solutions returned by $\aset$ are large. If the 3COL(5) instance is a \NI, then unfortunately it is still possible that the resulting instances of \NDPgrid will have large solutions. However, we can use these resulting solutions in order to further refine our reduction, and construct a new collection $\iset_{j+1}$ of instances of \NDPgrid. While in the \YI case we will continue to obtain large solutions to all \NDPgrid instances that we construct, we can show that in the \NI case, in some iteration of the algorithm, we will fail to find such a large solution. Our reduction is crucially  sequential, and we exploit the solutions returned by algorithm $\aset$ in previous iterations in order to construct new instances of \NDPgrid for the subsequent iterations.
It is interesting whether these techniques may be helpful in obtaining new hardness of approximation results for \DkS.

We note that our approach is completely different from the previous hardness of approximation proof of~\cite{NDPhardness2017}. The proof in~\cite{NDPhardness2017} proceeded by performing a reduction from 3SAT(5). Initially, a simple reduction from 3SAT(5) to the \NDP problem on a sub-graph of the grid graph is used in order to produce a constant hardness gap. The resulting instance of \NDP is called a level-$1$ instance. The reduction then employs a boosting technique, that, in order to obtain a level-$i$ instance, combines a number of level-$(i-1)$ instances with a single level-$1$ instance. The hardness gap grows by a constant factor from iteration to iteration, until the desired hardness of approximation bound is achieved. All source vertices in the constructed instances appear on the grid boundary, and a large number of vertices are carefully removed from the grid in order to create obstructions to routing, and to force the routing paths to behave in a prescribed way.
The reduction itself is a Karp-type reduction, and eventually produces a single instance of \NDP with a large gap between the \YI and \NI solutions.

\paragraph{Organization.}
We start with Preliminaries in Section~\ref{sec: prelims}, and introduce the new graph partitioning problems in Section~\ref{secL WGP}. The hardness of approximation proof for \NDPgrid appears in Section~\ref{sec: the hardness proof}, with the reduction from the graph partitioning problem to \NDPgrid deferred to Section~\ref{sec: from WGP to NDP}. Finally, we extend our hardness results to \NDPwall and \EDPwall in Section~\ref{sec: from NDP to EDP}.

\label{--------------------------------------sec: prelims---------------------------------------}
\section{Preliminaries}\label{sec: prelims}

We use standard graph-theoretic notation. Given a graph $G$ and a subset $W\subseteq V(G)$ of its vertices, $E(W)$ denotes the set of all edges of $G$ that have both their endpoints in $W$. Given a path $P$ and a subset $U$ of vertices of $G$, we say that $P$ is \emph{internally disjoint} from $U$ iff every vertex in $P\cap U$ is an endpoint of $P$. Similarly, $P$ is internally disjoint from a subgraph $G'$ of $G$ iff $P$ is internally disjoint from $V(G')$. Given a subset $\mset'\subseteq\mset$ of the demand pairs in $G$, we denote by $S(\mset')$ and $T(\mset')$ the sets of the source and the destination vertices of the demand pairs in $\mset'$, respectively. We let $\tset(\mset')=S(\mset')\cup T(\mset')$ denote the set of all terminals participating as a source or a destination in $\mset'$. All logarithms in this paper are to the base of 2.

\paragraph{Grid Graphs.}
For a pair  $h,\ell> 0$ of integers, we let $G^{h,\ell}$ denote the grid of height $h$ and length $\ell$.
The set of its vertices is $V(G^{h,\ell})=\set{v(i,j)\mid 1\leq i\leq h, 1\leq j\leq \ell}$, and the set of its edges is the union of two subsets: the set $E^H=\set{(v_{i,j},v_{i,j+1})\mid 1\leq i\leq h, 1\leq j<\ell}$ of horizontal edges  and the set $E^{V}=\set{(v_{i,j},v_{i+1,j})\mid 1\leq i< h, 1\leq j\leq\ell}$ of vertical edges. The subgraph of $G^{h,\ell}$ induced by the edges of $E^H$ consists of $h$ paths, that we call the \emph{rows} of the grid; for $1\leq i\leq h$, the $i$th row $R_i$ is the row containing the vertex $v(i,1)$. Similarly, the subgraph induced by the edges of $E^V$ consists of $\ell$ paths that we call the \emph{columns} of the grid, and for $1\leq j\leq \ell$, the $j$th column $W_j$ is the column containing $v(1,j)$. We think of the rows  as ordered from top to bottom and the columns as ordered from left to right. Given a vertex $v=v(i,j)$ of the grid, we denote by $\row(v)$ and $\col(v)$ the row and the column of the grid, respectively, that contain $v$. We say that $G^{h,\ell}$ is a \emph{square grid} iff $h=\ell$. The \emph{boundary of the grid} is $R_1\cup R_h\cup W_1\cup W_{\ell}$. We sometimes refer to $R_1$ and $R_h$ as the top and the bottom boundary edges of the grid respectively, and to $W_1$ and $W_{\ell}$ as the left and the right boundary edges of the grid. 

Given a subset $\rset'$ of consecutive rows of $G$ and a subset $\wset'$ of consecutive columns of $G$, the \emph{sub-grid of $G$ spanned by the rows in $\rset'$ and the columns in $\wset'$} is the sub-graph of $G$ induced by the set $\set{v \mid \row(v)\in \rset', \col(v)\in \wset'}$ of its vertices.

Given two vertices $u=v(i,j)$ and $u'=v(i',j')$ of a grid $G$, the shortest-path distance between them is denoted by $d(u,u')$. Given two vertex subsets $X,Y\subseteq V(G)$, the distance between them is $d(X,Y)=\min_{u\in X,u'\in Y}\set{d(u,u')}$. When $H, H'$ are subgraphs of $G$, we use $d(H,H')$ to denote $d(V(H),V(H'))$.

\paragraph{Wall Graphs.}
Let $G=G^{\ell,h}$ be a grid of length $\ell$ and height $h$. Assume that $\ell>0$ is an even integer, and that $h>0$.
For every column $W_j$ of the grid, let $e^j_1,\ldots,e^j_{h-1}$ be the edges of $W_j$ indexed in their top-to-bottom order. Let $E^*(G)\subseteq E(G)$ contain all edges $e^j_z$, where $z\neq j \mod 2$, and let $\hat G$ be the graph obtained from $G\setminus E^*(G)$, by deleting all degree-$1$ vertices from it.
Graph $\hat G$ is called a \emph{wall of length $\ell/2$ and height $h$} (see Figure~\ref{fig: wall}).
Consider the subgraph of $\hat G$ induced by all horizontal edges of the grid $G$ that belong to $\hat G$. This graph is a collection of $h$ node-disjoint paths, that we refer to as the \emph{rows} of $\hat G$, and denote them by $R_1,\ldots,R_h$ in this top-to-bottom order; notice that $R_j$ is also the $j$th row of the grid $G$ for all $j$. Graph $\hat G$ contains a unique collection $\wset$ of $\ell/2$ node-disjoint paths that connect vertices of $R_1$ to vertices of $R_h$ and are internally disjoint from $R_1$ and $R_h$. We refer to the paths in $\wset$ as the \emph{columns} of $\hat G$, and denote them by $W_1,\ldots,W_{\ell/2}$ in this left-to-right order. Paths $W_1, W_{\ell/2},R_1$ and $R_h$ are called the left, right, top and bottom boundary edges of $\hat G$, respectively, and their union is the boundary of $\hat G$.

\paragraph{The 3COL(5) problem.}
The starting point of our reduction is the 3COL(5) problem. In this problem, we are given a $5$-regular graph $G=(V,E)$. Note that, if $n=|V|$ and $m=|E|$, then $m=5n/2$.  We are also given a set $\cset=\set{r,b,g}$ of $3$ colors. A \emph{coloring} $\chi: V\rightarrow \cset$ is an assignment of a color in $\cset$ to every vertex  in $V$. 
We say that an edge $e=(u,v)$ is \emph{satisfied} by the coloring $\chi$ iff $\chi(u)\neq \chi(v)$. 
The coloring $\chi$ is \emph{valid} iff it satisfies every edge. We say that $G$ is a \yi iff there is a valid coloring $\chi: V\rightarrow \cset$. We say that it is a \ni with respect to some given parameter $\eps$, iff for every coloring $\chi: V\rightarrow \cset$, at most a $(1-\eps)$-fraction of the edges are satisfied by $\chi$. We use the following theorem of Feige et al.~\cite{Feige3COL5}:

\begin{theorem} \label{thm: GAPCOL is hard}[Proposition 15  in \cite{Feige3COL5}]
There is some constant $\eps$, such that distinguishing between the \yis and the  \nis (with respect to $\eps$) of 3COL(5) is NP-hard.
\end{theorem}


\paragraph{A Two-Prover Protocol.}
We use the following two-prover protocol. The two provers are called an edge-prover and a vertex-prover. Given a $5$-regular graph $G$, the verifier selects an edge $e=(u,v)$ of $G$ uniformly at random, and then selects a random endpoint (say $v$) of this edge. It then sends $e$ to the edge-prover and $v$ to the vertex-prover. The edge-prover must return an assignment of colors from $\cset$ to $u$ and $v$, such that the two colors are distinct; the vertex-prover must return an assignment of a color from $\cset$ to $v$. The verifier accepts iff both provers assign the same color to $v$. Given a $2$-prover game $\gset$, its \emph{value} is the maximum acceptance probability of the verifier over all possible strategies of the provers.  

Notice that, if $G$ is a \yi, then there is a strategy for both provers that guarantees acceptance with probability $1$: the provers fix a valid coloring $\chi: V\rightarrow \cset$ of $G$ and respond to the queries according to this coloring. 

We claim that if $G$ is a \ni, then for any strategy of the two provers, the verifier accepts with probability at most $(1-\eps/2)$. Note first that we can assume without loss of generality that the strategies of the provers are deterministic. Indeed,  if the provers have a probability distribution over the answers to each query $Q$, then the edge-prover, given a query $Q'$, can return an answer that maximizes the acceptance probability of the verifier under the random strategy of the vertex-prover. This defines a deterministic strategy for the edge-prover that does not decrease the acceptance probability of the verifier. The vertex-prover in turn, given any query $Q$, can return an answer that maximizes the acceptance probability of the verifier, under the new deterministic strategy of the edge-prover. The acceptance probability of this final deterministic strategy of the two provers is at least as high as that of the original randomized strategy. The deterministic strategy of  the vertex-prover defines a coloring of the vertices of $G$. This coloring must dissatisfy at least $\eps m$ edges. The probability that the verifier chooses one of these edge is at least $\eps$. The response of the edge-prover on such an edge must differ from the response of the vertex-prover on at least one endpoint of the edge. The verifier chooses this endpoint with probability at least $\half$, and so overall the verifier rejects with probability at least $\eps/2$. 
Therefore, if $G$ is a \yi, then the value of the corresponding game is $1$, and if it is a \ni, then the value of the game is at most $(1-\eps/2)$.

\paragraph{Parallel Repetition.}
We perform $\ell$ rounds of parallel repetition of the above protocol, for some integer $\ell>0$, that may depend on $n=|V(G)|$. Specifically, the verifier chooses a sequence $(e_1,\ldots,e_{\ell})$ of $\ell$ edges, where each edge $e_i$ is selected independently uniformly at random from $E(G)$. For each chosen edge $e_i$, one of its endpoints $v_i$ is then chosen independently at random. The verifier sends $(e_1,\ldots,e_{\ell})$ to the edge-prover, and $(v_1,\ldots,v_{\ell})$ to the vertex-prover. The edge-prover returns a coloring of both endpoints of each edge $e_i$. This coloring must satisfy the edge (so the two endpoints must be assigned different colors), but it need not be consistent across different edges. In other words, if two edges $e_i$ and $e_j$ share the same endpoint $v$, the edge-prover may assign different colors to each occurrence of $v$. The vertex-prover returns a coloring of the vertices in $(v_1,\ldots,v_{\ell})$. Again, if some vertex $v_i$ repeats twice, the coloring of the two occurrences need not be consistent. The verifier accepts iff for each $1\leq i\leq \ell$, the coloring of the vertex $v_i$ returned by both provers is consistent. (No verification of consistency is performed across different $i$'s. So, for example, if $v_i$ is an endpoint of $e_j$ for $i\neq j$, then it is possible that the two colorings do not agree and the verifier still accepts). We say that a pair $(A,A')$ of answers to the two queries $(e_1,\ldots,e_{\ell})$ and $(v_1,\ldots,v_{\ell})$  is \emph{matching}, or \emph{consistent}, iff it causes the verifier to accept. We let $\gset^\ell$ denote this 2-prover protocol with $\ell$ repetitions.

\begin{theorem}[Parallel Repetition]\cite{RazParallelRep,HolensteinParallelRep,RaoParallelRep}
There is some constant $0<\gamma'<1$, such that for each $2$-prover game $\tilde \gset$, if the value of $\tilde \gset$ is $x$, then the value of the game $\tilde\gset^\ell$, obtained from $\ell>0$ parallel repetitions of $\tilde\gset$, is $x^{\gamma'\ell}$. 
\end{theorem}

\begin{corollary}\label{cor: parallel repetition}
There is some constant $0<\gamma<1$, such that, if $G$ is a \yi, then $\gset^\ell$ has value $1$, and if $G$ is a \ni, then $\gset^\ell$ has value at most  $2^{-\gamma\ell}$.
\end{corollary}

We now summarize the parameters and introduce some basic notation:

\begin{itemize}
\item Let $\qset^E$ denote the set of all possible queries to the edge-prover, so each query is an $\ell$-tuple of edges. Then $|\qset^E|=m^{\ell}=(5n/2)^{\ell}$. Each query has $6^{\ell}$ possible answers -- $6$ colorings per edge. The set of feasible answers is the same for each edge-query, and we denote it by $\aset^E$.

\item Let $\qset^V$ denote the set of all possible queries to the vertex-prover, so each query is an $\ell$-tuple of vertices. Then $|\qset^V|=n^{\ell}$. Each query has $3^{\ell}$ feasible answers -- $3$ colorings per vertex. The set of feasible answers is the same for each vertex-query, and we denote it by $\aset^V$.

\item We think about the verifier as choosing a number of random bits, that determine the choices of the queries $Q\in \qset^E$ and $Q'\in \qset^V$ that it sends to the provers. We sometimes call each such random choice a ``random string''. The set of all such choices is denoted by $\rset$, where for each $R\in \rset$, we denote $R=(Q,Q')$, with $Q\in \qset^E$, $Q'\in \qset^V$ --- the two queries sent to the two provers when the verifier chooses $R$. Then $|\rset|=(2m)^{\ell}=(5n)^{\ell}$, and each random string $R\in \rset$ is chosen with the same probability.

\item It is important to note that each query $Q=(e_1,\ldots,e_{\ell})\in \qset^E$ of the edge-prover participates in exactly $2^{\ell}$ random strings (one random string for each choice of one endpoint per edge of $\set{e_1,\ldots,e_{\ell}}$), while each query $Q'=(v_1,\ldots,v_{\ell})$ of the vertex-prover participates in exactly $5^{\ell}$ random strings (one random string for each choice of an edge incident to each of $v_1,\ldots,v_{\ell}$).
\end{itemize}

A function $f: \qset^E\cup \qset^V\rightarrow \aset^E\cup \aset^V$ is called a \emph{global assignment of answers to queries} iff for every query $Q\in \qset^E$ to the edge-prover, $f(Q)\in \aset^E$, and for every query $Q'\in\qset^V$ to the vertex-prover, $f(Q')\in \aset^V$.  We say that $f$ is a \emph{perfect global assignment} iff for every random string $R=(Q^E,Q^V)$, $(f(Q^E),f(Q^V))$ is a matching pair of answers. The following simple theorem, whose proof appears in Section~\ref{appdx: proof of yi-partition-of-answers thm} of the Appendix, shows that in the \yi, there are many perfect global assignments, that neatly partition all answers to the edge-queries.

\begin{theorem} \label{thm: yi-partition-of-answers}
Assume that $G$ is a \yi. Then there are $6^{\ell}$ perfect global assignments $f_1,\ldots,f_{6^{\ell}}$ of answers to queries, such  that:

\begin{itemize}
\item for each query $Q\in \qset^E$ to the edge-prover, for each possible answer $A\in \aset^E$, there is exactly one index $1\leq i\leq 6^{\ell}$ with $f_i(Q)=A$; and

\item for each query $Q'\in \qset^V$ to the vertex-prover, for each possible answer $A'\in \aset^V$ to $Q'$, there are exactly $2^{\ell}$ indices $1\leq i\leq 6^{\ell}$, for which $f_i(Q')=A'$.
\end{itemize}
\end{theorem}

\paragraph{Two Graphs.}
Given a 3COL(5) instance $G$ with $|V(G)|=n$, and an integer $\ell>0$, we associate a graph $H$, that we call the \emph{constraint graph}, with it. For every query $Q\in \qset^E\cup \qset^V$, there is a vertex $v(Q)$ in $H$, while for each random string $R=(Q,Q')$, there is an edge $e(R)=(v(Q),v(Q'))$. Notice that $H$ is a bipartite graph. We denote by $U^E$ the set of its vertices corresponding to the edge-queries, and by $U^V$ the set of its vertices corresponding to the vertex-queries. Recall that $|U^E|=(5n/2)^{\ell}$, $|U^V|=n^{\ell}$; the degree of every vertex in $U^E$ is $2^{\ell}$; the degree of every vertex in $U^V$ is $5^{\ell}$, and $|E(H)|=|\rset|=(5n)^{\ell}$.

Assume now that we are given some subgraph $H'\subseteq H$ of the constraint graph. We build a bipartite graph $L(H')$ associated with it (this graph takes into account the answers to the queries; it may be convenient  for now  to think that $H'=H$, but later we will use smaller sub-graphs of $H$). The vertices of $L(H')$ are partitioned into two subsets:

\begin{itemize}
\item For each edge-query $Q\in \qset^E$ with $v(Q)\in H'$, for each possible answer $A\in \aset^E$ to $Q$, we introduce a vertex $v(Q,A)$. We denote by $S(Q)$ the set of these $6^{\ell}$ vertices corresponding to $Q$, and we call them a \emph{group representing $Q$}. We denote by $\hat U^E$ the resulting set of vertices:

\[\hat U^E=\set{v(Q,A)\mid (Q\in \qset^E\mbox{ and } v(Q)\in H'), A\in \aset^E}.\]

\item For each vertex-query $Q'\in \qset^V$ with $v(Q')\in H'$, for each possible answer $A'\in \aset^V$ to $Q'$, we introduce $2^{\ell}$ vertices $v_1(Q',A'),\ldots,v_{2^{\ell}}(Q',A')$. We call all these vertices \emph{the copies of answer $A'$ to query $Q'$}. We denote by $S(Q')$ the set of all vertices corresponding to $Q'$:

\[S(Q')=\set{v_i(Q',A')\mid A'\in \aset^V, 1\leq i\leq 2^{\ell}},\]
 
 so $|S(Q')|=6^{\ell}$. We call $S(Q')$ \emph{the group representing $Q'$}. We denote by $\hat U^V$ the resulting set of vertices:
 
 \[\hat U^V=\set{v_i(Q',A')\mid  (Q'\in \qset^V\mbox{ and } v(Q')\in H'), A'\in \aset^V, 1\leq i\leq 2^{\ell}}.\]
 
\end{itemize}

The final set of vertices of $L(H')$ is $\hat U^E\cup \hat U^V$. We define the set of edges of $L(H')$ as follows. For each random string $R=(Q^E,Q^V)$ whose corresponding edge $e(R)$ belongs to $H'$, for every answer $A\in \aset^E$ to $Q^E$, let $A'\in \aset^V$ be the unique answer to $Q^V$ consistent with $A$. For each copy $v_i(Q^V,A')$ of answer $A'$ to query $Q^V$, we add an edge $(v(Q^E,A),v_i(Q^V,A'))$. Let 

\[E(R)=\set{(v(Q^E,A),v_i(Q^V,A'))\mid A\in \aset^E, A'\in \aset^V, \mbox{ $A$ and $A'$ are consistent answers to $R$}, 1\leq i\leq 2^{\ell}}\]

be the set of the resulting edges, so $|E(R)|=6^{\ell}\cdot 2^{\ell}=12^{\ell}$. We denote by $\hat E$ the set of all edges of $L(H')$ --- the union of the sets $E(R)$ for all random strings $R$ with $e(R)\in H'$.

Recall that we have defined a partition of the set $\hat U^E$ of vertices into groups $S(Q)$ --- one group for each query $Q\in \qset^E$ with $v(Q)\in H'$. We denote this partition by $\uset_1$. Similarly, we have defined a partition of $\hat U^V$ into groups, that we denote by $\uset_2=\set{S(Q')\mid Q'\in \qset^V\mbox{ and } v(Q')\in H'}$.
Recall that for each group $U\in \uset_1\cup\uset_2$, $|U|=6^{\ell}$.

Finally, we need to define bundles of edges in graph $L(H')$. 
For every vertex $v\in \hat U^E\cup \hat U^V$, we define a partition $\bset(v)$ of the set of all edges incident to $v$ in $L(H')$ into bundles, as follows. Fix some group $U\in \uset_1\cup \uset_2$ that we have defined. If there is at least one edge of $L(H')$ connecting $v$ to the vertices of $U$, then we define a bundle containing all edges connecting $v$ to the vertices of $U$, and add this bundle to $\bset(v)$. Therefore, if $v\in S(Q)$, then for each random string $R$ in which $Q$ participates, with $e(R)\in H'$, we have defined one bundle of edges in $\bset(v)$. For each vertex $v\in \hat U^E\cup \hat U^V$, the set of all edges incident to $v$ is thus partitioned into a collection of bundles, that we denote by $\bset(v)$, and we denote $\beta(v)=|\bset(v)|$. Note that, if $v\in S(Q)$ for some query $Q\in \qset^E\cup \qset^V$, then $\beta(v)$ is exactly the degree of the vertex $v(Q)$ in graph $H'$. Note also that $\bigcup_{v\in V(H')}\bset(v)$ does not define a partition of the edges of $\hat E$, as each such edge belongs to two bundles. However, each of $\bigcup_{v\in \hat U^E}\bset(v)$ and $\bigcup_{v\in \hat U^V}\bset(v)$ does define a partition of $\hat E$.
It is easy to verify that every bundle that we have defined contains exactly $2^{\ell}$ edges.

\label{--------------------------------------sec: WGP---------------------------------------}
\section{The $(r,h)$-Graph Partitioning Problem}\label{secL WGP}

We will use a graph partitioning problem as a proxy in order to reduce the 3COL(5) problem to \NDPgrid. The specific graph partitioning problem is somewhat complex. We first define a simpler variant of this problem, and then provide the intuition and the motivation for the more complex variant that we eventually use. 

In the basic $(r,h)$-Graph Partitioning problem, that we denote by \WGP, we are given a bipartite graph $\tG=(V_1,V_2,E)$ and two integral parameters $h,r>0$. A solution consists of a partition $(W_1,\ldots,W_r)$ of $V_1\cup V_2$ into $r$ subsets, and for each $1\leq i\leq r$, a subset $E_i\subseteq E(W_i)$ of edges, such that $|E_i|\leq h$. The goal is to maximize $\sum_i|E_i|$.

One intuitive way to think about the \WGP problem is that we would like to partition the vertices of $\tG$ into $r$ clusters, that are roughly balanced (in terms of the number of edges in each cluster). However, unlike the standard balanced partitioning problems, that attempt to minimize the number of edges connecting the different clusters, our goal is to maximize the total number of edges that remain in the clusters. We suspect that the \WGP problem is very hard to approximate; in particular it appears to be somewhat similar to the Densest $k$-Subgraph problem (\DkS). Like in the \DkS problem, we are looking for dense subgraphs of $\tG$ (the subgraphs $\tG[W_i]$), but unlike \DkS, where we only need to find one such dense subgraph, we would like to partition all vertices of $\tG$ into a prescribed number of dense subgraphs. We can prove that \NDPgrid is at least as hard as \WGP (to within polylogarithmic factors; see below), but unfortunately we could not prove strong hardness of approximation results for \WGP. In particular, known hardness proofs for \DkS do not seem to generalize to this problem. To overcome this difficulty,  we define a slightly more general problem, and then use it as a proxy in our reduction.
Before defining the more general problem, we start with intuition.

\paragraph{Intuition:} Given a 3COL(5) instance $G$, we can construct the graph  $H$, and the graph $L(H)$, as described above. We  can then view $L(H)$ as an instance of \WGP, with $r=6^{\ell}$ and $h=|\rset|$. Assume that $G$ is a \yi. Then we can use the perfect global assignments $f_1,\ldots,f_r$ of answers to the queries, given by Theorem~\ref{thm: yi-partition-of-answers}, in order to partition the vertices of $L(H)$ into $r=6^{\ell}$ clusters $W_1,\ldots,W_r$, as follows. Fix some $1\leq i\leq r$. For each query $Q\in \qset^E$ to the edge-prover, set $W_i$ contains a single vertex $v(Q,A)\in S(Q)$, where $A=f_i(Q)$. For each query $Q'\in \qset^V$ to the vertex-prover, set $W_i$ contains a single vertex $v_j(Q',A')$, where $A'=f_i(Q')$, and the indices $j$ are chosen so that every vertex $v_j(Q',A')$ participates in exactly one cluster $W_i$. From the construction of the graph $L(H)$ and the properties of the assignments $f_i$ guaranteed by Theorem~\ref{thm: yi-partition-of-answers}, we indeed obtain a partition $W_1,\ldots,W_r$ of the vertices of $L(H)$. For each $1\leq i\leq r$, we then set $E_i=E(W_i)$. Notice that for every query $Q\in \qset^E\cup \qset^V$, exactly one vertex of $S(Q)$ participates in each cluster $W_i$. Therefore, for each group $U\in \uset_1\cup \uset_2$, each cluster $W_i$ contains exactly one vertex from this group. It is easy to verify that for each $1\leq i\leq r$, for each random string $R\in \rset$, set $E_i$ contains exactly one edge of $E(R)$, and so $|E_i|=|\rset|=h$, and the solution value is $h\cdot r$. Unfortunately, in the \ni, we may still obtain a solution of a high value, as follows: instead of distributing, for each query $Q\in \qset^E\cup \qset^V$, the vertices of $S(Q)$ to different clusters $W_i$, we may put all vertices of $S(Q)$ into a single cluster. While in our intended solution to the \WGP problem instance each cluster can be interpreted as an assignment of answers to the queries, and the number of edges in each cluster is bounded by the number of random strings satisfied by this assignment, we may no longer use this interpretation with this new type of solutions\footnote{We note that a similar problem arises if one attempts to design naive hardness of approximation proofs for \DkS.}. Moreover, unlike in the \yi solutions,
if we now consider some cluster $W_i$, and some random string $R\in \rset$, we may add several edges of $E(R)$ to $E_i$, which will further allow us to accumulate a high solution value. One way to get around this problem is to impose additional restrictions on the feasible solutions to the \WGP problem, which are consistent with our \yi solution, and thereby obtain a more general (and hopefully more difficult) problem. But while doing so we still need to ensure that we can prove that \NDPgrid remains at least as hard as the newly defined problem. Recall the definition of bundles in graph $L(H)$. It is easy to verify that in our intended solution to the \yi, every bundle contributes at most one edge to the solution. This motivates our definition of a slight generalization of the \WGP problem, that we call $(r,h)$-Graph Partitioning with Bundles, or \WGPwB.


The input to \WGPwB problem is almost the same as before: we are given a bipartite graph $\tG=(V_1,V_2,E)$, and two integral parameters $h,r>0$. Additionally, we are given a partition $\uset_1$ of $V_1$ into groups, and a partition $\uset_2$ of $V_2$ into groups, so that for each $U\in \uset_1\cup \uset_2$, $|U|=r$.
Using these groups, we define bundles of edges as follows: for every vertex $v\in V_1$, for each group $U\in \uset_2$, such that some edge of $E$ connects $v$ to a vertex of $U$, the set of all edges that connect $v$ to the vertices of $U$ defines a single bundle. Similarly, for every vertex $v\in V_2$, for each group $U\in \uset_1$, all edges that connect $v$ to the vertices of $U$ define a bundle. We denote, for each vertex $v\in V_1\cup V_2$, by $\bset(v)$ the set of all bundles into which the edges incident to $v$ are partitioned, and we denote by $\beta(v)=|\bset(v)|$ the number of such  bundles. We also denote by $\bset=\bigcup_{v\in V_1\cup V_2}\bset(v)$ -- the set of all bundles. Note that as before, $\bset$ is not a partition of $E$, but every edge of $E$ belongs to exactly two bundles: one bundle in $\bigcup_{v\in V_1}\bset(v)$, and one bundle in $\bigcup_{v\in V_2}\bset(v)$.
As before, we need to compute a partition $(W_1,\ldots,W_r)$ of $V_1\cup V_2$ into $r$ subsets, and for each $1\leq i\leq r$, select a subset $E_i\subseteq E(W_i)$ of edges, such that $|E_i|\leq h$.  But now there is an additional restriction: we require that for each $1\leq i\leq r$, for every bundle $B\in \bset$, $E_i$ contains at most one edge $e\in B$. As before, the goal is to maximize $\sum_i|E_i|$.


\paragraph{Valid Instances and Perfect Solutions.}

Given an instance $\iset=(\tG=(V_1,V_2,E), \uset_1,\uset_2,h,r)$ of \WGPwB, let $\beta^*(\iset)=\sum_{v\in V_1}\beta(v)$. Note that for any solution to $\iset$, the 
solution value must be bounded by $\beta^*(\iset)$, since for every vertex $v\in V_1$, for every bundle $B\in \bset(v)$, at most one edge from the bundle may contribute to the solution value. In all instances of \WGPwB that we consider, we always set $h=\beta^*(\iset)/r$. Next, we define valid instances; they are defined so that the instances that we obtain when reducing from 3COL(5) are always valid, as we show later.

\begin{definition}
We say that instance $\iset$ of \WGPwB is \emph{valid} iff $h=\beta^*(\iset)/r$ and $h\geq \max_{v\in V_1\cup V_2}\set{\beta(v)}$.
\end{definition}

Recall that for every group $U\in \uset_1\cup \uset_2$, $|U|=r$. We now define perfect solutions to the \WGPwB problem. We will ensure that our intended solutions in the \yi are always perfect, as we show later.

\begin{definition}
We say that a solution $((W_1,\ldots,W_r),(E_1,\ldots,E_r))$ to a valid \WGPwB instance $\iset$ is \emph{perfect} iff:

\begin{itemize}
\item For each group $U\in \uset_1\cup \uset_2$, exactly one vertex of $U$ belongs to each cluster $W_i$;  and
\item For each $1\leq i\leq r$, $|E_i|=h$.
\end{itemize}
\end{definition}

Note that the value of a perfect solution to a valid instance $\iset$ is $h\cdot r=\beta^*(\iset)$, and this is the largest value that any solution can achieve.


\subsection{From 3COL(5) to \WGPwB}

 Suppose we are given an instance $G$ of the  3COL(5) problem, and an integral parameter $\ell>0$ (the number of repetitions). Consider the corresponding constraint graph $H$, and suppose we are given some subgraph $H'\subseteq H$. We define an instance $\iset(H')$ of \WGPwB, as follows.

\begin{itemize}
\item The underlying graph is $L(H')=(\hat U^E,\hat U^V,\hat E)$;

\item The parameters are $r=6^{\ell}$ and $h=|E(H')|$;

\item The partition $\uset_1$ of $\hat U^E$ is the same as before: the vertices of $\hat U^E$ are partitioned into groups $S(Q)$ --- one group for each query $Q\in \qset^E$ with $v(Q)\in V(H')$. Similarly, the partition $\uset_2$ of $\hat U^V$ into groups is also defined exactly as before, and contains, for each query $Q'\in \qset^V$ with $v(Q')\in V(H')$, a group $S(Q')$. (Recall that for all $Q\in \qset^E\cup \qset^V$ with $v(Q)\in H'$, $|S(Q)|=6^{\ell}$).
\end{itemize}


\begin{claim}\label{claim: yi of 3col gives perfect solution to WGP}
Let $G$ be an instance of the 3COL(5) problem, $\ell>0$ an integral parameter, and $H'\subseteq H$ a subgraph of the corresponding constraint graph. Consider the corresponding instance $\iset(H')$ of \WGPwB.  Then $\iset(H')$ is a valid instance, and moreover, if $G$ is a \yi, then there is a perfect solution to $\iset(H')$.
\end{claim}

\begin{proof}
We first verify that $\iset(H')$ is a valid instance of \WGPwB. Recall that for a query $Q\in \qset^E$ to the edge-prover and an answer $A\in \aset^E$, the number of bundles incident to vertex $v(Q,A)$ in $L(H')$ is exactly the degree of the vertex $v(Q)$ in graph $H'$. The total number of bundles incident to the vertices of $S(Q)$ is then the degree of $v(Q)$ in $H'$ times $|\aset^E|$. Therefore, $\beta^*(\iset)=\sum_{v(Q,A)\in \hat U^E}|\beta(v)|=|E(H')|\cdot |\aset^E|=h\cdot r$. It is now immediate to verify that $h=\beta^*(\iset)/r$. Similarly, for a vertex $v=v_j(Q',A')\in U^V$, the number of bundles incident to $v$ is exactly the degree of $v$ in $H'$. Since $h=|E(H')|$, we get that $h\geq \max_{v\in V_1\cup V_2}\set{\beta(v)}$, and so $\iset(H')$ is a valid instance. 

Assume now that $G$ is a \yi. We  define a perfect solution $((W_1,\ldots,W_r),(E_1,\ldots,E_r))$ to this instance.
Let $\set{f_1,f_2,\ldots,f_{6^\ell}}$ be the collection of perfect global assignments of answers to the queries, given by Theorem~\ref{thm: yi-partition-of-answers}. Recall that $\uset_1=\set{S(Q)\mid Q\in\qset^E, v(Q)\in H'}$ and $\uset_2=\set{S(Q')\mid Q'\in\qset^V, v(Q')\in H'}$, where each group in $\uset_1\cup \uset_2$ has cardinality $r=6^{\ell}$. We now fix some $1\leq i\leq r$, and define the set $W_i$ of vertices. For each query $Q\in \qset^E$ to the edge-prover, if $A=f_i(Q)$, then we add the vertex $v(Q,A)$ to $W_i$. For each query $Q'\in \qset^V$ to the vertex-prover, if $A'=f_i(Q')$, then we select some index $1\leq j\leq 2^{\ell}$, and add the vertex $v_j(Q',A')$ to $W_i$. The indices $j$ are chosen so that every vertex $v_j(Q',A')$ participates in at most one cluster $W_i$. From the construction of the graph $L(H')$ and the properties of the assignments $f_i$ guaranteed by Theorem~\ref{thm: yi-partition-of-answers}, it is easy to verify that $W_1,\ldots,W_r$ partition the vertices of $L(H')$, and moreover, for each group $S(Q)\in \uset_1\cup \uset_2$, each set $W_i$ contains exactly one vertex of $S(Q)$. 

Finally, for each $1\leq i\leq r$, we set $E_i=E(W_i)$. We claim that for each bundle $B\in \bset$, set $E_i$ may contain at most one edge of $B$. Indeed, let $v\in W_i$ be some vertex, let $U\in \uset_1\cup \uset_2$ be some group, and let $B$ be the bundle containing all edges that connect $v$ to the vertices of $U$. Since $W_i$ contains 
 exactly one vertex of $U$, at most one edge of $B$ may belong to $E_i$. 


It now remains to show that $|E_i|=h$ for all $i$. Fix some $1\leq i\leq r$. It is easy to verify that for each random string $R=(Q,Q')$ with $e(R)\in H'$, set $W_i$ contains a pair of vertices $v(Q,A)$, $v_j(Q',A')$, where $A$ and $A'$ are matching answers to $Q$ and $Q'$ respectively, and so the corresponding edge connecting this pair of vertices in $L(H')$ belongs to $E_i$. Therefore, $|E_i|=|E(H')|=h$.
\end{proof}

\subsection{From \WGPwB to \NDP} 

The following definition will be useful for us later, when we extend our results to \NDP and \EDP on wall graphs.

\begin{definition}
Let $\pset$ be a set of paths in a grid $\hat G$.
We say that $\pset$ is a \emph{spaced-out} set iff for each pair $P,P' \in \pset$ of paths, $d(V(P),V( P')) \geq 2$, and all paths in $\pset$ are internally disjoint from the boundaries of the grid $\hat G$.    
\end{definition}

Note that if $\pset$ is a set of paths that is spaced-out, then all paths in $\pset$ are mutually node-disjoint.
The following theorem is central to our hardness of approximation proof.

\begin{theorem}\label{thm: from WGP to NDP}
There  is a constant $c^*>0$, and there is an efficient randomized algorithm, that, given a valid instance $\iset=(\tilde G, \uset_1,\uset_2,h,r)$ of \WGPwB with $|E(\tilde G)|=M$, constructs an instance $\hat{\iset}=(\hat G,\mset)$ of \NDPgrid with $|V(\hat G)|=O(M^4\log^2M)$, such that the following hold:

\begin{itemize}
\item If $\iset$ has a perfect solution (of value $\beta^*=\beta^*(\iset)$), then with probability at least $\half$ over the construction of $\hat \iset$, instance $\hat\iset$ has a solution $\pset$ that routes at least $\frac{\beta^*}{c^*\log^3M}$ demand pairs, such that the paths in $\pset$ are spaced-out; and

\item There is a deterministic efficient algorithm, that, given a solution $\pset^*$ to the \NDPgrid problem instance $\hat{\iset}$, constructs a solution to the \WGPwB instance $\iset$, of value at least $\frac{|\pset^*|}{c^*\cdot \log^3M}$.
\end{itemize}
\end{theorem}

We note that the theorem is slightly stronger than what is needed in order to prove hardness of approximation of \NDPgrid: if $\iset$ has a perfect solution, then it is sufficient to ensure that the set $\pset$ of paths in the corresponding \NDPgrid instance $\hat \iset$ is node-disjoint. But we will use the stronger guarantee that it is spaced-out when extending our results to \NDP and \EDP in wall graphs. Note that in the second assertion we are only guaranteed that the paths in $\pset^*$ are node-disjoint.
The proof of the theorem is somewhat technical and is deferred to Section~\ref{sec: from WGP to NDP}.

Assume now that we are given an instance $G$ of 3COL(5), an integral parameter $\ell>0$, and a subgraph $H'\subseteq H$ of the corresponding constraint graph. Recall that we have constructed a corresponding instance $\iset(H')$ of \WGPwB. We can then use Theorem~\ref{thm: from WGP to NDP} to construct a (random) instance of \NDPgrid, that we denote by $\hat{\iset}(H')$.
Note that $|E(L(H'))|\leq 2^{O(\ell)}\cdot |E(H)|\leq n^{O(\ell)}$. Let $\hat c$ be a constant, such that $|E(L(H'))|\leq n^{\hat c\ell}$ for all $H'\subseteq H$; we can assume w.l.o.g. that $\hat c>1$. We can also assume w.l.o.g. that $c^*\geq 1$, where $c^*$ is the constant from Theorem~\ref{thm: from WGP to NDP}, and we denote $\cyi=(\hat c\cdot c^*)^3$.
We obtain the following immediate corollary of Theorem~\ref{thm: from WGP to NDP}:

\begin{corollary}\label{cor: YI for 3COL has good routing}
Suppose we are given 3COL(5) instance $G$ that is a \yi, an integer $\ell>0$, and a subgraph $H'\subseteq H$ of the corresponding constraint graph. Then with probability at least $\half$, instance $\hat{\iset}(H')$ of \NDPgrid
has a solution of value at least $\frac{|E(H')|\cdot 6^{\ell}}{\cyi\ell^3\log^3n}$, where $n=|V(G)|$. (The probability is over the random construction of $\hat\iset(H')$).
\end{corollary}

\begin{proof}
From Claim~\ref{claim: yi of 3col gives perfect solution to WGP}, instance $\iset(H')=(L(H'), \uset_1,\uset_2,r,h)$ of \WGPwB is a valid instance, and it has a perfect solution, whose value must be $\beta^*=\beta^*(\iset(H'))=h\cdot r=|E(H')|\cdot 6^{\ell}$. From Theorem~\ref{thm: from WGP to NDP}, with probability at least $1/2$, instance $\hat \iset(H')$ of \NDPgrid has a solution of value at least 
$\frac{|E(H')|\cdot 6^{\ell}}{c^*\log^3M}$, where $M=|E(L(H'))|$. Since $\log M\leq \hat c \ell\log n$, and the corollary follows.
\end{proof}

\label{--------------------------------------sec: hardness proof---------------------------------------}
\section{The Hardness Proof}\label{sec: the hardness proof}
Let $G$ be an input instance of 3COL(5). 
Recall that $\gamma$ is the absolute constant from the Parallel Repetition Theorem (Corollary~\ref{cor: parallel repetition}).
We will set the value of the parameter $\ell$ later, ensuring that $\ell>\log^2 n$, where $n=|V(G)|$.
Let $\alpha^*=2^{\Theta(\ell/\log n)}$ be the hardness of approximation factor that we are trying to achieve.

Given the tools developed in the previous sections, a standard way to prove hardness of \NDPgrid would work as follows. Given an instance $G$ of 3COL(5) and the chosen parameter $\ell$, construct the corresponding graph $H$ (the constraint graph), together with the graph $L(H)$. We then construct an  instance $\iset(H)$ of \WGPwB as described in the previous section, and convert it into an instance $\hat{\iset}(H)$ of \NDPgrid. 

We note that, if $G$ is a \yi, then from Corollary~\ref{cor: YI for 3COL has good routing}, with constant probability there is a solution to $\hat \iset(H)$ of value $\frac{|\rset|\cdot 6^{\ell}}{\cyi\ell^3\log^3n}$. 
Assume now that $G$ is a \ni. If we could show that any solution to the corresponding \WGPwB instance $\iset(H)$ has value less than $\frac{|\rset|\cdot 6^{\ell}}{\cyi^2\cdot \alpha^*\ell^6\log^6n}$, we would be done. Indeed, in such a case,  from Theorem~\ref{thm: from WGP to NDP}, every solution to the \NDPgrid instance $\hat\iset(H)$ routes fewer than $\frac{ |\rset|\cdot 6^{\ell}}{\cyi\alpha^*\ell^3\log^3n}$ demand pairs. If we assume for contradiction that an $\alpha^*$-approximation algorithm exists for \NDPgrid, then, if $G$ is a \yi, the algorithm would have to return a solution to $\hat\iset(H)$ routing at least $\frac{|\rset|\cdot 6^{\ell}}{\cyi\alpha^*\ell^3\log^3n}$ demand pairs, while, if $G$ is a \ni, no such solution would exist. Therefore, we could use the $\alpha^*$-approximation algorithm for \NDPgrid to distinguish between the \yis and the \nis of 3COL(5).

Unfortunately, we are unable to prove this directly. Our intended solution to the \WGPwB instance $\iset(H)$, defined over the graph $L(H)$, for each query $Q\in \qset^E\cup\qset^V$, places every vertex of $S(Q)$ into a distinct cluster. Any such solution will indeed have a low value in the \ni. But a cheating solution may place many vertices from the same set $S(Q)$ into some cluster $W_j$. Such a solution may end up having a high value, but it may not translate into a good strategy for the two provers, that satisfies a large fraction of the random strings. In an extreme case, for each query $Q$, we may place all vertices of $S(Q)$ into a single cluster $W_i$. The main idea in our reduction is to overcome this difficulty by noting that such a cheating solution can be used to compute a partition of the constraint graph $H$. The graph is partitioned into as many as $6^{\ell}$ pieces, each of which is significantly smaller than the original graph $H$. At the same time, a large fraction of the edges of $H$ will survive the partitioning procedure. Intuitively, if we now restrict ourselves to only those random strings $R\in \rset$, whose corresponding edges have survived the partitioning procedure, then the problem does not become significantly easier, and we can recursively apply the same argument to the resulting subgraphs of $H$. We make a significant progress in each such iteration, since the sizes of the resulting sub-graphs of $H$ decrease very fast.
The main tool that allows us to execute this plan is the following theorem.

\begin{theorem}\label{thm: good strategy or partition}
Suppose we are given an instance $G$ of the 3COL(5) problem with $|V(G)|=n$, and an integral parameter $\ell>\log^2 n$, together with some subgraph $H'\subseteq H$ of the corresponding constraint graph $H$, and a parameter $P>1$. Consider the corresponding instance $\iset(H')$ of \WGPwB, and assume that we are given a solution to this instance of value at least $|E(H')|\cdot 6^{\ell}/\alpha$, where $\alpha=\cyi^2\cdot \alpha^*\cdot \ell^6\log^6n$. Then there is a  randomized algorithm whose running time is $O\left (n^{O(\ell)}\cdot  \log P\right )$, that returns one of the following:

\begin{itemize}
\item Either a randomized strategy for the two provers that satisfies, in expectation, more than a $2^{-\gamma \ell/2}$-fraction of the constraints $R\in \rset$ with $e(R)\in E(H')$; or

\item A collection $\hset$ of disjoint sub-graphs of $H'$, such that for each $H''\in \hset$, $|E(H'')|\leq |E(H')|/2^{\gamma\ell/16}$, and with probability at least $(1-1/P)$, $\sum_{H''\in \hset}|E(H'')|\geq \frac{c'|E(H')|}{\ell^2 \alpha^2}$, for some universal constant $c'$.
\end{itemize}
\end{theorem}

We postpone the proof of the theorem to the following subsection, after we complete the hardness proof for \NDPgrid.
We assume for contradiction that we are given a factor-$\alpha^*$ approximation algorithm $\aset$ for \NDPgrid (recall that $\alpha^*=2^{\Theta(\ell/\log n)}$). We will use this algorithm to distinguish between the \yis and the \nis of 3COL(5). Suppose we are given an instance $G$ of 3COL(5).

For an integral parameter $\ell>\log^2 n$, let $H$ be the constraint graph corresponding to $G$ and $\ell$. We next show a randomized algorithm, that uses $\aset$ as a subroutine, in order to determine whether $G$ is a \yi or a \ni. The running time of the algorithm is $n^{O(\ell)}$. 

Throughout the algorithm, we maintain a collection $\hset$ of sub-graphs of $H$, that we sometimes call clusters. Set $\hset$ is in turn partitioned into two subsets: set $\hset_1$ of active clusters, and set $\hset_2$ of inactive clusters. Consider now some inactive cluster $H'\in \hset_2$. This cluster defines a $2$-prover game $\gset(H')$, where the queries to the two provers are $\set{Q^E\in \qset^E \mid v(Q^E)\in V(H')}$, and $\set{Q^V\in \qset^V \mid v(Q^V)\in V(H')}$ respectively, and the constraints of the verifier are $\rset(H')=\set{R\in \rset\mid e(R)\in E(H')}$. For each inactive cluster $H'\in \hset_2$, we will store a (possibly randomized) strategy of the two provers for game $\gset(H')$, that satisfies at least a  $2^{-\gamma\ell/2}$-fraction of the constraints in $\rset(H')$.

At the beginning, $\hset$ contains a single cluster -- the graph $H$, which is active. The algorithm is executed while $\hset_1\neq \emptyset$, and its execution is partitioned into phases. In every phase, we process each of the clusters that belongs to $\hset_1$ at the beginning of the phase. Each phase is then in turn is partitioned into iterations, where in every iteration we process a distinct active cluster $H'\in \hset_1$. We describe an iteration when an active cluster $H'\in \hset_1$ is processed in Figure~\ref{fig: iteration} (see also the flowchart in Figure~\ref{fig: flowchart}).

\begin{figure}[h]
\program{Iteration for Processing a Cluster $H'\in \hset_1$}{
\begin{enumerate}
\item Construct an instance $\iset(H')$ of \WGPwB. 

\item Use Theorem~\ref{thm: from WGP to NDP} to independently construct $n^{4\ell}$ instances $\hat{\iset}(H')$ of the \NDPgrid problem.
\item Run the $\alpha^*$-approximation algorithm $\aset$ on each such instance $\hat{\iset}(H')$. If the resulting solution, for each of these instances, routes fewer than $\frac{|E(H')|6^{\ell}}{\cyi \alpha^*\ell^3\log^3n}$ demand pairs, halt and return ``$G$ is a \ni''. \label{alg returns NI 1}

\item Otherwise, fix any instance $\hat\iset(H')$ for which the algorithm returned a solution routing at least $\frac{|E(H')|6^{\ell}}{\cyi \alpha^*\ell^3\log^3n}$ demand pairs. Denote $|E(L(H'))|=M$, and recall that $M\leq n^{\hat c \ell}$. Use Theorem~\ref{thm: from WGP to NDP} to compute a solution $((W_1,\ldots,W_r),(E_1,\ldots,E_r))$ to the instance $\iset(H')$ of the \WGPwB problem, of value at least:

\[\frac{|E(H')|\cdot 6^{\ell}}{(\cyi  \alpha^*\ell^3\log^3n)(c^*\log^3M)}\geq \frac{|E(H')|\cdot 6^{\ell}}{\cyi c^* \hat c^3 \alpha^*\ell^6\log^6n}\geq  \frac{|E(H')|\cdot 6^{\ell}}{\cyi^2  \alpha^*\ell^6\log^6n}= \frac{|E(H')|\cdot 6^{\ell}}{\alpha}.\]

\item Apply the algorithm from Theorem~\ref{thm: good strategy or partition} to this solution, with the parameter $P=n^{c\ell}$, for a sufficiently large constant $c$.
 
 \begin{enumerate}
 \item If the outcome is a strategy for the provers satisfying more than a $2^{-\gamma\ell/2}$-fraction of constraints $R\in \rset$ with $e(R)\in E(H')$, then declare cluster $H'$  inactive and move it from $\hset_1$ to $\hset_2$. Store the resulting strategy of the provers. 
 
 \item Otherwise, let $\tilde \hset$ be the collection of sub-graphs of $H'$ returned by the algorithm. If $\sum_{H''\in \hset}|E(H'')|< \frac{c'|E(H')|}{\ell^2 \alpha^2}$, then  return ``$G$ is a \ni''. Otherwise, remove $H'$ from $\hset_1$ and add all graphs of $\tilde \hset$ to $\hset_1$. \label{alg returns NI2}
 \end{enumerate}
 \end{enumerate}
}

\caption{An iteration description\label{fig: iteration}}
\end{figure}
 
 If the algorithm terminates with $\hset$ containing only inactive clusters, then we return ``$G$ is a \yi''.

\begin{figure}
\center
\includegraphics[width=17cm]{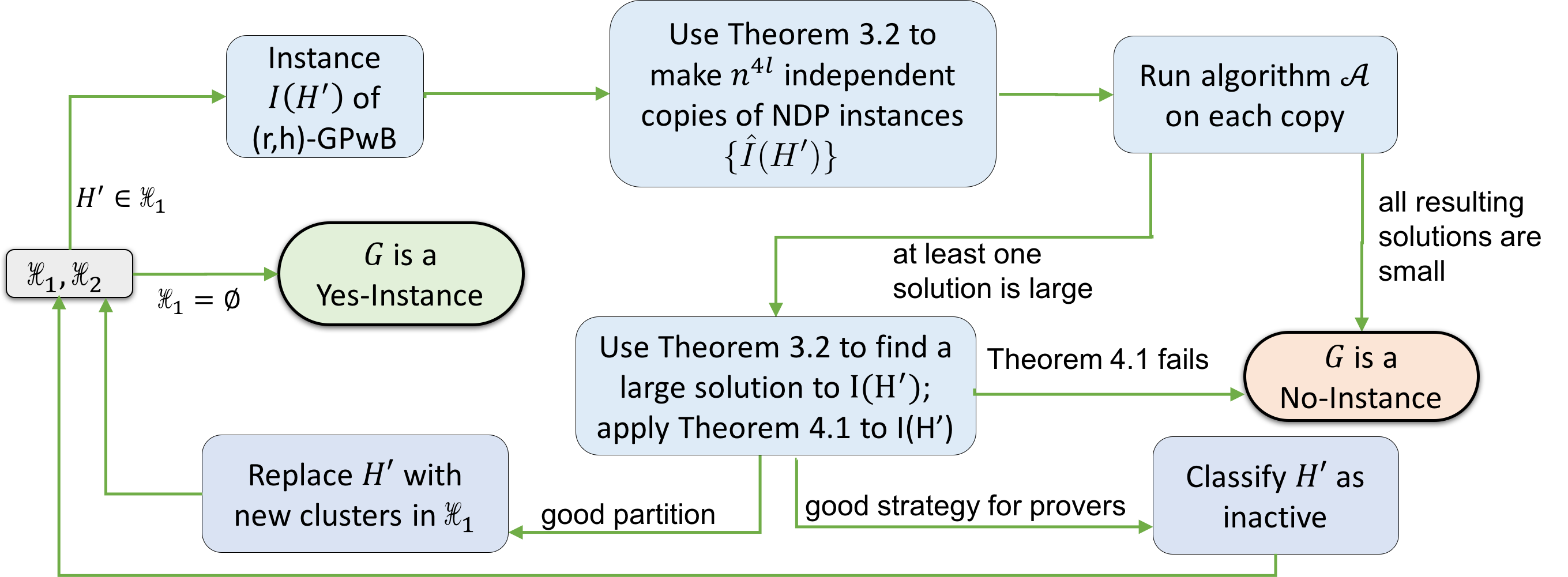}

\caption{Flowchart for an iteration execution} \label{fig: flowchart}
\end{figure}

 \paragraph{Correctness.}
 We establish the correctness of the algorithm in the following two lemmas.
 
 \begin{lemma}\label{lem: correctness soundness}
 If $G$ is a \yi, then with high probability, the algorithm returns ``$G$ is a \yi''.
 \end{lemma}
 
 \begin{proof}
 Consider an iteration of the algorithm when an active cluster $H'$ is processed. Notice that the algorithm may only determine that $G$ is a \ni in Step~(\ref{alg returns NI 1}) or in Step~(\ref{alg returns NI2}). We now analyze these two steps.
   
 Consider first Step~(\ref{alg returns NI 1}).
 From Corollary~\ref{cor: YI for 3COL has good routing}, with probability at least $1/2$, a random graph $\hat \iset (H')$ has a solution of value at least $\frac{|E(H')|\cdot 6^{\ell}}{\cyi\ell^3\log^3n}$, and our $\alpha^*$-approximation algorithm to \NDPgrid must then return a solution of value at least $\frac{|E(H')|\cdot 6^{\ell}}{\cyi \alpha^*\ell^3\log^3n}$ . Since we use $n^{4\ell}$ independent random constructions of $\hat \iset (H')$, with high probability, for at least one of them, we will obtain a solution of value at least $\frac{|E(H')|6^{\ell}}{\cyi \alpha^*\ell^3\log^3n}$.  Therefore, with high probability our algorithm will not return ``$G$ is a \ni'' due to Step~(\ref{alg returns NI 1}) in this iteration.
 
 Consider now Step~(\ref{alg returns NI2}). The algorithm can classify $G$ as a \ni in this step only if $\sum_{H''\in \hset}|E(H'')|< \frac{c'|E(H')|}{\ell^2 \alpha^2}$. From Theorem~\ref{thm: good strategy or partition}, this happens with probability at most $1/P$, and from our setting of the parameter $P$ to be $n^{c\ell}$ for a large enough constant $c$, with high probability our algorithm will not return ``$G$ is a \ni'' due to Step~(\ref{alg returns NI2}) in this iteration.
 
 It is not hard to see that our algorithm performs $n^{O(\ell)}$ iterations, and so, using the union bound, with high probability, it will classify $G$ as a \yi.
 \end{proof}

 \begin{lemma}\label{lem: correctness completeness}
 If $G$ is a \ni, then the algorithm always returns ``$G$ is a \ni''.
 
 \end{lemma}
 \begin{proof}
From Corollary~\ref{cor: parallel repetition}, it is enough to show that, whenever the algorithm classifies $G$ as a \yi, there is a strategy for the two provers, that satisfies more than a fraction-$2^{-\gamma\ell}$ of the constraints in $\rset$. 
 
 Note that the original graph $H$ has  at most $n^{\hat c \ell}$ edges. In every phase, the number of edges in each active graph decreases by a factor of at least $2^{\gamma \ell/16}$. Therefore, the number of phases is bounded by $O(\log n)$. If the algorithm classifies $G$ as a \yi, then it must terminate when no active clusters remain. In every phase, the number of edges in $\bigcup_{H'\in \hset}E(H')$ goes down by at most a factor $\ell^2\alpha^2/c'$. Therefore, at the end of the algorithm:
 
 \[\sum_{H'\in \hset_2}|E(H')|\geq \frac{|\rset|}{(\ell^2\alpha^2/c')^{O(\log n)}}=\frac{|\rset|}{(\ell^{14}\cdot (\alpha^*)^2\cdot \log ^{12}n)^{O(\log n)}}=\frac{|\rset|}{(\alpha^*)^{O(\log n)}}.\] 
 
 By appropriately setting $\alpha^*=2^{\Theta( \ell/\log n)}$, we will ensure that the number of edges remaining in the inactive clusters $H'\in \hset_2$ is at least $|\rset|/2^{\gamma\ell/4}$. Each such edge corresponds to a distinct random string $R\in \rset$. Recall that for each inactive cluster $H'$, there is a strategy for the provers in the corresponding game $\gset(H')$ that satisfies at least $|E(H')|/2^{\gamma\ell/2}$ of its constraints. Taking the union of all these strategies, we can satisfy more than $|\rset|/2^{\gamma\ell}$ constraints of $\rset$, contradicting the fact that $G$ is a \ni.
\end{proof}

\paragraph{Running Time and the Hardness Factor.}
As observed above, our algorithm has at most $n^{O(\ell)}$ iterations, where in every iteration it processes a distinct active cluster $H'\subseteq H$.  The corresponding graph $L(H')$ has at most $n^{O(\ell)}$ edges, and so each of the $n^{O(\ell)}$ resulting instances of \NDPgrid contains at most $n^{O(\ell)}$ vertices. Therefore, the overall running time of the algorithm is $n^{O(\ell)}$. From the above analysis, if $G$ is a \ni, then the algorithm always classifies it as such, and if $G$ is a \yi, then the algorithm classifies it as a \yi with high probability. The hardness factor that we obtain is $\alpha^* = 2^{\Theta(\ell/\log n)}$, while we only apply our approximation algorithm to instances of \NDPgrid containing at most $N=n^{O(\ell)}$ vertices. The running time of the algorithm is $n^{O(\ell)}$, and it is a randomized algorithm with a one-sided error.

%
%
%
%

Setting $\ell=\log^pn$ for a large enough integer $p$, we obtain $\alpha^*=2^{\Theta((\log N)^{1-2/(p+1)})}$, giving us a $2^{(\log n)^{(1-\eps)}}$-hardness of approximation for \NDPgrid for any constant $\eps$, assuming $\NP \not \subseteq \RTIME (n^{\poly \log n})$.

Setting $\ell=n^{\delta}$ for some constant $\delta$, we get that $N=2^{O(n^{\delta}\log n)}$ and $\alpha^*=2^{\Theta(n^{\delta}/\log n)}$, giving us a $n^{\Omega (1/(\log \log n)^2)}$-hardness of approximation for \NDPgrid, assuming that $\NP \not \subseteq \RTIME(2^{n^{\delta}})$ for some constant $\delta>0$.

%
%

\subsection{Proof of Theorem~\ref{thm: good strategy or partition}}

Recall that each edge of graph $H'$ corresponds to some constraint $R\in \rset$. Let $\rset'\subseteq\rset$ be the set of all constraints $R$ with $e(R)\in E(H')$. Denote the solution to the \WGPwB instance $\iset(H')$ by $((W_1,\ldots,W_r),(E_1,\ldots,E_r))$, and let $E'=\bigcup_{i=1}^r E_i$. Recall that for each random string $R\in \rset'$, there is a set $E(R)$ of $12^{\ell}$ edges in graph $L(H')$ representing $R$. Due to the way these edges are partitioned into bundles, at most $6^{\ell}$ edges of $E(R)$ may belong to $E'$. We say that a random string $R\in \rset'$ is \emph{good} iff $E'$ contains at least $6^{\ell}/(2\alpha)$ edges of $E(R)$, and we say that it is bad otherwise. 

\begin{observation}
At least $|\rset'|/(2\alpha)$ random strings of $\rset'$ are good.
\end{observation}
\begin{proof}
Let $x$ denote the fraction of good random strings in $\rset'$. A good random string contributes at most $6^{\ell}$ edges to $E'$, while a bad random string contributes at most $6^{\ell}/(2\alpha)$. If $x<1/(2\alpha)$, then a simple accounting shows that $|E'|<|\rset'|\cdot 6^{\ell}/\alpha=|E(H')|\cdot 6^{\ell}/\alpha$, a contradiction.
\end{proof}

Consider some random string $R\in \rset'$, and assume that $R=(Q^E,Q^V)$. We denote by $E'(R)=E(R)\cap E'$. Intuitively, say that a cluster $W_i$ is a \emph{terrible cluster} for $R$ if the number of edges of $E(R)$ that lie in $E_i$ is much smaller than $|Q^E\cap W_i|$ or $|Q^V\cap W_i|$. We now give a formal definition of a terrible cluster.

\begin{definition}
Given a random string $R\in \rset$ and an index $1\leq i\leq 6^{\ell}$, we say that a cluster $W_i$ is a \emph{terrible cluster}  for $R$, if: 

\begin{itemize}
\item either $|E'(R)\cap E_i|< |W_i\cap S(Q^V)|/(8\alpha)$; or
\item $|E'(R)\cap E_i|< |W_i\cap S(Q^E)|/(8\alpha)$.
\end{itemize}

We say that an edge $e\in E'(R)$ is a \emph{terrible edge} if it belongs to the set $E_i$, where $W_i$ is a terrible cluster for $R$. 
\end{definition}

\begin{observation}
For each good random string $R\in \rset'$, at most $6^{\ell}/(4\alpha)$ edges of $E'(R)$ are terrible.
\end{observation}

\begin{proof}
Assume for contradiction that more than $6^\ell/(4\alpha)$ edges of $E'(R)$ are terrible. Denote $R=(Q^E,Q^V)$.  Consider some such terrible edge $e\in E'(R)$, and assume that $e\in E_i$ for some cluster $W_i$, that is terrible for $R$. We say that $e$ is a type-$1$ terrible edge iff $|E'(R)\cap E_i|< |W_i\cap S(Q^V)|/(8\alpha)$, and it is a type-2 terrible edge otherwise, in which case $|E'(R)\cap E_i|< |W_i\cap S(Q^E)|/(8\alpha)$ must hold. Let $E^1(R)$ and $E^2(R)$ be the sets of all terrible edges of $E'(R)$ of types $1$ and $2$, respectively. Then either $|E^1(R)|> 6^{\ell}/(8\alpha)$, or $|E^2(R)|> 6^{\ell}/(8\alpha)$ must hold. 

Assume first that $|E^1(R)|> 6^{\ell}/(8\alpha)$. Fix some index $1\leq i\leq 6^{\ell}$, such that $W_i$ is a cluster that is terrible for $R$, and  $|E(R)\cap E_i|< |W_i\cap S(Q^V)|/(8\alpha)$. We assign, to each edge $e\in E_i\cap E^1(R)$, a set of $8\alpha$ vertices of $W_i\cap S(Q^V)$ arbitrarily, so that every vertex is assigned to at most one edge; we say that the corresponding edge is \emph{responsible} for the vertex. Every edge of $E_i\cap E^1(R)$ is now responsible for $8\alpha$ distinct vertices of  $W_i\cap S(Q^V)$. Once we finish processing all such clusters $W_i$, we will have assigned, to each edge of $E^1(R)$, a set of $8\alpha$ distinct vertices of $S(Q^V)$. We conclude that $|S(Q^V)|\geq 8\alpha |E^1(R)|> 6^{\ell}$. But $|S(Q^V)|=6^{\ell}$, a contradiction.

The proof for the second case, where $|E^2(R)|> 6^{\ell}/(16\alpha)$ is identical, and relies on the fact that $|S(Q^E)|=6^{\ell}$.
\end{proof}

We will use the following simple observation.

\begin{observation}\label{obs: not terrible}
Let $R\in \rset'$ be a good random string, with $R=(Q^E,Q^V)$, and let $1\leq i\leq 6^{\ell}$ be an index, such that $W_i$ is not terrible for $R$. Then $|W_i\cap S(Q^V)|\geq |W_i\cap S(Q^E)|/(8\alpha)$ and $|W_i\cap S(Q^E)|\geq |W_i\cap S(Q^V)|/(8\alpha)$.
\end{observation}

\begin{proof}
Assume first for contradiction that $|W_i\cap S(Q^V)|< |W_i\cap S(Q^E)|/(8\alpha)$. Consider the edges of $E(R)\cap E_i$. Each such edge must be incident to a distinct vertex of $S(Q^V)$. Indeed, if two edges $(e,e')\in E(R)\cap E_i$ are incident to the same vertex $v_j(Q^V,A)\in S(Q^V)$, then, since the other endpoint of each such edge lies in $S(Q^E)$, the two edges belong to the same bundle, a contradiction. Therefore, $|E_i\cap E(R)|\leq |W_i\cap S(Q^V)|< |W_i\cap S(Q^E)|/(8\alpha)$, contradicting the fact that $W_i$ is not a terrible cluster for $R$. 

The proof for the second case, where  $|W_i\cap S(Q^E)|< |W_i\cap S(Q^V)|/(8\alpha)$ is identical. As before, each edge of  $E(R)\cap E_i$ must be incident to a distinct vertex of $S(Q^E)$, as otherwise, a pair $e,e'\in E(R)$ of edges that are incident on the same vertex $v(Q^E,A)\in S(Q^E)$ belong the same bundle. Therefore, $|E_i\cap E(R)|\leq |W_i\cap S(Q^E)|< |W_i\cap S(Q^V)|/(8\alpha)$, contradicting the fact that $W_i$ is not a terrible cluster for $R$. 
\end{proof}

For each good random string $R\in \rset'$, we discard the terrible edges from set $E'(R)$, so $|E'(R)|\geq 6^{\ell}/(4\alpha)$ still holds. 

Let $z=2^{\gamma\ell/8}$. We say that cluster $W_i$ is \emph{heavy} for a random string $R=(Q^E,Q^V)\in \rset'$ iff $|W_i\cap S(Q^E)|,|W_i\cap S(Q^V)|>z$. We say that an edge $e\in E'(R)$ is heavy iff it belongs to set $E_i$, where $W_i$ is a heavy cluster for $R$. Finally, we say that a  random string $R\in \rset'$ is \emph{heavy} iff at least half of the edges in $E'(R)$ are heavy. Random strings and edges that are not heavy are called light.
We now consider two cases. The first case happens if at least half of the good random strings are light. In this case, we compute a randomized strategy for the provers to choose assignment to the queries, so that at least a $2^{-\gamma\ell/2}$-fraction of the constraints in $\rset'$ are satisfied in expectation. In the second case, at least half of the good random strings are heavy. We then compute a partition $\hset$ of $H'$ as desired.
We now analyze the two cases. Note that if $|E(H')|<z/(8\alpha)$, then Case 2 cannot happen. This is since $h=|E(H')|<z/(8\alpha)$ in this case, and so no random strings may be heavy. Therefore, if $H'$ is small enough, we will return  a strategy of the provers that satisfies a large fraction of the constraints in $\rset'$.

\paragraph{Case 1.} This case happens if at least half of the good random strings are light. Let $\lset\subseteq \rset'$ be the set of the good light random strings, so $|\lset|\geq |\rset'|/(4\alpha)$. For each such random string $R\in \lset$, we let $E^L(R)\subseteq E'(R)$ be the set of all light edges corresponding to $R$, so $|E^L(R)|\geq 6^{\ell}/(8\alpha)$. We now define a randomized algorithm to choose an answer to every query $Q\in \qset^E\cup \qset^V$ with $v(Q)\in H'$. Our algorithm chooses a random index $1\leq i\leq r$. For every query $Q\in \qset^E\cup \qset^V$ with $v(Q)\in H'$, we consider the set $\aset(Q)$ of all answers $A$, such that some vertex $v(Q,A)$ belongs to $W_i$ (for the case where $Q\in \qset^V$, the vertex is of the form $v_j(Q,A)$). We then choose one of the answers from $\aset(Q)$ uniformly at random, and assign it to $Q$. If $\aset(Q)=\emptyset$, then we choose an arbitrary answer to $Q$.

We claim that the expected number of satisfied constraints of $\rset'$ is at least $|\rset'|/2^{\gamma\ell/2}$. Since $\lset\geq |\rset'|/(4\alpha)$, it is enough to show that the expected fraction the good light constraints that are satisfied is at least $4\alpha |\lset|/2^{\gamma\ell/2}$, and for that it is sufficient to show that each light constraint $R\in \lset$ is satisfied with probability at least $4\alpha/2^{\gamma\ell/2}$. 

Fix one such constraint $R=(Q^E,Q^V)\in \lset$, and consider an edge $e\in E^L(R)$. Assume that $e$ connects a vertex $v(Q^E,A)$ to a vertex $v_j(Q^V,A')$, and that $e\in E_i$. We say that edge $e$ is happy iff our algorithm chose the index $i$, the answer $A$ to query $Q^E$, and the answer $A'$ to query $Q^V$. Notice that due to our construction of bundles, at most one edge $e\in E^L(R)$ may be happy with any choice of the algorithm; moreover, if any edge $e\in E^L(R)$ is happy, then the constraint $R$ is satisfied. The probability that a fixed edge $e$ is happy is at least $1/(8\cdot 6^{\ell}z^2\alpha)$. Indeed, we choose the correct index $i$ with probability $1/6^{\ell}$. Since $e$ belongs to $E_i$, $W_i$ is a light cluster for $R$, and so either $|S(Q^E)|\leq z$, or  $|S(Q^V)|\leq z$. Assume without loss of generality that it is the former; the other case is symmetric. Then, since $e$ is not terrible, from Observation~\ref{obs: not terrible}, $|S(Q^V)|\leq 8 \alpha z$, and so $|\aset(Q^V)|\leq 8\alpha z$, while $|\aset(Q^E)|\leq z$. Therefore, the probability that we choose answer $A$ to $Q^E$ and answer $A'$ to $A^V$ is at least $1/(8\alpha z^2)$, and overall,  the probability that a fixed constraint $R\in \lset$ is satisfied is at least $|E^L(R)|/(8\cdot 6^{\ell}z^2\alpha)\geq 1/(64z^2\alpha^2)\geq 4\alpha/ 2^{\gamma\ell/2}$, since $z=2^{\gamma\ell/8}$, and $\alpha<2^{\gamma\ell/32}$.  

\paragraph{Case 2.} This case happens if at least half of the good random strings are heavy. Let $\rset''\subseteq \rset'$ be the set of the heavy random strings, so $|\rset''|\geq |\rset'|/(4\alpha)$. For each such random string $R\in \rset''$, we let $E^H(R)\subseteq E'(R)$ be the set of all heavy edges corresponding to $R$. Recall that $|E^H(R)|\geq 6^{\ell}/(8\alpha)$.

Fix some heavy random string $R\in \rset''$ and assume that $R=(Q^E,Q^V)$. For each $1\leq i\leq r$, let $E_i(R)=E^H(R)\cap E_i$. Recall that, if $E_i(R)\neq\emptyset$, then 
$|W_i\cap S(Q^E)|, |W_i\cap S(Q^V)|\geq z$ must hold, and, from the definition of terrible clusters, $|E_i(R)|\geq z/(8\alpha)$. It is also immediate that $|E_i(R)|\leq |E'(R)|\leq 6^{\ell}$.

We partition the set $\set{1,\ldots,6^{\ell}}$ of indices into at most $\log(|E^H(R)|)\leq \log(6^{\ell})$ classes, where index $1\leq y\leq 6^{\ell}$ belongs to class $\cset_j(R)$ iff $2^{j-1}<|E^H(R)\cap E_y|\leq 2^j$.  Then there is some index $j_R$, so that $\sum_{y\in \cset_{j_R}(R)}|E^H(R)\cap E_y|\geq |E^H(R)|/\log (6^{\ell})$. We say that $R$ \emph{chooses} the index $j_R$. Notice that:

\[ \sum_{y\in \cset_{j_R}(R)}|E^H(R)\cap E_y|\geq  \frac{|E^H(R)|}{\log(6^{\ell})}\geq \frac{ 6^{\ell}}{8\ell \alpha\log 6}.\] 

Moreover,

\begin{equation}
|\cset_{j_R}(R)|\geq\frac{|E^H(R)|}{\log(6^{\ell})\cdot 2^{j}}\geq \frac{6^{\ell}}{8\cdot 2^j\cdot \ell\alpha\log 6}. \label{eq: bound on C}\end{equation}

Let $j^*$ be the index that was chosen by at least $|\rset''|/\log(6^{\ell})$ random strings, and let $\rset^*\subseteq \rset''$ be the set of all random strings that chose $j^*$. 
We are now ready to define a collection $\hset=\set{H_1,\ldots,H_{6^{\ell}}}$ of sub-graphs of $H'$. 
We first define the sets of vertices in these subgraphs, and then the sets of edges. Choose a random ordering of the clusters $W_1,\ldots,W_{6^{\ell}}$; re-index the clusters according to this ordering. For each query $Q\in \qset^E\cup \qset^V$ with $v(Q)\in H'$, add the vertex $v(Q)$ to set $V(H_i)$, where $i$ is the smallest index for which $W_i$ contains at least $2^{j^*-1}$ vertices of $S(Q)$; if no such index $i$ exists, then we do not add $v(Q)$ to any set.

In order to define the edges of each graph $H_i$, for every random string $R=(Q^E,Q^V)\in \rset^*$, if $i\in \cset_{j^*}(R)$, and both $v(Q^E)$ and $v(Q^V)$ belong to $V(H_i)$, then we add the corresponding edge $e(R)$ to $E(H_i)$. 
This completes the definition of the family $\hset=\set{H_1,\ldots,H_{6^{\ell}}}$ of subgraphs of $H'$. We now show that the family $\hset$ of graphs has the desired properties.
It is immediate to verify that the graphs in $\hset$ are disjoint.

\begin{claim}
For each $1\leq i\leq 6^{\ell}$, $|E(H_i)|\leq |E(H')|/2^{\gamma\ell/16}$.
\end{claim}

\begin{proof}
Fix some index $1\leq i\leq 6^{\ell}$. An edge $e(R)$ may belong to $H_i$ only if $R\in \rset^*$, and $i\in \cset_{j^*}(R)$. In that case, $E_i$ contained at least $z/(8\alpha)$ edges of $E(R)$ (since $W_i$ must be heavy for $R$ and it is not terrible for $R$). Therefore, the number of edges in $H_i$ is bounded by $|E_i|\cdot 8\alpha /z\leq 8\alpha h/z=8\alpha |E(H')|/2^{\gamma\ell/8}\leq |E(H')|/2^{\gamma\ell/16}$, since $\alpha \leq 2^{\gamma\ell/32}$.
\end{proof}

\begin{claim}
$\expect{\sum_{i=1}^r|E(H_i)|}\geq \frac{|E(H')|}{128\ell^2\alpha^2\log^26}$.
\end{claim}
\begin{proof}
Recall that $|\rset^*|\geq |\rset''|/\log(6^{\ell})\geq |\rset'|/(4\alpha\log(6^{\ell}))$. We now fix $R\in \rset^*$ and analyze the probability that $e(R)\in \bigcup_{i=1}^rE(H_i)$. Assume that $R=(Q^E,Q^V)$. 
Let $J$ be the set of indices $1\leq y\leq 6^{\ell}$, such that $|W_i\cap S(Q^V)|\geq 2^{j^*-1}$. Clearly, $|J|\leq  6^{\ell}/2^{j^*-1}$, and $v(Q^V)$ may only belong to graph $H_i$ if $i\in J$. Similarly, let $J'$ be the set of indices $1\leq y\leq 6^{\ell}$, such that $|W_i\cap S(Q^E)|\geq 2^{j^*-1}$. As before, $|J'|\leq  6^{\ell}/2^{j^*-1}$, and $v(Q^E)$ may only belong to graph $H_i$ if $i\in J$. 
Observe that every index $y\in \cset_{j^*}(R)$ must belong to $J\cap J'$, and, since $j^*=j_R$, from Equation~(\ref{eq: bound on C}), $|\cset_{j^*}(R)|\geq \frac{6^{\ell}}{8\cdot 2^{j^*}\cdot \ell\alpha\log 6}$.

 Let $y\in J\cup J'$ be the first index that occurs in our random ordering. If $y\in \cset_{j^*}(R)$, then edge $e(R)$ is added to $H_y$. The probability of this happening is at least:
 
 \[\frac{|\cset_{j^*}(R)|}{|J\cup J'|}\geq \frac{6^{\ell}/(8\cdot 2^{j^*}\cdot \ell\alpha\log 6)}{2\cdot 6^{\ell}/2^{j^*-1}}=\frac{1}{32\ell\alpha \log 6}.\]
 
  Overall, the expectation of $\sum_{i=1}^r|E(H_i)|$ is at least:

\[\frac{|\rset^*|}{32\ell\alpha\log 6}\geq \frac{|\rset'|}{128\ell^2\alpha^2\log^26}= \frac{|E(H')|}{128\ell^2\alpha^2\log^26}.\]
\end{proof}
 
Denote the expectation  of $\sum_{i=1}^r|E(H_i)|$  by $\mu$, and let $c=128\log^26$, so that $\mu=|E(H')|/(c \ell^2\alpha^2)$. 
Let $\event$ be the event that $\sum_{i=1}^r|E(H_i)|\geq |E(H')|/(2c \ell^2\alpha^2)=\mu/2$.
We claim that $\event$ happens with probability at least $1/(2c\ell^2\alpha^2)$. Indeed, assume that it happens with probability $p<1/(2c\ell^2\alpha^2)$. If $\event$ does not happen, then  $\sum_{i=1}^r|E(H_i)|\leq \mu/2$, and if it happens, then $\sum_{i=1}^r|E(H_i)|\leq |E(H')|$. Overall, this gives us that $\expect{\sum_{i=1}^r|E(H_i)|} \leq (1-p)\mu/2+p|E(H')|<\mu$, a contradiction.
We repeat the algorithm for constructing $\hset$ $O(\ell^2\alpha^2\poly\log n\log P)$ times. We are then guaranteed that with probability at least $(1-1/P)$, event $\event$ happens in at least one run of the algorithm. It is easy to verify that the running time of the algorithm is bounded by $O(n^{O(\ell)}\cdot \log P)$, since $|V(L(H'))|\leq n^{O(\ell)}$.
\label{----------------------------------------sec: from WGP to NDP-------------------------------}
\section{From \WGPwB to \NDPgrid}\label{sec: from WGP to NDP}
In this section we prove Theorem~\ref{thm: from WGP to NDP}, by providing a reduction from \WGPwB to \NDPgrid.
We assume that we are given an instance $\iset=(\tG=(V_1\cup V_2,E), \uset_1,\uset_2,h, r)$ of \WGPwB.  Let $|V_1|=N_1,|V_2|=N_2$, $|E|=M$, and $N=N_1+N_2$.
We assume that  $\iset$ is a valid instance, so, if we denote by $\beta^*= \beta^*(\iset) = \sum_{v\in V_1}\beta(v)$, then $h=\beta^*/r$, and $h\geq \max_{v\in V_1\cup V_2}\set{\beta(v)}$.

We start by describing a randomized construction of the instance $\hiset=(\hG,\mset)$ of \NDPgrid. 


\label{-------------------------------------------------subsec: the construction------------------------}
\subsection{The Construction}\label{subsec: the construction}


Fix an arbitrary ordering $\rho$ of the groups in $\uset_1$. Using $\rho$, we define an ordering $\sigma$ of the vertices of $V_1$, as follows. 
The vertices that belong to the same  group $U\in \uset_1$ are placed consecutively in the ordering $\sigma$, in an arbitrary order. The ordering between the groups in $\uset_1$ is the same as their ordering in $\rho$. We assume that $V_1=\set{v_1,v_2,\ldots,v_{N_1}}$, where the vertices are indexed according to their order in $\sigma$. Next, we select a {\bf random} ordering $\rho'$ of the groups in $\uset_2$. We then define an ordering $\sigma'$ of the vertices of $V_2$ exactly as before, using the ordering $\rho'$ of $\uset_2$. We assume that $V_2=\set{v'_1,v'_2,\ldots,v'_{N_2}}$, where the vertices are indexed according to their ordering in $\sigma'$. We note that the choice of the ordering $\rho'$ is the only randomized part of our construction.

Consider some vertex $v\in V_1$. Recall that   $\bset(v)$ denotes the partition of the edges incident to $v$ into bundles, where every bundle is a  non-empty subsets of edges, and that $\beta(v)=|\bset(v)|$. Each such bundle $B\in \bset(v)$ corresponds to a single group $U(B)\in \uset_2$, and contains all edges that connect $v$ to the vertices of $U(B)$. The ordering $\rho'$ of the groups in $\uset_2$ naturally induces an ordering of the bundles in $\bset(v)$, where $B$ appears before $B'$ in the ordering iff $U(B)$ appears before $U(B')$ in $\rho'$. We denote $\bset(v) = \set{B_1(v), B_2(v), \ldots, B_{\beta(v)}(v)}$, where the bundles are indexed according to this ordering.

Similarly, for a vertex $v' \in V_2$, every bundle $B\in \bset(v')$ corresponds to a group $U(B)\in \uset_1$,  and contains all edges that connect $v'$ to the vertices of $U(B)$. As before, the ordering $\rho$ of the groups in $\uset_1$ naturally defines an ordering of the bundles in $\bset(v')$. We denote $\bset(v') = \set{B_1(v'), B_2(v'), \ldots, B_{\beta(v')}(v')}$, and we assume that the bundles are indexed according to this ordering.



We are now ready to define the instance $\hiset=(\hG,\mset)$ of \NDPgrid, from the input instance $(\tilde G=(V_1,V_2,E),\uset_1,\uset_2,h,r)$ of \WGPwB.
Let $\ell= \twiceConstantForSizeOfBlocks \cdot \ceil{M^2 \cdot \log M}$.
The graph $\hG$ is simply the $(\ell\times \ell)$-grid, so $V(\hG)=O(M^4\log^2M)$ as required. We now turn to define the set $\mset$ of the demand pairs. We first define the set $\mset$ itself, without specifying the locations of the corresponding vertices in $\hat G$, and later specify a mapping of all vertices participating in the demand pairs to $V(\hG)$.

Consider the underlying graph $\tilde G=(V_1,V_2,E)$ of the \WGPwB problem instance. Initially, for every edge $e=(u,v)\in E$, with $u\in V_1,v\in V_2$, we define a demand pair $(s(e),t(e))$ representing $e$, and add it to $\mset$, so that the vertices participating in the demand pairs are  all distinct. Next, we process the vertices $v\in V_1\cup V_2$ one-by-one. Consider first some vertex $v\in V_1$, and some bundle $B\in \bset(v)$. Assume that $B=\set{e_1,\ldots,e_z}$. Recall that for each $1\leq i\leq z$, set $\mset$ currently contains a demand pair $(s(e_i),t(e_i))$ representing $e_i$. We unify all vertices $s(e_1),\ldots,s(e_z)$ into a single vertex $s_B$. We then replace the demand pairs $(s(e_1),t(e_1)),\ldots,(s(e_z),t(e_z))$ with the demand pairs $(s_{B},t(e_1)),\ldots,(s_{B},t(e_z))$.  Once we finish processing all vertices in $V_1$, we perform the same procedure for every vertex of $V_2$: given a vertex $v'\in V_2$, 
for every bundle $B' \in \bset(v')$, we unify all destination vertices $t(e)$ with $e\in B'$ into a single destination vertex, that we denote by $t_{B'}$, and we update $\mset$ accordingly. This completes the definition of the set $\mset$ of the demand pairs.

Observe that each edge of $e\in E$ still corresponds to a unique demand pair in $\mset$, that we will denote by $(s_{B(e)}, t_{B'(e)} )$, where $B(e)$ and $B'(e)$ are the two corresponding bundles containing $e$. Given a subset $E'\subseteq E$ of edges of $\tG$, we denote by $\mset(E')=\set{(s_{B(e)},t_{B'(e)})\mid e\in E'}$ the set of all demand pairs corresponding to the edges of $E'$.

In order to complete the reduction, we need to show a mapping of all source and all destination vertices of $\mset$ to the vertices of $\hG$.
Let $R'$ and $R''$ be two rows of the grid $\hG$, lying at a distance at least $\xi/4$ from each other and from the top and the bottom boundaries of the grid. We will map all vertices of $S(\mset)$ to $R'$, and all vertices of $T(\mset)$ to $R''$.

\paragraph{Locations of the sources.}
Let $\block_1,\block_2,\ldots,\block_{N_1}$ be a collection of $N_1$ disjoint sub-paths of $R'$, where each sub-path contains {$\constantForSizeOfBlocks \cdot\ceil{ h\cdot \log M}$} vertices; the sub-paths are indexed according to their left-to-right ordering on $R'$, and every consecutive pair of the paths is separated by at least $10M$ vertices from each other and from the left and the right boundaries of $\hG$. Observe that the width $\ell$ of the grid is large enough to allow this, as $h\leq M$ must hold.
For all $1\leq i\leq N_1$,
we call $\block_i$ the \emph{block representing the vertex $v_i\in V_1$}.
We now fix some $1 \leq i \leq N_1$ and consider the block $\block_i$ representing the vertex $v_i$. We map the source vertices $s_{B_1(v_i)},s_{B_2(v_i)},\ldots,s_{B_{\beta(v_i)}(v_i)}$ to vertices of $\block_i$, so that they appear on $\block_i$ in this order, so that every consecutive pair of sources is separated by exactly {$512 \cdot \ceil{h\cdot \log M/\beta(v_i)}$} vertices.


\paragraph{Locations of the destinations.}
 Similarly, we let $\block'_1,\block'_2,\ldots,\block'_{N_2}$ be a collection of $N_2$ disjoint sub-paths of $R''$, each of which contains {$\constantForSizeOfBlocks \cdot \ceil{h\cdot \log M}$}  vertices, so that the sub-paths are indexed according to their left-to-right ordering on $R''$, and every consecutive pair of the paths is separated by at least $10M$ vertices from each other and from the left and the right boundaries of $\hG$. We call $\block'_i$ the \emph{block representing the vertex $v'_i\in V_2$}.
 We now fix some $1 \leq i \leq N_2$ and consider the block $\block'_i$ representing the vertex $v'_i$. We map the destination vertices $t_{B_1(v'_i)},t_{B_2(v'_i)},\ldots,t_{B_{\beta(v'_i)}(v'_i)}$ to vertices of $\block'_i$, so that they appear on $\block'_i$ in this order, and every consecutive pair of destinations is separated by exactly {$512 \cdot \ceil{h\cdot \log M/\beta(v'_i)}$} vertices.

 This concludes the definition of the instance $\hiset=(\hG,\mset)$ of \NDPgrid. In the following subsections we analyze its properties. 
 The following immediate observation will be useful to us.
 
 \begin{observation}\label{obs: ordering of sources and destinations}
 Consider a vertex $v_j\in V_1$, and let $\nset_j\subseteq\mset$ be any subset of demand pairs, whose sources are all distinct and lie on $\block_j$. Assume that $\nset_j=\set{(s_1,t_1),\ldots,(s_{y},t_y)}$, where the demand pairs are indexed according to the left-to-right ordering of their source vertices on $\block_j$. Then $t_1,\ldots,t_y$ appear in this left-to-right order on $R''$.
 \end{observation}

 We will also use the following two auxiliary lemmas, whose proofs are straightforward and are deferred to Section~\ref{sec: auxiliarly lemmas} of the Appendix.


\paragraph{Auxiliary Lemmas}  Assume that we are given a set $U$ of $n$ items, such that  $P$ of the items are pink, and $Y=n-P$ items are yellow. Consider a random permutation $\pi$ of these items. 

\begin{lemma} \label{lem: random ordering}
For any $\log n\leq \mu\leq Y$, the probability that there is a sequence of $\ceil{4n\mu/P}$ consecutive items in $\pi$ that are all yellow, is at most $n/e^{\mu}$.
\end{lemma}

\begin{lemma} \label{lem: random ordering2}
For any $\log n\leq \mu\leq P$, the probability that there is a set $S$ of $\floor{\frac{n\mu}{P}}$ consecutive items in $\pi$, such that more than $4\mu$ of the items are pink, is at most $n/4^{\mu}$.
\end{lemma}

\label{-------------------------------------------------subsec: YI from partitioning to routing ------------------------}
\subsection{From Partitioning to Routing}\label{subsec: YI from partitioning to routing}

The goal of this subsection is to prove the following theorem.

\begin{theorem}\label{thm: yi}
Suppose we are given a valid instance $\iset=(\tG=(V_1,V_2,E),\uset_1,\uset_2,h,r)$ of \WGPwB, such that $\iset$ has a perfect solution. Then with probability at least $1/2$ over the random choices made in the construction of the corresponding instance $\hiset$ of \NDPgrid, there is a solution to $\hiset$, routing $\Omega(\beta^*(\iset)/\log^3M)$ demand pairs via a set of spaced-out paths.
\end{theorem}

The remainder of this subsection is devoted to the proof of the theorem. We assume w.l.o.g. that $|E|=M > 2^{50}$, as otherwise, since $\beta^*(\iset)\leq |E|$, routing a single demand pair is sufficient.

Let  $((W_1,\ldots,W_r),(E_1,\ldots,E_r))$ be a perfect solution to $\iset$.
Recall that by the definition of a perfect solution, for each group $U\in \uset_1 \cup \uset_2$, every set $W_ i$ contains exactly one vertex of $U$, and moreover, for each $1\leq i\leq r$, $|E_i|=h=\beta^*(\iset)/r$.




We let $E^0=\bigcup_{i=1}^rE_i$, so $|E^0|=\beta^*(\iset) = hr$. Let $\mset^0\subseteq \mset$ be the set of all demand pairs corresponding to the edges of $E^0$. Note that we are guaranteed that no two demand pairs in $\mset^0$ share a source or a destination, since no two edges of $E^0$ belong to the same bundle.


Next we define a property of subsets of the demand pairs, called a \emph{distance property}. We later show that every subset $\mset'\subseteq \mset^0$ of demand pairs that has this property can be routed via spaced-out paths, and that there is a large subset $\mset'\subseteq \mset^0$ of the demand pairs with this property.

Given a subset $\mset'\subseteq \mset^0$ of the demand pairs, we start by defining an ordering $\sigma_{\mset'}$ of the destination vertices in $T(\mset')$. This ordering is somewhat different from the ordering of the vertices of $T(\mset')$ on row $R''$. We first provide a motivation and an intuition for this new ordering $\sigma_{\mset'}$. Recall that the rows $R'$ and $R''$ of $\hat G$, where all source and all destination vertices lie, respectively, are located at a distance at least $\ell/4$ from each other and from the grid boundaries. Let $R$ be any row of $\hat G$, lying between $R'$ and $R''$, at a distance at least $\ell/16$ from both $R'$ and $R''$. Let $X$ be some subset of $|\mset'|$ vertices of $R$. If we index the vertices of $T(\mset')$ as $\set{t_1,t_2,\ldots,t_{|\mset'|}}$ according to their order in the new ordering $\sigma_{\mset'}$, then we view the $i$th vertex of $X$, that we denote by $x_i$, as representing the terminal $t_i$. For each $1\leq i\leq |\mset'|$,  we denote the source vertex corresponding to $t_i$ by $s_i$, that is, $(s_i,t_i)\in \mset'$. Note that the ordering of the vertices of $S(\mset')$ on $R'$ may be completely different from the one induced by these indices.  Similarly, the ordering of the vertices of $T(\mset')$ on $\rset''$ may be inconsistent with this indexing. Eventually, we will construct a set $\pset$ of spaced-out paths routing the demand pairs in $\mset'$, so that the path $P_i\in \pset$, connecting $s_i$ to $t_i$, intersects the row $R$ exactly once -- at the vertex $x_i$. In this way, we will use the ordering $\sigma_{\mset'}$ of the destination vertices in $T(\mset')$ to determine the order in which the path of $\pset$ intersect $R$.

Assume now that we are given some subset $\mset'\subseteq \mset^0$ of demand pairs. Recall that the sources and the destinations of all demand pairs in $\mset'$ are distinct. We are now ready to define the ordering $\sigma_{\mset'}$ of $T(\mset')$. We partition the vertices of $T(\mset')$ into subsets $J_1,J_2,\ldots,J_r$, as follows. 
Consider some vertex $v'_j\in V_2$ of $\tilde G$, and assume that it lies in the cluster $W_i$. Then all destination vertices of $T(\mset')$ that belong to the corresponding block $K'_j$ are added to the set $J_i$. 
 To obtain the final ordering $\sigma_{\mset'}$, we place the vertices of  $J_1,J_2,\ldots,J_r$ in this order, where within each set $J_i$, the vertices are ordered  according to their ordering along the row $R''$. Notice that a selection of a subset $\mset'\subseteq \mset_0$ completely determines the ordering $\sigma_{\mset'}$. Given two demand pairs $(s,t),(s',t')\in \mset'$, we let $N_{\mset'}(s,s')$ denote the number of destination vertices that lie between $t$ and $t'$ in the ordering $\sigma_{\mset'}$ (note that this is well-defined as the demand pairs in $\mset'$ do not share their sources or destinations). Recall that $d(s,s')$ is the distance between $s$ and $s'$ in graph $\hat G$.

\begin{definition} Suppose we are given a subset $\mset'\subseteq \mset^0$ of the demand pairs. 
We say that two distinct vertices $s,s'\in S(\mset')$ are \emph{consecutive} with respect to $\mset'$, iff no other vertex of $S(\mset')$ lies between $s$ and $s'$ on $R'$. We say that $\mset'$ has the \emph{distance property} iff for every pair $s,s'\in S(\mset')$ of vertices that are consecutive with respect to $\mset'$, $N_{\mset'}(s,s')<d(s,s')/4$.
\end{definition}

We first show that there is a large subset of the demand pairs in $\mset^0$ with the distance property in the following lemma, whose proof appears in the next subsection.

\begin{lemma}\label{lem: large set with distance property}
With probability at least $1/2$ over the construction of $\hat \iset$, there is a subset $\mset'\subseteq \mset^0$ of demand pairs that has the distance property, and $|\mset'|=\Omega(|\mset^0|/\log^3M)$.
\end{lemma}

Finally, we show that every set $\mset'$ of demand pairs with the distance property can be routed via spaced-out paths. 


\begin{lemma}\label{lem: can find routing}
Assume that $\mset'\subseteq \mset^0$ is a subset of demand pairs that has the distance property. 
Then there is a spaced-out set $\pset$ of paths routing all pairs of $\mset'$ in graph $\hat G$.
\end{lemma}

The above two lemmas finish the proof of Theorem~\ref{lem: large set with distance property}, since $|\mset^0|=\beta^*(\iset)$. We prove these lemmas in the following two subsections.

\subsubsection{Proof of Lemma~\ref{lem: large set with distance property}}

We assume that $|\mset^0|>c\log^3M$ for some large enough constant $c$, since otherwise we can return a set $\mset'$ containing a single demand pair. 
We gradually modify the set $\mset^0$ of the demand pairs, by selecting smaller and smaller subsets $\mset^1,\mset^2$, $\mset^3$, and $\mset'$. 
For each vertex vertex $v \in V_1 \cup V_2$ of the \WGPwB instance $\tG$, let $\delta(v)$ denote the set of all edges of $E(\tG)$ incident to $v$.

We start by performing two ``regularization'' steps on the vertices of $V_2$ and $V_1$ respectively. Intuitively, we will select two integers $p$ and $q$, and a large enough subset $\mset^2\subseteq \mset^0$ of demand pairs, so that for every vertex $v_i\in V_1$, either no demand pair in $\mset^2$ has its source on $\block_i$, or roughly $2^p$ of them do. Similarly, for every vertex $v'_j\in V_2$, either no demand pairs in $\mset^2$ has its destination on $K'_j$, or roughly $2^q$ of them do. We will not quite achieve this, but we will come close enough.

\paragraph{Step 1 [Regularizing the degrees in $V_2$].} In this step we select a large subset $\mset^1\subseteq \mset^0$ of the demand pairs, and an integer $q$, such that, for each vertex $v\in V_2$, the number of edges of $E(\tG)$ incident to $v$, whose corresponding demand pair lies in $\mset^1$, is either $0$, or roughly $2^q$. In order to do this,
we partition the vertices of $V_2$ into classes $Z_1,\ldots,Z_{\ceil{\log M}}$, where a vertex $v' \in V_2$ belongs to class $Z_y$ iff $2^{y-1}\leq |E^0 \cap \delta(v')|< 2^y$. If $v\in Z_y$, then we say that all edges in $\delta(v)\cap E^0$ belong to the class $Z_y$.
Therefore, each edge of $E^0$ belongs to exactly one class, and there is some index $1\leq q\leq \ceil{\log M}$, such that at least $\Omega(|\mset^0|/\log M)$ edges of $E^0$ belong to class $Z_{q}$. We let $E^1\subseteq E^0$ be the set of all edges that belong to the class $Z_q$, and we let $\mset^1 = \mset(E^1) \subseteq \mset^0$ be the corresponding subset of the demand pairs.


\paragraph{Step 2 [Regularizing the degrees in $V_1$].} This step is similar to the previous step, except that it is now performed on the vertices of $V_1$.
We partition the vertices of $V_1$ into classes $Y_1,\ldots,Y_{\ceil{\log M}}$, where a vertex $v \in V_1$ belongs to class $Y_z$ iff $2^{z-1}\leq |E^1 \cap \delta(v)|<2^z$. If $v\in Y_z$, then we say that all edges in $\delta(v)\cap E^1$ belong to the class $Y_z$.
As before, every edge of $E^1$ belongs to exactly one class, and there is some index $1\leq p\leq \ceil{\log M}$, such that at least $\Omega(|E^1|/\log M)\geq \Omega(|\mset^0|/\log^2 M)$ edges of $E^1$ belong to the class $Y_{p}$. We let $E^2\subseteq E^1$ denote the set of all edges that belong to class $Y_p$, and $\mset^2 = \mset(E^2) \subseteq \mset^1$ denote the corresponding subset of the demand pairs, so that $|\mset^2|=\Omega(|\mset^0|/\log^2M)$.


Notice that so far, for every vertex $v\in V_1$, if $\delta(v)\cap E^2\neq \emptyset$, then $2^{p-1}\leq |\delta(v)\cap E^2|<2^{p}$. However, for a vertex $v\in V_2$ with $\delta(v)\cap E^2\neq \emptyset$, we are only guaranteed that $|\delta(v)\cap E^2|<2^{q}$, since we may have discarded some edges that were incident to $v$ from $E^1$. Moreover, the subset $\mset^2$ of the demand pairs is completely determined by the solution $((W_1,\ldots,W_r),(E_1,\ldots,E_r))$ to the \WGPwB problem, and is independent of the random choices made in our construction of the \NDPgrid problem instance.
The following observation will be useful for us later.

\begin{observation}\label{obs: h is large}
$h\geq 2^{p}\cdot 2^{q}/4$.
\end{observation}

\begin{proof}
From our assumption that $|\mset^0|>c\log^3M$, $E^2\neq \emptyset$. Therefore, there must be an index $1\leq i\leq r$ with $E_i\cap E^2\neq \emptyset$. Fix any such index $i$. Then there is a vertex $v'_j\in W_i\cap V_2$, such that at least one edge of $\delta(v'_j)$ belongs to $E^2$. But then, from the definition of $E^2$, at least $2^{q-1}$ edges of $\delta(v'_j)$ belong to $E^2$. Assume without loss of generality that these edges connect $v'_j$ to vertices $v_1,\ldots,v_{2^{q-1}}\in V_1$. All these vertices must also belong to $W_i$, and  for each $1\leq x\leq 2^{q-1}$, vertex $v_x$ has at least one edge in $\delta(v_x)\cap E^2$. From our definition of $E^2$ and $E^1$, at least $2^{p-1}$ edges of $\delta(v_x)$ belonged to $E^1$. Therefore, $|E^1\cap E_i|\geq 2^{q-1}\cdot 2^{p-1}$. But $|E_i|\leq h$, and so $h\geq 2^p\cdot 2^q/4$.
\end{proof}

For each vertex $v_i\in V_1$, let $X_i\subseteq V(\block_i)$ be the set of all vertices that serve as the sources of the demand pairs in $\mset$, so $X_i=S(\mset)\cap V(\block_i)$. Recall that $|X_i|=\beta(v_i)$, and every pair of vertices in $X_i$ is separated by at least $512\ceil{\frac{h\log M}{\beta(v_i)}}$ vertices of $\block_i$. We let $X'_i\subseteq X_i$ denote the subset of vertices that serve as sources of the demand pairs in $\mset^2$. We say that a sub-path $Q\subseteq \block_i$ is \emph{heavy} iff $|V(Q)|=\floor{\frac{512 h\log^2M}{2^p}}$, and $|V(Q)\cap X'_i|>16\log M$.

\begin{observation}\label{obs: no heavy path exist}
With probability at least $0.99$ over the choice of the random permutation $\rho'$, for all $v_i\in V_1$, no heavy sub-path $Q\subseteq \block_i$ exists.
\end{observation}

\begin{proof}
Since $|V_1|\leq M$, from the union bound, it is enough to prove that for a fixed vertex $v_i\in V_1$, the probability that a heavy sub-path $Q\subseteq \block_i$ exists is at most $1/(100M)$. We now fix some vertex $v_i\in V_1$. Observe that $\block_i$ may only contain a heavy sub-path if $\beta(v_i)=|X_i|\geq 16\log M$. We call the vertices of 
$X'_i$ pink, and the remaining vertices of $X_i$ yellow. Let $P$ denote the number of the pink vertices. Then $2^{p-1}\leq P<2^p$. Let $\mu=4\log M$, and let $\event_i$ be the bad event that there is a set of $\floor{|X_i|\mu/P}$ consecutive vertices of $X_i$, such that at least $4\mu$ of them are pink. 

Observe that the selection of the pink vertices only depends on the solution to the \WGPwB problem, and is independent of our construction of the \NDPgrid instance. 
The ordering of the vertices in $X_i$ is determined by the permutation $\rho'$ of $\uset_2$, and is completely random.
Therefore, from Lemma~\ref{lem: random ordering2}, the probability of $\event_i$ is at most $|X_i|/4^{\mu}\leq M/4^{4\log M}\leq 1/M^7$. 

Let $\event'_i$ be the event that some sub-path $Q$ of $\block_i$ is heavy. We claim that $\event'_i$ may only happen if event $\event_i$ happens. Indeed, consider some sub-path $Q$ of $\block_i$, and assume that it is heavy. Recall that $Q$ contains $\floor{\frac{512 h\log^2M}{2^p}}$ vertices. Since every pair of vertices in $X_i$ is separated by at least $512\ceil{\frac{h\log M}{\beta(v_i)}}$ vertices, we get that:

\[ |V(Q)\cap X_i| \leq  \frac{\floor{512 h\log^2M/2^p}}{512\ceil {h\log M/\beta(v_i)}}+1
\leq \frac{\beta(v_i)\log M}{2^p}+1
\leq \floor{\frac{|X_i|\mu} P},
\]

 as $2^{p-1}\leq P<2^p$. 
Since $Q$ is heavy, at least $4\mu=16\log M$ of the vertices of $V(Q)\cap X_i$ belong to $X'_i$, that is, they are pink. Therefore, there is a set of  $\floor{|X_i|\mu/P}$ consecutive vertices of $X_i$, out of which $4\mu$ are pink, and $\event_i$ happens. We conclude that $\prob{\event'_i}\leq \prob{\event_i}\leq 1/M^7$, and overall, since we have assumed that $M>2^{50}$, the probability that a heavy path exists in any block $\block_i$ is bounded by $0.99$ as required.
\end{proof}

Let $\badevent$ be the bad event that for some $v_i\in V_1$, block $\block_i$ contains a heavy path. From Observation~\ref{obs: no heavy path exist}, the probability of $\badevent$ is at most $0.01$.

The following claim will be used to bound the values $N_{\mset^2}(s,s')$.

\begin{claim}\label{claim: dist prop}
Consider some vertex $v_j\in V_1$ in graph $\tilde G$, and the block $\block_j$ representing it.
Then with probability at least $(1-1/M^3)$, for every pair $s,s'\in S(\mset^2)\cap V(\block_j)$ of source vertices that are consecutive with respect to $\mset^2$, $N_{\mset^2}(s,s')\leq 128h\log M/2^p$.
\end{claim}


\begin{proof}
Fix some vertex $v_j\in V_1$ and consider the block $\block_j$ representing it. Assume that $v_j$ belongs to the cluster $W_i$ in our solution to the \WGPwB problem.
Let $s,s'\in S(\mset^2)\cap V(\block_j)$ be a pair of source vertices  that are consecutive with respect to $\mset^2$.  Recall that we have defined a subset $J_i\subseteq T(\mset^2)$ of destination vertices, that appear consecutively in the ordering $\sigma_{\mset^2}$, and contain all vertices of $T(\mset^2)$, that lie in blocks $\block'_{j'}$, whose corresponding vertices $v_{j'}\in V_2\cap W_i$.


Let $A= W_i\cap V_2$; let $A'\subseteq A$ contain all vertices that have an edge of $E^2\cap E_i$ incident to them; and let $A''\subseteq A'$ contain all vertices that have an edge of $E^2\cap E_i$ connecting them to $v_j$. Since the solution to the \WGPwB problem instance is perfect, every vertex of $A$ (and hence $A'$ and $A''$) belongs to a distinct group of $U\in \uset_2$. We denote by $\uset'\subseteq \uset_2$ the set of all groups to which the vertices of $A'$ belong, and we define $\uset''\subseteq \uset'$ similarly for $A''$. Consider now some group $U\in \uset'$, and let $v'_a$ be the unique vertex of $U$ that belongs to $A'$. We denote by $C(U)$ the set of all vertices of the corresponding block $\block'_a$ that belong to $T(\mset^2)$. Therefore, we now obtain a partition $\set{C(U)}_{U\in \uset'}$ of all vertices of $J_i$ into subsets, where each subset contains at most $2^q$ vertices. Moreover, in the ordering $\sigma_{\mset^2}$, the vertices of each such set $C(U)$ appear consecutively, in the order of their appearance on $R''$, while the ordering between the different sets $C(U)$ is determined by the ordering of the corresponding groups $U$ in $\rho'$.  Let $\rho''$ be the ordering of the groups in $\uset'$ induced by $\rho'$, so that $\rho''$ is a random ordering of $\uset'$. Observe that, since the choice of the set $\mset^2$ is independent of the ordering $\rho'$ (and only depends on the solution to the \WGPwB problem instance), so is the choice of the sets $\uset'$ and $\uset''$.

Let $t$ and $t'$ be the destination vertices that correspond to $s$ and $s'$, respectively, that is, $(s,t),(s',t')\in \mset^2$. Assume that $t\in \block'_z$ and $t'\in \block'_{z'}$, where $v'_z$ and $v'_{z'}$ are vertices of $V_2$. From our definition, both $v'_z$ and $v'_{z'}$ must belong to the set $A''$. Assume that $v'_z$ belongs to the group $U'$ in $\uset_2$, while $v'_{z'}$ belongs to group $U''$. Again, from our definitions, both $U',U''\in \uset''$. From the above discussion, if the number of groups $U\in \uset'$ that fall between $U'$ and $U''$ is $\gamma$, then the number of destination vertices lying between $t$ and $t'$ in $\sigma_{\mset^2}$ is at most $2^q\cdot (\gamma+2)$. Therefore, it is now enough to bound the value of $\gamma$. In order to do so, we think of the groups of $\uset''$ as pink, and the remaining groups of $\uset'$ as yellow.
Let $P$ denote the total number of all pink groups, and let $n^*=|\uset'|$. From the construction of $\mset^2$, $P=|S(\mset^2)\cap V(\block_j)|\geq 2^{p-1}$. We use the following observation to upper-bound $n^*$.

\begin{observation} \label{obs: nset is small}
$n^* \leq h/2^{q-1}$.
\end{observation}
\begin{proof}
Since we have started with a perfect solution to $\iset$, for each group $U \in \uset'$, there is exactly one vertex of $U$ in $W_i$.
Due to Step 1 of regularization, each such vertex contributed at least $2^{q-1}$ edges to $E^1\cap E_i$, while $|E^1\cap E_i|\leq h$.
Therefore, $n^*\leq h/2^{q-1}$.
\end{proof}

Let $\mu=4\log M$, and let $\event_Y$ be the event that there are at least  $\ceil{4n^*\mu/P}$ consecutive yellow groups in the ordering $\rho''$ of $\uset'$. From Lemma~\ref{lem: random ordering}, the probability of $\event_Y$ is at most  $n^*/e^{\mu}\leq M/e^{\mu}\leq 1/M^3$.  If event $\event_Y$ does not happen, then  the length of the longest consecutive sub-sequence of $\rho''$ containing only yellow groups is bounded by:

\[\ceil{\frac{4n^*\mu}{P}}\leq \frac{64h\log M}{2^{q} 2^p}.\]

Assume now that $\event_Y$ does not happen, and consider any two vertices $s,s'\in  S(\mset^2)\cap V(\block_j)$ that are consecutive with respect to  $\mset^2$. Assume that their corresponding destination vertices are $t$ and $t'$ respectively, and that $t$ and $t'$ belong to the groups $U$ and $U'$, respectively. Then $U$ and $U'$ are pink groups. Moreover, since the ordering of the vertices of $S(\mset)\cap V(\block_j)$ on $R'$ is identical to the ordering of the groups of $\uset''$ to which their destinations belong in $\rho'$, no other pink group appears between $U$ and $U'$ in $\rho''$. Therefore, at most $\frac{64 h\log M}{2^{q} 2^p}$ groups of $\uset'$ lie between $U$ and $U'$ in $\rho''$. Recall that for each group $U''\in \uset'$, at most one vertex $v'_z\in U''$ belongs to $W_i$, and that the corresponding block $\block'_z$ may contribute at most $2^q$ destination vertices to $T(\mset^2)$. We conclude that the number of vertices separating $t$ from $t'$ in $\sigma_{\mset_3}$ is bounded by:

\[\left(\frac{64h\log M}{2^{q} 2^p}+2\right )\cdot 2^q\leq \frac{128h\log M}{2^p}.\]

(We have used Observation~\ref{obs: h is large}).
Therefore, if event $\event_Y$ does not happen, then  for every pair $s,s'\in  S(\mset^2)\cap V(\block_j)$ of vertices that are consecutive with respect to $\mset^2$, $N_{\mset^2}(s,s')\leq  \frac{128h\log M}{2^p}$.
\end{proof}

Let $\badevent'$ be the bad event that for some vertex $v_j\in V_1$, for some pair $s,s'\in S(\mset^2)\cap V(\block_j)$ of source vertices that are consecutive with respect to $\mset^2$, $N_{\mset^2}(s,s')> 128h\log M/2^p$. By applying the Union Bound to the result of Claim~\ref{claim: dist prop}, we get that the probability of $\badevent'$ is at most $1/M^2$. 
Notice that the probability that neither $\badevent$ nor $\badevent'$ happen is at least $1/2$. We assume from now on that this is indeed the case, and show how to compute a large subset $\mset'\subseteq \mset^2$ of the demand pairs that has the distance property. This is done in the following two steps.

\paragraph{Step 3 [Sparsifying the Sources.]}

Assume that $\mset^2=\set{(s_1,t_1),\ldots,(s_{|\mset^2|},t_{|\mset^2|})}$, where the demand pairs are indexed according to the left-to-right ordering of their sources on $R'$, that is $s_1,s_2,\ldots,s_{|\mset^2|}$ appear on $R'$ in this order. We now define:

\[\mset^3=\set{(s_i,t_i)\mid i=1\mod \ceil{32\log M}}.\]

Let $E^3\subseteq E^2$ be the set of edges of $\tG$ whose corresponding demand pairs belong to $\mset^3$.
It is easy to verify that $|\mset^3|\geq \Omega(|\mset^2|/\log M)=\Omega(|\mset^0|/\log^3 M)$.

 We also obtain the following claim.
 
\begin{claim}\label{claim: almost distance property}
Assume that events $\badevent$ and $\badevent'$ did not happen. Then
for each $1\leq j\leq N_1$, for every pair  $s,s'\in S(\mset^3)\cap V(\block_j)$ of source vertices, $N_{\mset^3}(s,s')\leq 128 d(s,s')$.
\end{claim}
\begin{proof}
Fix some $1\leq j\leq N_1$, and some pair $s,s'\in S(\mset^3)\cap V(\block_j)$ of vertices. Let $S'=\set{s_1,s_2,\ldots,s_z}$ be the set of all vertices of $S(\mset^2)$ that appear between $s$ and $s'$ on $R'$. Assume w.l.o.g. that $s$ lies to the left of $s'$ on $R'$, and denote $s_0=s$ and $s_{z+1}=s'$. Assume further that the vertices of $S'$ are indexed according to their left-to-right ordering on $R'$. Note that, from the definition of $\mset^3$, $z\geq \ceil{32\log M}$ must hold. 

Let $I\subseteq \block_j$ be the sub-path of $\block_j$ between $s$ and $s'$. We partition $I$ into paths containing $\floor{512 h\log^2M/2^p}$ vertices each, except for the last path that may contain fewer vertices. Since no such path may be heavy, we obtain at least $n'=\floor{\frac{z}{16\log M}}$ disjoint sub-paths of $I$, each of which contains $\floor{512 h\log^2M/2^p}$ vertices. We conclude that:

\[d(s,s')\geq n'\cdot \floor{512 h\log^2M/2^p}-1\geq \floor{\frac{z}{16\log M}}\cdot \floor{512 h\log^2M/2^p}-1\geq \frac{2zh\log M}{2^p}.\]

On the other hand, since $\sigma_{\mset^3}$ is the same as the ordering of $T(\mset^3)$ induced by $\sigma_{\mset^2}$, we get that:

\[N_{\mset^3}(s,s')\leq N_{\mset^2}(s,s')\leq \sum_{i=0}^zN_{\mset^2}(s_i,s_{i+1})\leq (z+1)\cdot \frac{128 h\log M}{2^p}\leq 128 d(s,s').\]
\end{proof}

\paragraph{Step 4 [Decreasing the Values $N_{\mset^3}(s,s')$].}
We are now ready to define the final set $\mset'\subseteq \mset^3$ of demand pairs. 

Assume that $\mset^3=\set{(s_1,t_1),\ldots,(s_{|\mset^3|},t_{|\mset^3|})}$, where the demand pairs are indexed according to the ordering of their destinations in $\sigma_{\mset^3}$ (notice that this is the same as the ordering of $T(\mset^3)$ induced by $\sigma_{\mset^2}$). We now define:

\[\mset'=\set{(s_i,t_i)\mid i=1\mod 512}.\]

We claim that $\mset'$ has the distance property. Indeed, consider any two pairs $(s,t),(s',t')\in \mset'$, such that $s$ and $s'$ are consecutive with respect to $\mset'$. If $s$ and $s'$ lie in different blocks $\block_j$, then, since the distance between any such pair of blocks is at least $10M$, while $|\mset'|\leq M$, we get that $N_{\mset'}(s,s')\leq d(s,s')/4$. Otherwise, both $s$ and $s'$ belong to the same block $\block_j$. But then it is easy to verify that $N_{\mset'}(s,s')\leq N_{\mset^3}(s,s')/512\leq d(s,s')/4$ from Claim~\ref{claim: almost distance property}.
We conclude that with probability at least $1/2$, neither of the events $\badevent$, $\badevent'$ happens, and in this case, $\mset'$ has the distance property.


\subsubsection{Proof of Lemma~\ref{lem: can find routing}}

Recall that $R'$ and $R''$ are the rows of $\hat G$ containing the vertices of $S(\mset)$ and $T(\mset)$ respectively, and that $R'$ and $R''$ lie at distance at least $\ell/4$ form each other and from the top and the bottom boundaries of $\hat G$, where $\ell>M^2$ is the dimension of the grid. Let $R$ be any row lying between $R'$ and $R''$, within distance at least $\ell/16$ from each of them.

We denote $M'=|\mset'|$, and $\mset'=\set{(s_1,t_1),\ldots,(s_{M'},t_{M'})}$, where the pairs are indexed according to the ordering of their destination vertices in $\sigma_{\mset'}$. To recap, the ordering $\sigma_{\mset'}$ of $T(\mset')$ was defined as follows. We have defined a partition $(J_1,\ldots,J_r)$ of the vertices of $T(\mset')$ into subsets, where each set $J_i$ represents a cluster $W_i$ in our solution to the \WGPwB problem instance, and contains all destination vertices $t\in T(\mset')$ that lie in blocks $\block'_j$, for which the corresponding vertex $v_j\in V_2$ belongs to $W_i$. The ordering of the destination vertices inside each set $J_i$ is the same as their ordering on $R''$, and the different sets $J_1,\ldots,J_r$ are ordered in the order of their indices.

Let $X=\set{x_i\mid 1\leq i\leq M'}$ be a set of vertices of $R$, where $x_i$ is the $(2i)$th vertex of $R$ from the left. We will construct a set $\pset^*$ of spaced-out paths routing all demand pairs in $\mset'$, so that the path $P_i$ routing $(s_i,t_i)$ intersects the row $R$ at exactly one vertex --- the vertex $x_i$.
Notice that row $R$ partitions the grid $\hat G$ into two sub-grids: a top sub-grid $G^t$ spanned by all rows that appear above $R$ (including $R$), and a bottom sub-grid $G^b$ spanned by all rows that appear below $R$ (including $R$).

It is now enough to show that there are two sets of spaced-out paths: set $\pset^1$ routing all pairs in $\set{(s_i,x_i)\mid 1\leq i\leq M'}$ in the top grid $G^t$, and  set $\pset^2$ routing all pairs in $\set{(t_i,x_i)\mid 1\leq i\leq M'}$ in the bottom grid $G^b$, so that both sets of paths are internally disjoint from $R$.

\paragraph{Routing in the Top Grid.}Consider some vertex $v_j\in V_1$, and the corresponding block $\block_j$. We construct a sub-grid $\hat \block_j$ of $\hat G$, containing $\block_j$,  that we call a box, as follows. Let $\cset_j$ be the set of all columns of $\hat G$ intersecting $\block_j$. We augment $\cset_j$ by adding $2M$ columns lying immediately to the left and $2M$ columns lying immediately to the right of $\cset_j$, obtaining a set $\hat{\cset}_j$ of columns. Let $\hat \rset$ contain three rows: row $R'$; the row lying immediately above $R'$; and the row lying immediately below $R'$. Then box $\hat \block_j$ is the sub-grid of $\hat G$ spanned by the rows in $\hat \rset$ and the columns in $\hat {\cset}_j$. Since every block is separated by at least $10M$ columns from every other block, as well as the left and the right boundaries of $\hat G$, the resulting boxes are all disjoint, and every box is separated by at least $2M$ columns of $\hat G$  from every other box, and from the left and the right boundaries of $\hat G$. 

We will initially construct a set $\pset^1$ of spaced-out paths in $G^t$, such that each path $P_i\in \pset^1$, for $1\leq i\leq M'$, originates from the vertex $x_i$, and visits the boxes $\hat \block_1,\hat \block_2,\ldots,\hat \block_{N_1}$ in turn. We will ensure that each such path $P_i$ contains the corresponding source vertex $s_i$. Eventually, by suitably truncating each such path $P_i$, we will ensure that it connects $x_i$ to $s_i$.

\begin{claim}\label{claim: selecting columns}
Consider some vertex $v_j\in V_1$, and the corresponding box $\hat \block_j$. Denote $Y_j=S(\mset')\cap V(\block_j)$, $M_j=|Y_j|$, and assume that $Y_j=\set{s_{i_1},s_{i_2},\ldots,s_{i_{M_j}}}$, where the indexing of the vertices of $Y_j$ is consistent with the indexing of the vertices of $S(\mset')$ that we have defined above, and $i_1<i_2<\ldots<i_{M_j}$. 
Then there is a set $\wset_j$ of $M'$  columns of the box $\hat \block_j$, such that:

\begin{itemize}
\item set $\wset_j$ does not contain a pair of consecutive columns; and
\item for each $1\leq z\leq M_j$, the $i_z$th column of $\wset_j$ from the left contains the source vertex $s_{i_z}$.
\end{itemize}
\end{claim}

\begin{proof}
Observe that from Observation~\ref{obs: ordering of sources and destinations}, the vertices $s_{i_1},s_{i_2},\ldots,s_{i_{M_j}}$ must appear in this left-to-right order on $R'$, while the vertices $x_{i_1},x_{i_2},\ldots,x_{i_{M_j}}$ appear in this left-to-right order on $R$. Moreover, for all $1\leq z<M_j$, $i_{z+1}-i_z-1=N_{\mset'}(s_{i_z},s_{i_{z+1}})\leq d(s_{i_z},s_{i_{z+1}})/4$. We add to $\wset_j$ all columns of $\hat \block_j$ where the vertices of $Y_j$ lie. For each $1\leq z<M_j$, we also add to $\wset_j$ an arbitrary set of $(i_{z+1}-i_z-1)$ columns lying between the column of $s_{i_z}$ and the column of $s_{i_{z+1}}$, so that no pair of columns in $\wset_j$ is consecutive. Finally, we add to $\wset_j$ $(i_1-1)$ columns that lie to the left of the column of $s_{i_1}$, and $(M'-i_{M_j})$ columns that lie to the right of the column of $s_{i_{M_j}}$. We make sure that no pair of columns in $\wset_j$ is consecutive -- it is easy to see that there are enough columns to ensure that. 
\end{proof}

\begin{claim}\label{claim: the paths}
There is a set $\pset^1=\set{P_1,\ldots,P_{M'}}$ of spaced-out paths in $G^t$, that are internally disjoint from $R$, such that for each $1\leq i\leq M'$, path $P_i$ originates from vertex $x_i$, and for all  $1\leq j\leq N_1$, it contains the $i$th column of $\wset_j$; in particular, it contains $s_i$.
\end{claim}

We defer the proof of the claim to Section \ref{appdx: routing in yi} of the appendix; see Figure~\ref{fig: routing} for an illustration of the routing.

\begin{figure}[h]
\centering
\includegraphics[width = 12cm]{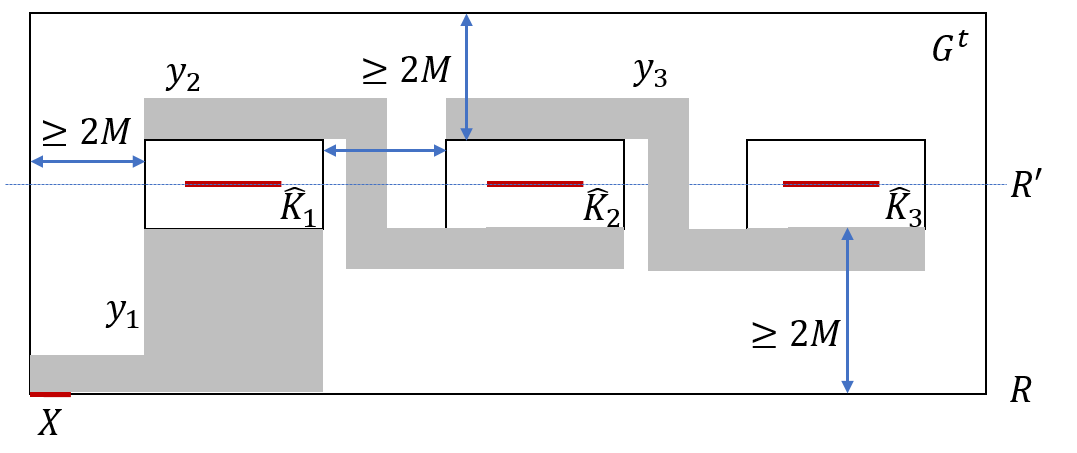} \label{fig: global-routing}
\caption{Routing in the top grid.}
\label{fig: routing}
\end{figure}

\paragraph{Routing in the Bottom Grid.}

Consider some vertex $v_j'\in V_2$, and the corresponding block $\block'_j$. We construct a box $\hat \block'_j$ containing $\block'_j$ exactly as before. As before, the resulting boxes are all disjoint, and every box is separated by at least $2M$ columns of $\hat G$  from every other box, and from the left and the right boundaries of $\hat G$. 

As before, we will initially construct a set $\pset^2$ of spaced-out paths in $G^b$, such that each path $P'_i\in \pset^2$, for $1\leq i\leq M'$, originates from the vertex $x_i$, and visits the boxes $\hat \block'_1,\hat \block'_2,\ldots,\hat \block'_{N_2}$ in turn. We will ensure that each such path $P_i'$ contains the corresponding destination vertex $t_i$. Eventually, by suitably truncating each such path $P'_i$, we will ensure that it connects $x_i$ to $t_i$.

\begin{claim}\label{claim: selecting columns2}
Consider some vertex $v_j'\in V_2$, and the corresponding box $\hat \block_j'$. Denote $Y_j'=T(\mset')\cap V(\block'_j)$, $M'_j=|Y'_j|$, and assume that $Y'_j=\set{t_{i_1},t_{i_2},\ldots,t_{i_{M_j}}}$, where the indexing of the vertices of $Y'_j$ is consistent with the indexing of the vertices of $T(\mset')$ that we have defined above, and $i_1<i_2<\ldots<i_{M'_j}$. 
Then there is a set $\wset'_j$ of $M'$  columns of the box $\hat \block_j'$, such that:

\begin{itemize}
\item set $\wset_j'$ does not contain a pair of consecutive columns; and
\item for each $1\leq z\leq M_j'$, the $i_z$th column of $\wset'_j$ from the left contains the destination vertex $t_{i_z}$.
\end{itemize}
\end{claim}

\begin{proof}
From our construction of $\sigma_{\mset'}$, the vertices $t_{i_1},t_{i_2},\ldots,t_{i_{M_j}}$ appear consecutively in this order in $\sigma_{\mset'}$. Moreover, from the construction of $\mset$, every pair of these destination vertices is separated by at least one vertex. We add to $\wset_j'$ all columns in which the vertices $t_{i_1},t_{i_2},\ldots,t_{i_{M_j}}$ lie. We also add to $\wset_j'$ $i_1-1$ columns that lie to the left of the column of $t_{i_1}$, and $M'-i_{M_j'}$ columns that lie to the right of the column of $t_{i_{M_j}}$ in $\hat \block_j'$. We make sure that no pair of columns in $\wset_j'$ is consecutive -- it is easy to see that there are enough columns to ensure that. 
\end{proof}

The proof of the following claim is identical to the proof of Claim~\ref{claim: the paths} and is omitted here.
\begin{claim}\label{claim: the paths2}
There is a set $\pset^2=\set{P_1',\ldots,P_{M'}'}$ of spaced-out paths in $G^b$, that are internally disjoint from $R$, such that for each $1\leq i\leq M'$, path $P_i'$ originates from vertex $x_i$, and for all  $1\leq j\leq N_2$, it contains the $i$th column of $\wset_j'$; in particular it contains $t_i$
\end{claim}

By combining the paths in sets $\pset^1$ and $\pset^2$, we obtain a new set $\pset^*=\set{P_1^*,\ldots,P_{M'}^*}$ of spaced-out paths, such that for all $1\leq i\leq M'$, path $P_i^*$ contains $s_i,x_i$ and $t_i$. By suitably truncating each such path, we obtain a collection of spaced-out paths routing all demand pairs in $\mset'$.

\label{-------------------------------------------------subsec: NI from routing to partitioning------------------------}
\subsection{From Routing to Partitioning}\label{subsec: NI from routing to partitioning}

The goal of this subsection is to prove the following theorem, that will complete the proof of Theorem~\ref{thm: from WGP to NDP}.

\begin{theorem}\label{thm: can get good partition from routing}
There is a deterministic  efficient algorithm, that,  given a valid instance $\iset=(\tilde G=(V_1,V_2,E),\uset_1,\uset_2,h,r)$ of the \WGPwB problem with $|E|=M$, the corresponding (random) instance $\hat {\iset}$ of \NDPgrid, and 
 a solution $\pset^*$ to $\hat{\iset}$, computes a solution to the \WGPwB instance $\iset$ of value at least $\Omega(|\pset^*|/\log^3M)$.
\end{theorem}

Let $\mset^*\subseteq \mset$ be the set of the demand pairs routed by the solution $\pset^*$, and let $E^*\subseteq E$ be the set of all edges $e$, whose corresponding demand pair belongs to $\mset^*$. Let $\tG'\subseteq \tG$ be the sub-graph of $\tG$ induced by the edges in $E^*$. Notice that whenever two edges of $\tG$ belong to the same bundle, their corresponding demand pairs share a source or a destination. Since all paths in $\pset^*$ are node-disjoint, all demand pairs in $\mset^*$ have distinct sources and destinations, and so no two edges in $E^*$ belong to the same bundle.

Note that, if $|\pset^*| \leq 2^{64} h \log^3 M$, then we can return the solution $((W_1,\ldots,W_r), (E_1,\ldots,E_r))$, where $W_1=V(\tG)$ and $W_2=W_3=\cdots=W_r=\emptyset$; set $E_1$ contains an arbitrary subset of $\ceil{\frac{|\pset^*|}{2^{64}\log^3M}} \leq h$ edges of $E^*$, and all other sets $E_i$ are empty. Since no two edges of $E^*$ belong to the same bundle, we obtain a feasible solution to the \WGPwB problem instance of value  $\Omega(|\pset^*|/\log^3M)$.
Therefore, from now on, we assume that $|\pset^*|>2^{64}h\log^3M$.

Our algorithm computes a solution to the \WGPwB instance $\iset$ by repeatedly partitioning $\tG'$ into smaller and smaller sub-graphs, by employing suitably defined balanced cuts.


Recall that, given a graph $\H$, a \emph{cut} in $\H$ is a bi-partition $(A,B)$ of its vertices. We denote by $E_{\H}(A,B)$ the set of all edges with one endpoint in $A$ and another in $B$, and by $E_{\H}(A)$ and $E_{\H}(B)$ the sets of all edges with both endpoints in $A$ and in $B$, respectively. Given a cut $(A,B)$ of $\H$, the \emph{value} of the cut is $|E_{\H}(A,B)|$. We will omit the subscript $\H$ when clear from context. 

\begin{definition}
Given a graph $\H$ and a parameter $0<\rho<1$, a cut $(A,B)$ of $\H$ is called a $\rho$-edge-balanced cut iff $|E(A)|,|E(B)| \geq \rho \cdot |E(\H)|$. 
\end{definition}


The following theorem is central to the proof of Theorem~\ref{thm: can get good partition from routing}.

\begin{theorem}\label{thm: balanced partition}
There is an efficient algorithm, that, given a vertex-induced subgraph $\H$ of $\tG'$ with $|E(\H)| >  2^{64}h\log^3M$, computes a $1/32$-edge-balanced cut of $\H$, of value at most $\frac{|E(\H)|}{64 \log M}$.
\end{theorem}

We prove Theorem~\ref{thm: balanced partition} below, after we complete the proof of Theorem~\ref{thm: can get good partition from routing} using it.
Our algorithm maintains a collection $\gset$ of disjoint vertex-induced sub-graphs of $\tG'$, and consists of a number of phases. The input to the first phase is the collection $\gset$ containing a single graph - the graph $\tG'$. The algorithm continues as long as $\gset$ contains a graph $\H\in \gset$ with $|E(\H)| > 2^{64} \cdot h \log^3 M$; if no such graph $\H$ exists, the algorithm terminates. Each phase is executed as follows. We process every graph $\H\in \gset$ with $|E(\H)| > 2^{64} \cdot h \log^3 M$ one-by-one. When graph $\H$ is processed, we apply Theorem~\ref{thm: balanced partition} to it, obtaining a $1/32$-edge-balanced cut $(A,B)$ of $\H$, of value at most $\frac{|E(\H)|}{64 \log M}$. We then remove $\H$ from $\gset$, and add $\H[A]$ and $\H[B]$ to $\gset$ instead. This completes the description of the algorithm.
 We use the following claim to analyze it.

\begin{claim} \label{clm: gset is good solution2}
 Let $\gset'$ be the final set of disjoint sub-graphs of $\tG'$ obtained at the end of the algorithm. Then    $\sum_{\H \in \gset'} |E(\H)| \geq \Omega(|E(\tG')|)$, and   $|\gset'| \leq r$.
\end{claim}
\begin{proof}
We construct a binary partitioning tree $\tau$ of graph $\tG'$, that simulates the graph partitions computed by the algorithm. 
For every sub-graph $\H\subseteq \tG'$ that belonged to $\gset$ over the course of the algorithm, tree $\tau$ contains a vertex $v(\H)$. The root of the tree is the vertex $v(\tG')$. If, over the course of our algorithm, we have partitioned the graph $\H$ into two disjoint vertex-induced sub-graphs $\H'$ and $\H''$, then we add an edge from $v(\H)$ to $v(\H')$ and to $v(\H'')$, that become the children of $v(\H)$ in $\tau$. 

The \emph{level} of a vertex $v(\H)$ in the tree is the length of the path connecting $v(\H)$ to the root of the tree; so the root of the tree is at level $0$. The \emph{depth} of the tree, that we denote by $\Delta$, is the length of the longest leaf-to-root path in the tree. 
Since the cuts computed over the course of the algorithm are $1/32-$edge-balanced, $\Delta\leq \frac{\log M}{\log{(32/31)}}$. 
Consider now some  level $0\leq i\leq \Delta$, and let $V_i$ be the set of all vertices of the tree $\tau$ lying at level $i$. Let $\hat E_i=\bigcup_{v(\H)\in V_i}E(\H)$ be the set of all edges contained in all sub-graphs of $\tG'$, whose corresponding vertex belongs to level $i$. Finally, let $m_i=|E_i|$. Then $m_0=|E(\tG')|$, and
for all $1\leq i\leq \Delta$, the number of edges discarded over the course of phase $i$ is $m_{i-1}-m_{i}\leq m_{i-1}/(64\log M)\leq m_0/(64\log M)$ --- this is since, whenever we partition a graph $\H$ into two subgraphs, we lose at most $|E(\H)|/(64 \log M)$ of its edges. Overall, we get that:

\[m_0-m_{\Delta}\leq \frac{\Delta m_0}{64\log M}  \leq \frac{\log M}{\log{32/31}}\cdot \frac{m_0}{64\log M}\leq \frac{m_0} 2,\]

and so $\sum_{\H \in \gset'} |E(\H)|=m_{\Delta}\geq m_0/2$. This finishes the proof of the first assertion. We now turn to prove the second assertion. Recall that no two edges of $E^*$ may belong to the same bundle. Since $h=\beta^*(\iset)/r=\left(\sum_{v\in V_1}\beta(v)\right )/r$, for any subset $E'\subseteq E^*$ of edges, $|E'|\leq \sum_{v\in V_1}\beta(v)\leq hr$ must hold, and in particular, $|E^*|\leq hr$. It is now enough to prove that for every leaf vertex $v(\H)$ of $\tau$, $|E(\H)|\geq h$ --- since all graphs in $\gset'$ are mutually disjoint, and each such graph corresponds to a distinct leaf of $\tau$, this would imply that $|\gset'|\leq r$.

Consider now some leaf vertex $v(\H)$ of $\tau$, and let $v(\H')$ be its parent. The $|E(\H')|\geq 2^{64}h\log^3M$, and, since the partition of $\H'$ that we have computed was $1/32$-balanced, $|E(\H)|\geq |E(\H')|/32\geq h$. We conclude that $|\gset'|\leq r$.
\end{proof}

We are now ready to define the solution $((W_1,\ldots,W_r),(E_1,\ldots,E_r))$ to the \WGPwB problem instance $\iset$. Let $\gset'$ be the set of the sub-graphs of $\tG'$ obtained at the end of our algorithm, and denote $\gset'=\set{\H_1,\H_2,\ldots,\H_z}$. Recall that from Claim~\ref{clm: gset is good solution2}, $z\leq r$. For $1\leq i\leq z$, we let $W_i=V(\H_i)$. If $|E(\H_i)|\leq h$, then we let $E_i=E(\H_i)$; otherwise, we let $E_i$ contain any subset of $h$ edges of $E(\H_i)$. Since $|E(\H_i)|\leq 2^{64}h\log^3M$, in either case, $|E_i|\geq \Omega(|E(\H_i)|/\log^3M)$. For $i>z$, we set $W_i=\emptyset$ and $E_i=\emptyset$. Since, as observed before, no pair of edges of $E^*$ belongs to the same bundle, it is immediate to verify that we obtain a feasible solution to the \WGPwB problem instance. The value of the solution is:

\[\sum_{i=1}^r|E_i|\geq \sum_{i=1}^r\Omega(|E(\H_i)|/\log^3M)=\Omega(|E(\tG')|/\log^3M)= \Omega(|\pset^*|/\log^3M),\]

from Claim~\ref{clm: gset is good solution2}. In order to complete the proof of Theorem~\ref{thm: can get good partition from routing}, it now remains to prove Theorem~\ref{thm: balanced partition}.

\begin{proofof}{Theorem~\ref{thm: balanced partition}}
Let $\H$ be a vertex-induced subgraph of $\tG'$ with $|E(\H)| >  2^{64}h\log^3M$. 
Our proof consists of two parts. First, we show an efficient algorithm to compute a drawing of $\H$ with relatively few crossings. Next, we show how to exploit this drawing in order to compute a $1/32$-edge-balanced cut of $\H$ of small value. We start by defining a drawing of a given graph in the plane and the crossings in this drawing.

\begin{definition}
A \emph{drawing} of a given graph $\H'$ in the plane is a mapping, in which every vertex of $\H$ is mapped to a point in the plane, and every edge to a continuous curve connecting the images of its endpoints, such that no three curves meet at the same point; no curve intersects itself; and no curve contains an image of any vertex other than its endpoints.
A \emph{crossing} in such a drawing is a point contained in the images of two edges.
\end{definition}

\begin{lemma}\label{lem: drawing of H}
There is an efficient algorithm that, given a solution $\pset^*$ to the \NDPgrid problem instance $\hat \iset$, and a vertex-induced subgraph $\H\subseteq \tG'$, computes a drawing of $\H$ with at most  $\twiceConstantForSizeOfBlocks \cdot |E(\H)| \ceil{h \log M}$ crossings.
\end{lemma}

\begin{proof}
We think of the grid $\hat G$ underlying the \NDPgrid instance $\hat \iset$ as the drawing board, and map the vertices of $\H$ to points inside some of its carefully selected cells. Consider a vertex $v_i \in V(\H) \cap V_1$, and let $\block_i$ be the block representing this vertex.
Let $\kappa_i$ be any cell of the grid $\hat G$ that has a vertex of $\block_i$ on its boundary. We map the vertex $v_i$ to a point $p_i$ lying in the middle of the cell $\kappa_i$  (it is sufficient that $p_i$ is far enough from the boundaries of the cell).
For every vertex $v'_j\in V(\H)\cap V_2$, we select a cell $\kappa'_j$ whose boundary contains a vertex of the corresponding block $\block'_j$, and map $v'_j$ to a point $p'_j$ lying in the middle of $\kappa'_j$ similarly. 

Next, we define the drawings of the edges of $E(\H)$. 
Consider any such edge $e = (v_i,v'_j) \in E(\H)$, with $v_i\in V_1$ and $v'_j\in V_2$, and let $(s,t)\in \mset^*$ be its corresponding demand pair.
Let $\block_i$ and $\block'_j$ be the blocks containing $s$ and $t$ respectively, and let $P\in \pset^*$ be the path routing the demand pair $(s,t)$ in our solution to the \NDPgrid problem instance $\hat \iset$.
The drawing of the edge $e$ is a concatenation of the following three segments: (i) the image of the path $P$, that we refer to as a type-1 segment; (ii) a straight line connecting $s$ to the image of $v_i$, that we refer to as a type-2 segment; and (iii) a straight line connecting $t$ to the image of $v'_i$, that we refer to as a type-3 segment. If the resulting curve has any self-loops, then we delete them.

We now bound the number of crossings in the resulting drawing of $\H$. Since the paths in $\pset^*$ are node-disjoint, whenever the images of two edges $e$ and $e'$ cross, the crossing must be between the type-1 segment of $e$, and either the type-2 or the type-3 segment of $e'$ (or the other way around). 

Consider now some edge $e=(v_i,v'_j) \in E(\H)$ with $v_i\in V_1$ and $v'_j\in V_2$, and let $\block_i$ and $\block'_j$ be the blocks representing $v_i$ and $v'_j$ respectively. Let $(s,t)\in \mset^*$ be the demand pair corresponding to $e$. Assume that a type-1 segment of some edge $e'$ crosses a type-$2$ segment of $e$. This can only happen if 
the path $P'\in \pset^*$ routing the demand pair $(s',t')$ corresponding to the edge $e'$ contains a vertex of $\block_i$. Since $|V(\block_i)| \leq \constantForSizeOfBlocks \cdot \ceil{h \log M}$, at most $ \constantForSizeOfBlocks \cdot \ceil{h \log M}$ type-1 segments of other edges may cross the type-2 segment of $e$. The same accounting applies to the type-3 segment of $e$. 
Overall, the number of crossings in the above drawing is bounded by: 

\[\sum_{e \in E(\H)} 2 \cdot \constantForSizeOfBlocks \cdot \ceil{h \log M} = \twiceConstantForSizeOfBlocks \cdot |E(\H)| \ceil{h \log M}.\]
%
\end{proof}

Next, we show an efficient algorithm that computes a small balanced partition of any given graph $\H'$, as long as its maximum vertex degree is suitably bounded, and we are given a drawing of $\H'$ with a small number of crossings.

\begin{lemma}\label{lemma: good balanced cut in low crossing number}
There is an efficient algorithm that, given any graph $\H$ with $|E(\H)|=m$ and maximum vertex degree at most $d$, and a drawing $\phi$ of $\H$ with at most $\cro\leq md\alpha$ crossings for some $\alpha>1$, such that $m>2^{20}d\alpha$, computes a $1/32$-edge-balanced cut $(A,B)$ of $\H$ of value $|E(A,B)|\leq 64\sqrt{8md\alpha}$.
\end{lemma}

Before we complete the proof of the lemma, we show that the proof of Theorem~\ref{thm: balanced partition} follows from it. Let $m=|E(\H)|$. Note that the maximum vertex degree in graph $\H$ is bounded by $\max_{v\in V_1\cup V_2}\set{\beta(v)}$, as $E(\H)$ cannot contain two edges that belong to the same bundle. From the definition of valid instances, $h\geq \beta(v)$, and so the maximum vertex degree in $\H$ is bounded by $d=h$. From Lemma~\ref{lem: drawing of H}, the drawing $\phi$ of $\H$ has at most $2048 |E(\H)| \ceil{h \log M}\leq 2^{12}m d\log M$ crossings. Setting $\alpha=2^{12}\log M$, the number of crossings $\cro$ in $\phi$ is bounded by $md\alpha$. Moreover, since $m>2^{64}h\log^3M$, we get that $m>2^{20}d\alpha$. We can now apply Lemma~\ref{lemma: good balanced cut in low crossing number} to graph $\H$ to obtain a $1/32$-edge-balanced cut $(A,B)$ with $|E(A,B)|\leq 64\sqrt{8md\alpha}\leq 64\sqrt{2^{15}mh\log M}\leq 64\sqrt{m^2/(2^{49}\log^2M)}\leq \frac{|E(\H)|}{64\log M}$, since we have assumed that $|E(\H)|=m>2^{64}h\log^3M$. It now remains to prove Lemma~\ref{lemma: good balanced cut in low crossing number}.

\begin{proofof}{Lemma~\ref{lemma: good balanced cut in low crossing number}}
For each vertex $v \in V(\H)$, we denote the degree of $v$ in $\H$ by $d_v$. We assume without loss of generality that for all $v\in V(\H)$, $d_v\geq 1$: otherwise, we can remove all isolated vertices from $\H$ and then apply our algorithm to compute a $1/32$-edge-balanced cut $(A,B)$ in the remaining graph. At the end we can add the isolated vertices to $A$ or $B$, while ensuring that the cut remains $1/32$-balanced, and without increasing its value.

We construct a new graph $\hat \H$ from graph $\H$ as follows. For every vertex $v\in V(\H)$, we add a  $(d_v \times d_v)$-grid $Q_v$ to $\hat \H$, so that the resulting grids are mutually disjoint. We call the edges of the resulting grids \emph{regular edges}. 
Let $e_1(v), \ldots, e_{d_v}(v)$ be the edges of $\H$ incident to $v$, indexed in the clockwise order of their entering the vertex $v$ in the drawing $\phi$ of $\H$. We denote by $\Pi(v)=\set{p_1(v),\ldots,p_{d_v}(v)}$ the set of vertices on the top boundary of $Q_v$, where the vertices are indexed in the clock-wise order of their appearance on the boundary of $Q_v$. We refer to the vertices of $\Pi(v)$ as the \emph{portals} of $Q_v$ (see Figure~\ref{fig: v_to_gv}). Let $\Pi=\bigcup_{v\in V(\H)}\Pi(v)$ be the set of all portals. For every edge $e=(u,v)\in E(\H)$, we add a new \emph{special edge} to graph $\hat \H$, as follows. Assume that $e=e_i(v)=e_j(u)$. Then we add an edge $(p_i(v),p_j(u))$ to $\hat \H$. We think of this edge as the special edge representing $e$. This finishes the definition of the graph $\hat \H$. It is immediate to see that the drawing $\phi$ of $\H$ can be extended to a drawing $\phi'$ of $\hat \H$ without introducing any new crossings, that is, the number of crossings in $\phi'$ remains at most $\cro$. Note that every portal vertex is incident to exactly one special edge, and the maximum vertex degree in $\hat \H$ is $4$. We will use the following bound on $|V(\hat \H)|$:

\begin{figure}
\center
\includegraphics[width=8cm]{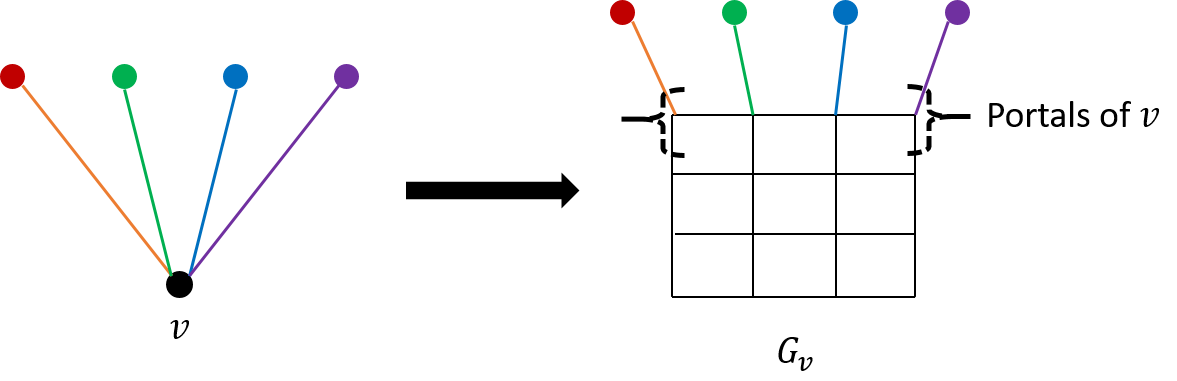}
\caption{Grid $Q_v$ obtained from $v$} \label{fig: v_to_gv}
\end{figure}

\begin{observation}\label{obs: number of vertices in hat H}
$|V(\hat \H)|\leq (2m+d)d$.
\end{observation}

\begin{proof}
Clearly, $|V(\hat \H)|=\sum_{v\in V(\H)}d^2_v$. Observe that for every pair $a\geq b\geq 0$ of integers, $(a+1)^2+(b-1)^2\geq a^2+b^2$. Since maximum vertex degree in $\H$ is bounded by $d$, the sum is maximized when all but possibly one summand are equal to $d$, and, since $\sum_{v\in \H}d_v=2m$, there are at most $\ceil{2m/d}$ summands. Therefore, $|V(\hat \H)|\leq \floor{2m/d}\cdot d^2+d^2\leq (2m+d)d$.
\end{proof}

Let $\hat \H'$ be the graph obtained from $\hat \H$ by replacing every intersection point in the drawing $\phi'$ of $\hat \H$ with a vertex. Then $\hat \H'$ is a planar graph with $|V(\hat \H')|\leq |V(\hat \H)|+\cro\leq  (2m+d)d+md\alpha\leq 4md\alpha$, as $\alpha\geq 1$. We assign weights to the vertices of $\hat \H'$ as follows: every vertex of $\Pi$ is assigned the weight $1$, and every other vertex is assigned the weight $0$. Note that the weight of a vertex is exactly the number of special edges incident to it, and the total weight of all vertices is $W=|\Pi|=2m$.
%
%
%
%
We will use the following version of the planar separator theorem~\cite{planar-separator-theorem2,planar-separator-theorem1,planar-separator-theorem3}.

\begin{theorem}[\cite{planar-separator-theorem1}]\label{thm: balanced separator}
There is an efficient algorithm, that, given a planar graph $G=(V,E)$  with $n$ vertices, and an assignment $w:V\rightarrow R^+$ of non-negative weights to the vertices of $G$, with $\sum_{v\in V}w(v)=W$, computes a partition $(A,X,B)$ of $V(G)$, such that:

\begin{itemize}
\item no edge connecting a vertex of $A$ to a vertex of $B$ exists in $G$;
\item $\sum_{v\in A}w(v),\sum_{v\in B}w(v)\leq 2W/3$; and
\item $|X|\leq 2\sqrt{2n}$.
\end{itemize}
\end{theorem}


%

We apply Theorem~\ref{thm: balanced separator} to graph $\hat \H'$, to obtain a partition $(A,X,B)$ of $V(\hat \H')$, with $|X|\leq 2\sqrt{2|V(\hat \H')|}\leq 2\sqrt{8md\alpha}$. Since $W=\sum_{v\in V(\hat \H')}w(v)=2m$, we get that $|A\cap \Pi|= \sum_{v\in A}w(v)\leq 2W/3\leq 4m/3$, and similarly $|B\cap \Pi|\leq 4m/3$.
Assume without loss of generality that $|A\cap \Pi|\leq |B\cap \Pi|$. We obtain a bi-partition $(A',B')$ of $V(\hat \H')$ by setting $A'=A\cup X$ and $B'=B$. Since $|X|\leq 2\sqrt{8md\alpha}\leq m/3$ (as $m>2^{20}d\alpha$), we are guaranteed that $|A'\cap \Pi|,|B'\cap \Pi|\leq 4m/3$ holds. Moreover, as all vertex degrees in $\hat \H'$ are at most $4$, $|E(A',B')|\leq 4|X|\leq 8\sqrt{8md\alpha}$.

Unfortunately, the cut $(A',B')$ of $\hat H'$ does not directly translate into a balanced cut in $\H$, since for some vertices $v\in V(\H)$, the corresponding grid $Q_v$ may be split between $A'$ and $B'$. We now show how to overcome this difficulty, by moving each such grid entirely to one of the two sides.
Before we proceed, we state a simple fact about grid graphs.

\begin{observation}\label{obs: cuts in grids} Let $z>1$ be an integer, and let $Q$ be the $(z\times z)$-grid. Let $U$ be the set of vertices lying on the top row of $Q$, and let $(X,Y)$ be a bi-partition of $V(Q)$. Then $|E_Q(X,Y)|\geq \min\set{|U\cap X|,|U\cap Y|}$.
\end{observation}

\begin{proof}
It is easy to verify that for any bi-partition $(X',Y')$ of $U$ into two disjoint subsets, there is a set $\pset$ of $\min\set{|X'|,|Y'|}$ node-disjoint paths in $Q$ connecting vertices of $X'$ to vertices of $Y'$. The observation follows from the maximum flow -- minimum cut theorem.
\end{proof}

We say that a vertex $v \in V(\H)$ is \textit{split} by the cut $(A',B')$ iff $V(Q_v)\cap A'$ and $V(Q_v)\cap B'\neq \emptyset$. We say that it is \emph{split evenly} iff $|\Pi(v)\cap A'|,|\Pi(v)\cap B'|\geq d_v/8$; otherwise we say that it is \emph{split unevenly}.
We modify the cut $(A',B')$ in the following two steps, to ensure that no vertex of $V(\H)$ remains split.

\paragraph{Step 1 [Unevenly split vertices].}
We process each vertex $v\in V(\H)$ that is unevenly split one-by-one. Consider any such vertex $v$. If $|\Pi(v)\cap A'|>|\Pi(v) \cap B'|$, then we move all vertices of 
$Q_v$ to $A'$; otherwise we move all vertices of $Q_v$ to $B'$. Assume without loss of generality that the former happens.
Notice that from Observation~\ref{obs: cuts in grids}, $|E(A',B')|$ does not increase, since $E(Q_v)$ contributed at least $|\Pi(v)\cap B'|$ regular edges to the cut before the current iterations.  Moreover, $|\Pi(v)\cap A'|$ increases by the factor of at most $8/7$. Therefore, at the end of this procedure, once all unevenly split vertices of $\H$ are processed, $|A'\cap \Pi|,|B'\cap \Pi|\leq \frac 8 7 \cdot \frac 4 3 m=\frac{32}{21}m$ and $|E(A',B')|\leq 8\sqrt{8md\alpha}$.


\paragraph{Step 2 [Evenly split vertices].}
In this step, we process each vertex $v\in V(\H)$ that is evenly split one-by-one. Consider an iteration where some such vertex $v\in V(\H)$ is processed. If $|A'\cap \Pi|\leq |B'\cap \Pi|$, then we move all vertices of $Q_v$ to $A'$; otherwise we move them to $B'$. Assume without loss of generality that the former happened. Then before the current iteration $|A'\cap \Pi|\leq |\Pi|/2\leq m$, and,  since $|\Gamma(v)|\leq d<m/21$,  $|A'\cap \Pi|\leq \frac{32}{21}m$, while $|B'\cap \Pi|\leq \frac{32}{21}m$ as before. Moreover, from Observation~\ref{obs: cuts in grids}, before the current iteration, the regular edges of $Q_v$ contributed at least $d(v)/8$ edges to $E(A',B')$, and after the current iteration, no regular edges of $Q_v$ contribute to the cut, but we may have added up to $d(v)$ new special edges to it. Therefore, after all vertices of $\H$ that are evenly split are processed, $|E(A',B')|$ grows by the factor of at most $8$, and remains at most $64\sqrt{8md\alpha}$.

We are now ready to define the final cut $(A^*,B^*)$ in graph $\H$. We let $A^*$ contain all vertices $v\in V(\H)$ with $V(Q_v)\subseteq A'$, and we let $B^*$ contain all remaining vertices of $V(\H)$. Clearly, $|E_{\H}(A^*,B^*)|\leq |E_{\hat \H'}(A',B')|\leq 64\sqrt{8md\alpha}$. It remains to show that $|E_{\H}(A^*)|, |E_{\H}(B^*)| \geq |E(\H)|/32$. We show that $|E_{\H}(A^*)| \geq |E(\H)|/32$; the proof that $|E_{\H}(B^*)| \geq |E(\H)|/32$ is symmetric.
Observe that $\sum_{v\in B^*}d_v=|B'\cap \Pi|\leq \frac{32m}{21}$, while $\sum_{v\in V(\H)}d_v=2m$. Therefore, $\sum_{v\in A^*}d_v\geq 2m-\frac{32m}{21}=\frac{10m}{21}$. But $|E_{\H}(A^*,B^*)|\leq 64\sqrt{8md\alpha}\leq 64\sqrt{m^2/2^{17}}<m/4$ (since $m>2^{20}d\alpha$). Therefore,

\[|E_{\H}(A^*)|=\frac{\sum_{v\in A^*}d_v-|E_{\H}(A^*,B^*)|}{2}\geq \frac{5m}{21}-\frac m 8\geq\frac m {32}.\]

\end{proofof} \end{proofof}

\label{----------------------------------------sec: from NDP to EDP-------------------------------}
\section{Hardness of \NDP and \EDP on Wall Graphs}\label{sec: from NDP to EDP}
In this section we extend our results to \NDP and \EDP on wall graphs, completing the proofs of Theorem~\ref{thm: master NDP} and Theorem~\ref{thm: master EDP}. We first prove hardness of \NDPwall, and show later how to extend it to \EDPwall.
Let $\hat G=G^{\ell,h}$ be a grid of length $\ell$ and height $h$, where $\ell>0$ is an even integer, and $h>0$. 
We denote by $\hat G'$ the wall corresponding to $\hat G$, as defined in Section~\ref{sec: prelims}.
We prove the following analogue of Theorem~\ref{thm: from WGP to NDP}.

\begin{theorem}\label{thm: from WGP to NDP in walls}
There  is a constant $c^*>0$, and there is an efficient randomized algorithm, that, given a valid instance $\iset=(\tilde G, \uset_1,\uset_2,h,r)$ of \WGPwB with $|E(\tilde G)|=M$, constructs an instance $\hat{\iset}'=(\hat G',\mset)$ of \NDPwall with $|V(\hat G')|=O(M^4\log^2M)$, such that the following hold:

\begin{itemize}
\item If $\iset$ has a perfect solution (of value $\beta^*=\beta^*(\iset)$), then with probability at least $\half$ over the construction of $\hat \iset'$, instance $\hat\iset'$ has a solution $\pset'$ that routes at least $\frac{\beta^*}{c^*\log^3M}$ demand pairs via node-disjoint paths; and

\item There is a deterministic efficient algorithm, that, given a solution $\pset^*$ to the \NDPwall problem instance $\hat{\iset}'$, constructs a solution to the \WGPwB instance $\iset$, of value at least $\frac{|\pset^*|}{c^*\cdot \log^3M}$.
\end{itemize}
\end{theorem}

Notice that plugging Theorem~\ref{thm: from WGP to NDP in walls} into the hardness of approximation proof instead of Theorem~\ref{thm: from WGP to NDP}, we extend the hardness result to the \NDP problem on wall graphs and complete the proof of Theorem~\ref{thm: master NDP}.

\begin{proofof}{Theorem~\ref{thm: from WGP to NDP in walls}}
Let $\hat \iset=(\hat G,\mset)$ be the instance of \NDPgrid constructed in Theorem~\ref{thm: from WGP to NDP}. In order to obtain an instance $\hat \iset'$ of \NDPwall, we replace the grid $\hat G$ with the corresponding wall $\hat G'$ as described above; the set of the demand pairs remains unchanged. We now prove the two assertions about the resulting instance $\hat{\iset}'$, starting from the second one.

Suppose we are given a solution $\pset^*$ to the \NDPwall problem instance $\hat{\iset}'$. Since $\hat G'\subseteq \hat G$, and the set of demand pairs in instances $\hat \iset$ and $\hat \iset'$ is the same, $\pset^*$ is also a feasible solution to the \NDPgrid problem instance $\hat \iset$, and so we can use the deterministic efficient algorithm from Theorem~\ref{thm: from WGP to NDP} to construct  a solution to the \WGPwB instance $\iset$, of value at least $\frac{|\pset^*|}{c^*\cdot \log^3M}$.

It now remains to prove the first assertion. Assume that $\iset$  has a perfect solution. Let $\event$ be the good event that the instance $\hat \iset$ of \NDPgrid has a solution $\pset$  that routes at least $\frac{\beta^*}{c^*\log^3M}$ demand pairs via  paths that are spaced-out. From Theorem~\ref{thm: from WGP to NDP}, event $\event$ happens with probability at least $\half$. It is now enough to show that whenever event $\event$ happens, there is a solution of value $\frac{\beta^*}{c^*\log^3M}$ to the corresponding instance $\hat \iset'$ of \NDPwall.

Consider the spaced-out set $\pset$ of paths in $\hat G$. Recall that for every pair $P,P'$ of paths, $d(V(P),V(P'))\geq 2$, and all paths in $\pset$ are internally disjoint from the boundaries of the grid $\hat G$. For each path $P\in \pset$, we will slightly modify $P$ to obtain a new path $P'$ contained in the wall $\hat G'$, so that the resulting set $\pset'=\set{P'\mid P\in \pset}$ of paths is node-disjoint.

For all $1\leq i<\ell$, $1\leq j\leq \ell$, let $e_i^j$ denote the $i$th edge from the top lying in column $W_j$ of the grid $\hat G$, so that $e^i_j=(v(i,j),v(i+1,j))$.
Let $E^*=E(\hat G)\setminus E(\hat G')$ be the set of edges that were deleted from  the grid $\hat G$ when constructing the wall $\hat G'$. We call the edges of $E^*$ \emph{bad edges}.
Notice that only vertical edges may be bad, and, if $e^j_i\in E(W_j)$ is a bad edge, for $1<j<\ell$, then $e^{j+1}_i$ is a good edge. Consider some bad edge $e^j_i=(v(i,j),v(i+1,j))$, such that $1<j<\ell$, so $e^j_i$ does not lie on the boundary of $\hat G$. Let $Q^j_i$ be the path $(v(i,j),v(i,j+1),v(i+1,j+1),v(i+1,j))$. Clearly, path $Q^j_i$ is contained in the wall $\hat G'$. For every path $P\in \pset$, we obtain the new path $P'$ by replacing every bad edge $e^j_i\in P$ with the corresponding path $Q^j_i$. It is easy to verify that $P'$ is a path with the same endpoints as $P$, and that it is contained in the wall $\hat G'$. Moreover, since the paths in $\pset$ are spaced-out, the paths in the resulting set $\pset'=\set{P'\mid P\in \pset}$ are node-disjoint.
\end{proofof}

This completes the proof of Theorem~\ref{thm: master NDP}. In order to prove Theorem~\ref{thm: master EDP}, we show an approximation-preserving reduction from \NDPwall to \EDPwall.

\begin{claim} \label{NDPwall to EDP wall}
Let $\iset=(G,\mset)$ be an instance of \NDPwall, and let $\iset'$ be the instance of \EDPwall consisting of the same graph $G$ and the same set $\mset$ of demand pairs.
Let $\opt$ and $\opt'$ be the optimal solution values for $\iset$ and $\iset'$, respectively. Then $\opt'\geq \opt$, and there is an efficient algorithm, that, given any solution $\pset'$ to instance $\iset'$ of \EDPwall, computes a solution $\pset$ to instance $\iset$ of \NDPwall of value $\Omega(|\pset'|)$.
\end{claim}

The following corollary immediately follows from Claim~\ref{NDPwall to EDP wall} and completes the proof of Theorem~\ref{thm: master EDP}.
\begin{corollary}
If there is an $\alpha$-approximation algorithm for \EDPwall with running time $f(n)$, for $\alpha>1$ that may be a function of the graph size $n$, then there is an $O(\alpha)$-approximation algorithm for \NDPwall with running time $f(n)+\poly(n)$.
\end{corollary}

It now remains to prove Claim~\ref{NDPwall to EDP wall}.

\begin{proof}
The assertion that $\opt'\geq \opt$ is immediate, as any set $\pset$ of node-disjoint paths in the wall $G$ is also a set of edge-disjoint paths. 

Assume now that we are given a set $\pset'$ of edge-disjoint paths in $G$. We show an efficient algorithm to compute a subset $\pset\subseteq \pset'$ of $\Omega(|\pset'|)$ paths that are node-disjoint. Since the maximum vertex degree in $G$ is 3, the only way for two paths $P,P'\in\pset'$ to share a vertex $x$ is when $x$ is an endpoint of at least one of these two paths. If $x$ is an endpoint of $P$, and $x\in V(P')$, then we say that $P$ has a conflict with $P'$.

We construct a directed graph $H$, whose vertex set is $\set{v_P\mid P\in\pset'}$, and there is an edge $(v_P,v_{P'})$ iff $P$ has a conflict with $P'$. It is immediate to verify that the maximum out-degree of any vertex in $H$ is at most $4$, as each of the two endpoints of a path $P$ may be shared by at most two additional paths. Therefore, every sub-graph $H'\subseteq H$ of $H$ contains a vertex of total degree at most $8$. We construct a set $U$ of vertices, such that no two vertices of $U$ are connected by an edge, using a standard greedy algorithm: while $H\neq \emptyset$, select a vertex $v\in H$ with total degree at most $8$ and add it to $U$; remove $v$ and all its neighbors from $H$. It is easy to verify that at the end of the algorithm, $|U|=\Omega(|V(H)|)=\Omega(|\pset'|)$, and no pair of vertices in $U$ is connected by an edge. Let $\pset=\set{P\mid v_P\in U}$. Then the paths in $\pset$ are node-disjoint, and $|\pset|=\Omega(|\pset'|)$.
\end{proof}

\appendix

\label{--------------------------------------------Appendix-----------------------------------------------------}

\section*{Appendix}


\section{Proof of Theorem~\ref{thm: yi-partition-of-answers}}\label{appdx: proof of yi-partition-of-answers thm}

Suppose $G$ is a \yi, and let $\chi$ be a valid coloring of $V(G)$. Let $\pi_1,\ldots,\pi_6$ be $6$ different permutations of $\set{r,g,b}$. 
For each $1\leq i\leq 6$, permutation $\pi_i$ defines a valid coloring $\chi_i$ of $G$: for every vertex $v\in V(G)$, if $v$ is assigned a color $c\in \set\cset$ by $\chi$, then 
$\chi_i$ assigns the color $\pi_i(c)$ to $v$. Notice that for each vertex $v$ and for each color $c\in \cset$, there are exactly two indices $i\in \set{1,\ldots, 6}$, such that $\chi_i$ assigns the color $c$ to $v$. Notice also that for each edge $(u,v)$, if $c,c'\in \cset$ is any pair of distinct colors, then there is exactly one index $i\in \set{1,\ldots,6}$, such that $u$ is assigned the color $c$ and $v$ is assigned the color $c'$ by $\chi_i$.

Let $B$ be the set of all vectors of length $\ell$, whose entries belong to $\set{1,\ldots,6}$, so that $|B|=6^{\ell}$. For each such vector $b\in B$, we define a perfect global assignment $f_b$ of answers to the queries, as follows. Let $Q\in \qset^E$ be a query to the edge-player, and assume that $Q=(e_1,\ldots,e_{\ell})$. Fix some index $1\leq j\leq \ell$, and assume that $e_j=(v_j,u_j)$. Assume that $b_j=z$, for some $1\leq z\leq 6$. We assign to $v_j$ the color $\chi_z(v_j)$, and we assign to $u_j$ the color $\chi_z(u_j)$. Since $\chi_z$ is a valid coloring of $V(G)$, the two colors are distinct. This defines an answer $A\in \aset^E$ to the query $Q$, that determines $f_b(Q)$.

Consider now some query $Q'\in \qset^V$ to the vertex-player, and assume that $Q'=(v_1,\ldots,v_{\ell})$. Fix some index $1\leq j\leq \ell$, and assume that $b_j=z$, for some $1\leq z\leq 6$. We assign to $v_j$ the color $\chi_z(v_j)$. This defines an answer $A'\in \aset^V$ to the query $Q'$, that determines $f_b(Q')$. Notice that for each $1\leq j\leq \ell$, the answers that we choose for the $j$th coordinate of each query are consistent with the valid coloring $\chi_{b_j}$ of $G$. Therefore, it is immediate to verify that for each $b\in B$, $f_b$ is a perfect global assignment.

We now fix some query $Q\in \qset^E$ of the edge-prover, and some answer $A\in \aset^E$ to it. Assume that $Q=(e_1,\ldots,e_{\ell})$, where for $1\leq j\leq \ell$, $e_j=(v_j,u_j)$. Let $c_j,c'_j$ are the assignments to $v_j$ and $u_j$ given by the $j$th coordinate of $A$, so that $c_j\neq c'_j$. Recall that there is exactly one index $z_j\in \set{1,\ldots,6}$, such that $\chi_{z_j}$ assigns the color $c_j$ to $v_j$ and the color $c'_j$ to $u_j$. Let $b^*\in B$ be the vector, where for $1\leq j\leq \ell$, $b^*_j=z_j$. Then $f_{b^*}(Q)=A$, and for all $b\neq b^*$, $f_b(Q)\neq A$.

Finally, fix some query $Q'\in \qset^V$ of the vertex-prover, and some answer $A'\in \aset^V$ to it. Let $Q'=(v_1,\ldots,v_{\ell})$. Assume that for each $1\leq j\leq \ell$, the $j$th coordinate of $A'$ contains the color $c_j$. Recall that there are exactly two indices $z\in \set{1,\ldots,6}$, such that $\chi_z$ assigns the color $c_j$ to $v_j$. Denote this set of two indices by $Z_j\subseteq\set{1,\ldots,6}$. Consider now some vector $b\in B$. If, for all $1\leq j\leq \ell$, $b_j\in Z_j$, then $f_b(Q')=A'$; otherwise, $f_b(Q')\neq A'$. Therefore, the total number of vectors $b\in B$, for which $f_b(Q')=A'$ is exactly $2^{\ell}$.

\label{-------------------------------------------------sec: auxiliary lemmas------------------------}
\section{Auxiliary Lemmas}\label{sec: auxiliarly lemmas}

The goal of this section is to prove Lemma \ref{lem: random ordering} and Lemma \ref{lem: random ordering2}. We will use the following simple observation.

\begin{observation} \label{obs: decreasing ratios}
For any two positive integers $a$ and $b$, $\frac{b-1}{a+b-1} < \frac{b}{a+b}$.
\end{observation}


Recall that we are given a set $U$ of $n$ items, such that  $P$ of the items are pink, and the remaining $Y=n-P$ items are yellow. We consider a random permutation $\pi$ of these items. Given a set $S\subseteq\set{1,\ldots,n}$ of $\delta$ indices, we let $\event(S)$ be the event that for all $i\in S$, the item of $\pi$ located at the $i$th position is yellow.

\begin{claim}\label{claim: event delta}
$\prob{\event(S)}\leq \left(\frac{Y}{n}\right )^{|S|}.$
\end{claim}

\begin{proof}

As every subset of $\delta$ items of $U$ is equally likely to appear at the indices of $S$, we get that: 

\[\prob{\event(S)} = \frac{\binom{Y}{\delta}}{\binom{P+Y}{\delta}}=
 \frac{Y \cdot (Y-1)\cdots (Y-\delta+1)}{(P+Y) \cdot (P+Y-1)\cdots (P+Y-\delta+1)}\leq  \left (\frac{Y}{P+Y}\right )^\delta=\left(\frac{Y}{n}\right )^{|S|}.\]

(the last inequality follows from Observation \ref{obs: decreasing ratios}).
\end{proof}

We now turn to prove Lemma \ref{lem: random ordering}. 

\begin{lemma}
[Restatement of Lemma \ref{lem: random ordering}.] For any $\log n\leq \mu\leq Y$, the probability that there is a sequence of $\ceil{4n\mu/P}$ consecutive items in $\pi$ that are all yellow, is at most $n/e^{\mu}$.
\end{lemma}

\begin{proof}
Let $\delta=\ceil{\frac{4n\mu}{P}}$. Consider a set $S$ of $\delta$ consecutive indices of $\set{1,\ldots,n}$. 
%
%
%
From Claim~\ref{claim: event delta}, the probability that all items located at the indices of $S$ are yellow is at most:

\[\left (\frac{Y}{n}\right )^\delta =  \left (1-\frac{P}{n}\right )^{\ceil{4n\mu/P}}\leq \left (1-\frac{P}{n}\right )^{4n\mu/P}\leq e^{-\mu}.\]
 
Since there are at most $n$ possible choices of a set $S$ of $\delta$ consecutive indices, from the Union Bound, the probability that any such set only contains yellow items is bounded by $n/e^{\mu}$.
\end{proof}

\begin{lemma}
[Restatement of Lemma \ref{lem: random ordering2}.] For any $\log n\leq \mu\leq P$, the probability that there is a set $S$ of $\floor{\frac{n\mu}{P}}$ consecutive items in $\pi$, such that more than $4\mu$ of the items are pink, is at most $n/4^{\mu}$.
\end{lemma}

\begin{proof}
 
 Let $x=\floor{\frac{n\mu}{P}}$, and let $S$ be any set of $x$ consecutive indices of $\set{1,\ldots,n}$. Denote $\delta=\ceil{4\mu}$, and let $S'\subseteq S$ be any subset of $\delta$ indices from $S$. From Claim~\ref{claim: event delta} (by reversing the roles of the pink and the yellow items), the probability that all items located at the indices of $S'$ are pink is at most $\left (\frac{P}{n}\right )^\delta$. 
Since there are  ${x\choose \delta}$ ways to choose the subset $S'$ of $S$, from the Union Bound,  the probability that at least $\delta$ items of $S$ are pink is at most:

\[\begin{split}
\left (\frac{P}{n}\right )^\delta\cdot {x\choose \delta}&\leq \left (\frac{P}{n}\right )^\delta\cdot \left(\frac{ex}{\delta}\right )^\delta\\
&\leq  \left (\frac{P}{n}\right )^\delta\cdot \left(\frac{e\cdot n\mu}{P \delta}\right )^\delta\\
&= \left (\frac{e\cdot \mu}{ \ceil{4\mu}}\right )^{\ceil{4\mu}}\\
&\leq \left(\frac{e}{4}\right)^{4 \mu}<\frac{1}{4^{\mu}}.\end{split}\]

Since there are at most $n$ sets $S$ of $x$ consecutive items, taking the Union Bound over all such sets completes the proof.
\end{proof}

\section{Proof of Claim~\ref{claim: the paths}} \label{appdx: routing in yi}

The goal of this section is to prove Claim \ref{claim: the paths}.
We show an efficient algorithm to construct a set $\pset^1=\set{P_1,\ldots,P_{M'}}$ of spaced-out paths, that originate at the vertices of $X$ on $R$, and traverse the boxes $\hat \block_j$ in a snake-like fashion (see Figure \ref{fig: routing}).
We will ensure that for each $1 \leq i \leq M'$ and $1\leq j\leq N_1$, the intersection of the path $P_i$ with the box $\hat\block_j$ is the $i$th column of $\wset_j$, and that $P_i$ contains the vertex $x_i$.

Fix some index $1\leq j\leq N_1$, and consider the box $\hat \block_j$. Let $I'_j$ and $I''_j$ denote the top and the bottom boundaries of $\hat \block_j$, respectively. 
For each $1 \leq i \leq M'$, let $W_j^i$ denote the $i$th column of $\wset_j$, and let $x'(j,i)$ and $x''(j,i)$ denote the topmost and the bottommost vertices of $W_j^i$, respectively.

For convenience, we also denote  $I''_0 = R$, and, for each $1\leq i\leq M'$, $x''(0,i) = x_i$.
The following claim is central to our proof.

\begin{claim}\label{claim: sub-paths}
There is an efficient algorithm to construct, for each $1\leq j\leq N_1$, a set $\pset_j=\set{P_j^1,\ldots,P_j^{M'}}$ of paths, such that:

\begin{itemize}
\item For each $1\leq j\leq N_1$ and $1\leq i\leq M'$, path $P_j^i$ connects $x''(j-1,i)$ to $x'(j,i)$; 
\item The set $\bigcup_{j=1}^{N_1}\pset_j$ of paths is spaced-out; and
\item All paths in $\bigcup_{j=1}^{N_1}\pset_j$ are contained in $G^t$, and are internally disjoint from $R$.
\end{itemize}
\end{claim}

Notice that by combining the paths in $\bigcup_{j=1}^{N_1}\pset_j$ with the set $\bigcup_{j=1}^{N_1}\wset_j$ of columns of the boxes, we obtain the desired set $\pset^1$ of paths, completing the proof of Claim~\ref{claim: the paths}. In the remainder of this section we focus on proving Claim~\ref{claim: sub-paths}.

Recall that each box $\hat \block_j$ is separated by at least $2M$ columns from every other such box, and from the left and right boundaries of $\hat G$. It is also separated by at least $4M$ rows from the top boundary of $\hat G$ and from the row $R$.
We exploit this spacing in order to construct the paths, by utilizing a special structure called a \emph{snake}, that we define next.

Given a set $\lset$ of consecutive rows of $\hat G$ and a set $\wset$ of consecutive columns of $\hat G$, we denote by $\Y(\lset, \wset)$ the subgraph of $\hat G$ spanned by the rows in $\lset$ and the columns in $\wset$; we refer to such a graph as a \emph{corridor}.
Let $\Y=\Y(\lset,\wset)$ be any such corridor.
Let $L'$ and $L''$ be the top and the bottom row of $\lset$ respectively, and let $W'$ and $W''$ be the first and the last column of $\wset$ respectively.
The four paths $\Y\cap L',\Y\cap L'',\Y\cap W'$ and $\Y\cap W''$ are called the top, bottom, left and right boundaries of $\Y$ respectively, and their union is called the \emph{boundary} of $\Y$. The width of the corridor $\Y$ is $w(\Y)=\min\set{|\lset|,|\wset|}$.
We say that two corridors $\Y,\Y'$ are \emph{internally disjoint}, iff every vertex $v\in \Y\cap \Y'$ belongs to the boundaries of both corridors.
We say that two internally disjoint corridors $\Y,\Y'$ are \emph{neighbors} iff $\Y\cap \Y'\neq \emptyset$.
We are now ready to define snakes.

A snake $\yset$ of length $z$ is a sequence $(\Y_1,\Y_2,\ldots,\Y_{z})$ of $z$ corridors that are pairwise internally disjoint, such that for all $1\leq z',z'' \leq z$, $\Y_{z'}$ is a neighbor of $\Y_{z''}$ iff $|z'-z''|=1$. The width of the snake is defined to be the minimum of two quantities: (i) $\min_{1\leq z'<z}\set{|\Y_{z'}\cap \Y_{z'+1}|}$; and (ii) $\min_{1\leq z'\leq z}\set{w(\Y_{z'})}$.
Notice that, given a snake $\yset$, there is a unique simple cycle $\sigma(\yset)$ contained in $\bigcup_{z'=1}^z \Y_{z'}$, such that, if $D$ denotes the disc on the plane whose boundary is $\sigma(\yset)$, then every vertex of $\bigcup_{z'=1}^z \Y_{z'}$ lies in $D$, while every other vertex of $\hat G$ lies outside $D$.
We call $\sigma(\yset)$ the \emph{boundary} of $\yset$.
We say that a vertex $u$ belongs to a snake $\yset$, and denote $u \in \yset$, iff $u \in \bigcup_{z'=1}^z \Y_{z'}$.
 We use the following simple claim, whose proof can be found, e.g. in~\cite{NDPhardness2017}.

\begin{claim}\label{claim: routing in a snake}
Let $\yset=(\Y_1,\ldots,\Y_{z})$ be a snake of width $w$, and let $A$ and $A'$ be two sets of vertices with $|A|=|A'|\leq w-2$, such that the vertices of $A$ lie on a single boundary edge of $\Y_1$, and the vertices of $A'$ lie on a single boundary edge of $\Y_{z}$. There is an efficient algorithm, that, given the snake $\yset$, and the sets $A$ and $A'$ of vertices as above, computes a set $\qset$ of node-disjoint paths contained in $\yset$, that  connect every vertex of $A$ to a distinct vertex of $A'$. 
\end{claim}

 Let $L, L'$ be any pair of rows of $G$.
 Let $\qset$ be a set of node-disjoint paths connecting some set of vertices $B \subseteq L$ to $B' \subseteq L'$.
 We say that the paths in $\qset$ are \emph{order-preserving} iff the left-to-right ordering of their endpoints on $L$ is same as that of their endpoints on $L'$.

 \begin{corollary} \label{cor: spaced out snakes}
Let $\yset=(\Y_1,\ldots,\Y_{z})$ be a snake of width $w$, and let $B$ and $B'$ be two sets of $r \leq \floor{w/2}-1$ vertices each, such that the vertices of $B$ lie on the bottom boundary edge of $\Y_1$, the vertices of $B'$ lie on the top boundary edge of $\Y_{z}$ and for every pair $v,v'\in B\cup B'$ of vertices, $d_{\hat G}(v,v')\geq 2$.
There is an efficient algorithm, that, given the snake $\yset$, and the sets $B$ and $B'$ of vertices as above, computes a set $\hat \qset$ of spaced-out order-preserving paths contained in $\yset$.
 \end{corollary}

\begin{proof}
Let $B=\set{b_1, b_2, \ldots, b_r}$ and $B'=\set{b'_1, b'_2, \ldots, b'_r}$. Assume that the vertices in both sets are indexed according to their left-to-right ordering on their corresponding rows of the grid.
Since set $B$ does not contain a pair of neighboring vertices, we can augment it to a larger set $A$, by adding a vertex between every consecutive pair of vertices of $B$. In other words, we obtain $A=\set{a_1, \ldots, a_{2r-1}}$, such that for  all $1 \leq i \leq r$, $a_{2i-1} = b_i$, and the vertices of $A$ are indexed according to their left-to-right ordering on the bottom boundary of $\Y_1$. Similarly, we can augment the set $B'$ to a set $A' = \set{a'_1, \ldots, a'_{2r-1}}$ of vertices, such that for all $1 \leq i \leq r$, $a'_{2i-1} = b'_i$, and the vertices of $A'$ are indexed according to their left-to-right ordering on the top boundary of $\Y_z$.

We apply Claim \ref{claim: routing in a snake} to the sets $A,A'$ of vertices, obtaining a set $\qset$ of node-disjoint paths, that are contained in $\yset$, and connect every vertex of $A$ to a distinct vertex of $A'$.  For all $1\leq i\leq r$, let $Q_i$ be the path originating from $a_i$.  We claim that $\qset$ is an order-preserving set of paths.
Indeed, assume for contradiction that some path $Q_i \in \qset$ connects $a_i$ to $a'_{i'}$, for $i \neq i'$. Notice that the path $Q_i$ partitions the snake $\yset$ into two sub-graphs: one containing $(i-1)$ vertices of $A$ and $(i'-1)$ vertices of $A'$; and the other containing the remaining vertices of $A$ and $A'$ (excluding the endpoints of $Q_i$). Since $i \neq i'$, there must be a path $Q_{i''} \in \qset$ intersecting the path $Q_i$, a contradiction to the fact that $\qset$ is a  set of node-disjoint paths.

Similarly, it is easy to see that for all $1\leq i<r$, $d(Q_{2i-1},Q_{2i+1})\geq 2$. This is since the removal of the path $Q_{2i}$ partitions the snake $\yset$ into two disjoint sub-graphs, with path $Q_{2i-1}$ contained in one and path $Q_{2i+1}$ contained in the other.

Our final set of path is $\hat \qset = \set{Q_{2i-1} : 1 \leq i \leq r}$. From the above discussion, it is a spaced-out set of paths contained in $\yset$, and for each $1\leq i\leq r$, path $Q_{2i-1} \in \hat \qset$ connects $b_i$ to $b'_{i}$. 
\end{proof}

In order to complete the proof, we need the following easy observation.

\begin{observation}\label{obs: snakes}
There is an efficient algorithm that constructs, for each $1\leq j\leq N_1$, a snake $\yset_j$ of width at least $2M$ in $G^t$, such that all resulting snakes are mutually disjoint, and for each $1\leq j\leq N_1$:

\begin{itemize} 
\item the bottom boundary of the first corridor of $\yset_j$ contains $I''_{j-1}$; 
\item the top boundary of the last corridor of $\yset_j$ contains $I'_j$; and
\item all snakes are disjoint from $R$, except for $\yset_1$, that contains $I_0''\subseteq R$ as part of is boundary, and does not contain any other vertices of $R$.
\end{itemize}
\end{observation}

The construction of the snakes is immediate and exploits the ample space between the boxes $\hat \block_j$;  (see Figure \ref{fig: routing} for an illustration).
From Corollary \ref{cor: spaced out snakes}, for each $1 \leq j \leq N_1$, we obtain a set $\pset_j$ of spaced-out paths contained in $\yset_j$, such that for each $1 \leq i \leq M'$, there is a path $P^i_j \in \pset_j$ connecting $x''(j-1,i)$ to $x'(j,i)$.
For each $1 \leq i \leq M'$, let $P_i$ be the path obtained by concatenating the paths $\set{P^i_1, \wset^i_1, P^i_2, \ldots, P^i_{N_1}, \wset^i_{N_1}}$.
The final set of paths is $\pset^1 = \set{P_1, \ldots, P_{M'}}$.

\bibliography{NDP-in-grids-hardness}
 \bibliographystyle{alpha}
\end{document}